\pdfoutput=1
\documentclass[noreview, a4paper, USenglish]{lipics-v2021}

\usepackage{booktabs}   

\usepackage{float}

\usepackage[T1]{fontenc}
\usepackage[scaled=0.85]{beramono}

\EventEditors{Manu Sridharan and Anders M{\o}ller}
\EventNoEds{2}
\EventLongTitle{35th European Conference on Object-Oriented Programming (ECOOP 2021)}
\EventShortTitle{ECOOP 2021}
\EventAcronym{ECOOP}
\EventYear{2021}
\EventDate{July 12--16, 2021}
\EventLocation{Aarhus, Denmark (Virtual Conference)}
\EventLogo{}
\SeriesVolume{194}
\ArticleNo{20}
\nolinenumbers
\usepackage{amssymb}
\usepackage{framed}
\usepackage{xspace}
\usepackage{fancyvrb}
\usepackage{listings}
\usepackage{booktabs}
\usepackage{array}
\usepackage{diagbox}
\usepackage{comment}
\usepackage[inline]{enumitem}
\usepackage{pifont}

\usepackage{stmaryrd}
\usepackage{xcolor}
\usepackage{mathtools}
\usepackage{apptools}
\usepackage{chngcntr}
\usepackage{scalerel}
\usepackage{xspace}
\usepackage{mathpartir}
\usepackage{graphicx}
\usepackage{csvsimple}
\usepackage{pdfpages}
\usepackage{tikz}
\usepackage{anyfontsize}
\usepackage{listings}

\def\withcolor{}

\ifdefined\withcolor
  \definecolor{haskellblue}{rgb}{0.0, 0.0, 1.0}
  \definecolor{haskellstr}{rgb}{0.2, 0.2, 0.6}
  \definecolor{haskellred}{rgb}{1.0, 0.0, 0.0}
  \definecolor{gray_ulisses}{gray}{0.55}
  \definecolor{castanho_ulisses}{rgb}{0.71,0.33,0.14}
  \definecolor{preto_ulisses}{rgb}{0.41,0.20,0.04}
  \definecolor{green_ulises}{rgb}{0.2,0.75,0}
  \definecolor{haskelltypes}{rgb}{0.71,0.33,0.14}
  \definecolor{logiccolor}{rgb}{0.0, 0.0, 1.0}
\else
  \definecolor{haskellblue}{gray}{0.1}
  \definecolor{haskellstr}{gray}{0.1}
  \definecolor{haskellred}{gray}{0.1}
  \definecolor{gray_ulisses}{gray}{0.1}
  \definecolor{castanho_ulisses}{gray}{0.1}
  \definecolor{preto_ulisses}{gray}{0.1}
  \definecolor{green_ulisses}{gray}{0.1}
  \definecolor{haskelltypes}{gray}{0.1}
  \definecolor{logiccolor}{gray}{0,1}
\fi

\definecolor{subtleOpHighlight}{rgb}{0.4, 0.2, 0.0}

\definecolor{lcolor}{gray}{0.0}
\definecolor{lappcolor}{gray}{0.0}
\definecolor{lappascolor}{gray}{0.0}

\def\codesize{\small}
\def\incodesize{\small}

\newcommand\showfocus[1]{\color{purple}{\textbf{#1}}}
\newcommand\showlogic[1]{\color{logiccolor}{#1}}

\lstdefinelanguage{HaskellUlisses} {
  aboveskip=\smallskipamount,
  belowskip=\smallskipamount,
  basicstyle=\ttfamily\codesize,
  moredelim=[is][\showfocus]{\#}{\#},
  moredelim=[is][\showlogic]{!}{!},
  sensitive=true,
  backgroundcolor=\color{white},
  morecomment=[l][\ttfamily\itshape\codesize]{--},
  morecomment=[s][\ttfamily\itshape\codesize]{\{-}{-\}},
  morestring=[b]",
  stringstyle=\color{haskellstr},
  showstringspaces=false,
  numberstyle=\codesize,
  numberblanklines=true,
  showspaces=false,
  breaklines=true,
  showtabs=false,
  literate={
           {<!}{{{\color{lcolor}<!}}}2
           {`}{{{$^{\backprime}{}$}}}1
           {?}{{{\color{lcolor}?}}}1
           {<=}{{$\leq$}}1
           {/=}{{$\neq$}}1
           {bot}{{$\bot$}}1
           {top}{{$\top$}}1
           {theta}{{$\theta$}}1
           {gf}{{{\color{lappascolor}f}}}1
           {rmap}{{{\color{lappcolor}map}}}3
           {gmap}{{{\color{lappascolor}map}}}3
           {r.}{{{\color{lappcolor}.}}}1
           {g.}{{{\color{lappascolor}.}}}1
           {r++}{{{\showfocus{++}}}}2
           {g++}{{{\color{lappascolor}++}}}2
           {>>}{{{\color{haskellblue}>>}}}3
           {>>=}{{{\color{haskellblue}>>=}}}5
           {</>}{{{\color{haskellblue}</>}}}3
           {<*>}{{{\color{haskellblue}<*>}}}3
           {++}{{{\color{haskellblue}++}}}3
           {gfib}{{{\color{lappascolor}fib}}}3
           {rfib}{{{\color{lappcolor}fib}}}3
           {r++}{{{\color{lappcolor}++}}}2
           {env}{{$\Gamma$}}1
           {|-}{{$\vdash$}}1
           {<=!}{{{\color{lcolor}<=!}}}3
           {!=}{{$\neq$}}1
           {->}{{$\rightarrow$}}2
           {~>}{{$\imparrow$}}2
           {<-}{{$\leftarrow$}}2
           {dollar}{{$\texttt{\$}$}}1
           {Set_mem}{{$\in$}}1
           {Set_cup}{{$\cup$}}1
           {Set_cap}{{$\cap$}}1
           {Set_emp}{{$\emptyset$}}1
           {Set_sub}{{$\subseteq$}}1
           {<=>}{{$\Leftrightarrow$}}3
           {=>}{{$\Rightarrow$}}2
           {1->}{{$\rightarrow$}}1
           {1=>}{{$\Rightarrow$}}1
           {||-}{{$\vdash$}}1
           {|->}{{$\mapsto$}}2
           {<:}{{$\preceq$}}1
           {Inarritu}{Inarritu}8},
  emph=
  {[1]
    succ,incr,two,incrMany,three,id,get,set,return,grant,revoke,fresh,main,canRead,safeRead,
    test1,test2,
    gnt,rev,next,foo,bar,baz,pure,client,done,clynt,fin,
    FilePath,IOError,abs,acos,acosh,all,and,any,appendFile,approxRational,asTypeOf,asin,
    asinh,atan,atan2,atanh,basicIORun,break,catch,ceiling,chr,compare,concat,concatMap,
    cos,cosh,curry,cycle,decodeFloat,denominator,digitToInt,div,divMod,drop,
    dropWhile,either,elem,encodeFloat,enumFrom,enumFromThen,enumFromThenTo,enumFromTo,
    error,even,exponent,fail,mapMaybe,filter,flip,floatDigits,floatRadix,floatRange,floor,
    foldl,foldl1,foldr1,fromDouble,fromEnum,fromInt,fromInteger,fromIntegral,
    fromRational,fst,gcd,put,tick,tock,tocker,ticker,getChar,getContents,getLine,head,inRange,index,init,intToDigit,
    interact,ioError,isAlpha,isAlphaNum,isAscii,isControl,isDenormalized,isDigit,isHexDigit,
    isIEEE,isInfinite,isLower,isNaN,isNegativeZero,isOctDigit,isPrint,isSpace,isUpper,iterate,
    last,lcm,length,lex,lexDigits,lexLitChar,lines,log,logBase,lookup,mapM,mapM_,max,
    maxBound,posMax,negMax,maximum,maybe,min,minBound,minimum,mod,negate,notElem,null,numerator,odd,
    or,ord,pi,pred,primExitWith,print,product,putChar,putStr,putStrLn,quot,
    quotRem,range,rangeSize,read,readDec,readFile,readFloat,readHex,readIO,readInt,readList,readLitChar,
    readLn,readOct,readParen,readSigned,reads,readsPrec,realToFrac,recip,rem,repeat,replicate,return,
    reverse,round,scaleFloat,scanl,scanl1,scanr,scanr1,seq,sequence,sequence_,show,showChar,showInt,
    showList,showLitChar,showParen,showSigned,showString,shows,showsPrec,significand,signum,sin,
    sinh,snd,span,splitAt,sqrt,subtract,tail,take,takeWhile,tan,tanh,threadToIOResult,toEnum,
    toInteger,toLower,toRational,toUpper,truncate,uncurry,undefined,unlines,until,unwords,unzip,
    unzip3,userError,words,writeFile,zip,zip3,zipWith,zipWith3,listArray,doParse,empty,for,initTo,
        assert,compose,checkGE,maxEvens,empty,create,get,set,initialize,idVec,fastFib,fibMemo,
        ex1,ex2,ex3,inc,dec,isPos,positives,find,flatten, expand,exAll,
        ind,evenLen,lenAppend,exDistOr,allDistAnd,len,size,union,fromList,initUpto,trim,
        insertSort,decsort,qsort,reverse,append,upperCase, ifM, whileM, get, decrM, diff,
        project, select, sel, elts, keys, dkeys, dfun, addKey, pTrue, emptyRD, rFalse,
                dom, rng, isI, isD, isS, movie1, movie2,  toI, toS, toD, good_titles, runState, ret,
                update, getCtr, setCtr, ctr, rdCtr, wrCtr, ifTest, whileTest, posCtr, zeroCtr, decr, decCtr,
                pread , pwrite , plookup , pcontents, pcreateF , pcreateFP, pcreateD, active, caps, pset, eqP,
                write, contents, alloc, derivP, copyP, createDir, store, copyRec, copySpec,
                forM_, when, flookup, fread, createDir, pcreateFile, isFile, copyFrame, ?
  },
  emphstyle={[1]\color{haskellblue}},
  emph=
  {[2]Show,Eq,Iso,VerifiedOrd,Ord,Num,UpClosed,Comp,Wit,Witness,Inductive,Meet,Flip,TRUE,
      Peano,Nat,Pos,SInt,Neg,IntGE,Plus,List,PAnd, POr, POrL, POrR,
        Bool,Char,Double,Either,Float,IO,Integer,Int,Maybe,Up,Mono,Identity,
        Ordering,Rational,Ratio,ReadS,ShowS,File,Token,Tk,ST,String,Str,Word8,
        InPacket,Tree,Vec,NullTerm,IncrList,DecrList,
        UniqList,BST,MinHeap,MaxHeap,World,RIO,IO,HIO,Post,Pre, OptEq,
        Privilege, Prop, Chain, ChainTy, Range, Dict, RD, Dom, Set, P, Univ, Schema, MovieSchema, RT,
        TDom, TRange, MoviesTable, RTSubEqFlds, RTEqFlds, Disjoint, Union, Ret, Seq, Trans, Map,
        Pure, Then, Else, Exit, Inv, OneState, Priv, Path, FH, Stable,
      Monoid, VerifiedMonoid, VerifiedComMonoid, Plus_2_2_eq_4, Plus_2_2, Nat_up, Int_up,
      AppendNilId, AppendAssoc,MapFusion,
      Plus_comm, Par, Term,
      Formula, Assignment, Body, Accel, Real, Accel', RVar, VerifiedCommutativeMonoid, CommutativeMonoid
  },
  emphstyle={[2]\color{castanho_ulisses}},
  emph=
  {[3]
    case,class,newtype,data,deriving,do,else,if,unpack,import,in,infixl,infixr,instance,let,
    module,of,primitive,then,refinement,type,where,rforall,forall,bound,
    measure,reflect,predicate, instance, class,
    exists
  },
  emphstyle={[3]\color{preto_ulisses}\textbf},
  emph=
  {[4]
    quot,rem,div,mod,elem,notElem,seq
  },
  emphstyle={[4]\color{castanho_ulisses}\textbf},
  emph=
  {[5]
    EQ,GT,LT,Left,Right
  },
  emphstyle={[5]\color{preto_ulisses}\textbf},
  emph=
  {[6]
      axiomatize, measure, inline
  },
  emphstyle={[6]\color{lcolor}}
}

\lstnewenvironment{code}
{\lstset{language=HaskellUlisses}}
{}

\lstnewenvironment{mcode}
{\lstset{language=HaskellUlisses,columns=fullflexible,keepspaces,mathescape}}
{}

\lstnewenvironment{mcodef}
{\lstset{language=HaskellUlisses,columns=fullflexible,keepspaces,mathescape,frame=single}}
{}

\lstMakeShortInline[language=HaskellUlisses,mathescape,keepspaces,mathescape,basicstyle=\ttfamily\codesize,breakatwhitespace]@

\lstdefinelanguage{Pseudo} {
  basicstyle=\ttfamily\codesize,
  sensitive=true,
  mathescape=true,
  morecomment=[l][\color{gray_ulisses}\ttfamily\codesize]{--},
  morecomment=[s][\color{gray_ulisses}\ttfamily\codesize]{\{-}{-\}},
  morestring=[b]",
  showstringspaces=false,
  numberstyle=\codesize,
  numberblanklines=true,
  showspaces=false,
  breaklines=true,
  showtabs=false
}

\lstdefinelanguage{java} {
    keywordstyle=[1],
    keywordstyle=[2]\color{ForestGreen},
    keywordstyle=[3]\color{Bittersweet},
    keywordstyle=[4]\color{RoyalPurple},
    morekeywords={region,private,synchronized}
}

\ifdefined\withcolor
        \definecolor{typecol}{rgb}{0.0,0.5,0.0}
        \definecolor{funcol}{rgb}{0.0,0.1,0.9}
\else
        \definecolor{typecol}{gray}{0.0}
        \definecolor{funcol}{gray}{0.0}
\fi

\lstdefinelanguage{pseudo2}{
  language=Python,
  basicstyle=\ttfamily\normalsize,
  mathescape=true,
  morekeywords={type,def,do,let,unpack},
  emph={[1] \Expr,\Pred, HP, FP, Int},
  emphstyle={[1]\itshape\color{typecol}},
  emph={[2] \pbe,normalize},
  emphstyle={[2]\itshape\color{funcol}},
  emph={[3] repeat,until},
  emphstyle={[3]\textbf}
}

\lstnewenvironment{code2}{\lstset{language=pseudo2}}{}

\lstnewenvironment{fscode}
{\lstset{language=HaskellUlisses,basicstyle=\ttfamily\small,mathescape=true}}
{}

\lstnewenvironment{codebox}
{\lstset{language=HaskellUlisses,frame=tlbr,basicstyle=\ttfamily\codesize,mathescape=true}}
{}

\def \ha {\lstinline[language=HaskellUlisses,basicstyle=\ttfamily\incodesize,mathescape=true]}

\usepackage{semantic}
\usepackage{float}
\usepackage{commands}

\usepackage{rotating}
\usepackage{pdflscape}
\usepackage[shortcuts]{extdash}
\usepackage{multirow}
\usepackage{bbm}
\usepackage{caption}
\usepackage{subcaption}

\title{Refinements of Futures Past: Higher-Order Specification with Implicit Refinement Types (Extended Version)}
\titlerunning{Refinements of Futures Past (Extended Version)}

\makeatletter
\author{Anish Tondwalkar}{University of California San Diego, CA, USA}{atondwal@eng.ucsd.edu}{}{}
\author{Matthew Kolosick}{University of California San Diego, CA, USA}{mkolosick@eng.ucsd.edu}{}{}
\author{Ranjit Jhala}{University of California San Diego, CA, USA}{rjhala@cs.ucsd.edu}{}{}
\makeatother
\authorrunning{A. Tondwalkar, M. Kolosick, and R. Jhala}
\Copyright{Tondwalkar, Kolosick, and Jhala}

\hideLIPIcs
\begin{document}

\makeatletter
\maketitle
\makeatother

\begin{abstract}
Refinement types decorate types with assertions 
that enable automatic verification. Like assertions, 
refinements are limited to binders that are in scope, 
and hence, cannot express higher-order specifications. 
\emph{Ghost} variables circumvent this limitation 
but are prohibitively tedious to use as the programmer 
must divine and explicate their values at all call-sites.
We introduce \emph{Implicit Refinement Types}
which turn ghost variables into implicit pair and function types, 
in a way that lets the refinement typechecker 
automatically \emph{synthesize} their values at compile time.
Implicit Refinement Types further take
advantage of refinement type information, allowing them to be used as a
lightweight verification tool, rather than merely as a technique to
automate programming tasks.
We evaluate the utility of Implicit Refinement Types by 
showing how they enable the modular specification 
and automatic verification of various higher-order examples including 
stateful protocols, access control, and resource usage. 
\end{abstract}

\section{Introduction}
\label{sec:intro}

\emph{Refinement types} allow programmers to decorate types 
with statically checked assertions (``refinements'') that 
can be used for automatic verification over decidable 
theories,
with applications including checking
array bounds \cite{xi_eliminating_1998,rondon_liquid_2008}, totality
\cite{vazou_refinement_2014}, data structure invariants \cite{LiquidPLDI09},
cryptographic protocols \cite{GordonRefinement09,FournetCCS11}, and properties
of web applications~\cite{SwamyOAKLAND11}.

\mypara{Problem: Higher-Order Reasoning}
Unfortunately, refinements cannot express the \emph{higher-order} specifications
needed for higher-order imperative programs~\cite{vazou_bounded_2015}.
Consider the access-control API:
\begin{code}
  grant :: File -> IO ()                                    read :: File -> IO String
\end{code}
Here, \ha{grant} and \ha{read} represent a file 
access API that enforces access control policies:
to \ha{read} a given file, we must have been
\ha{grant}'d permission to that file.
Concretely, this means that calls to \ha{grant f} update
the state of the world to one in which \ha{f} has been added to the set of files we 
have permission to \ha{read}.

Next, consider the API for operating over a stream 
of tokens in a linear fashion:
\begin{code}
  init :: ST Token                                          next :: Token -> ST Token
\end{code}
Here, \ha{init} simply gives us the first 
token to start a stream.
To use \ha{next} to recieve a token we must 
pass it the previous token, which it will 
then invalidate.
Both of these APIs include functions where correct usage fundamentally depends
on the state of the world at their eventual call site, as they are composed using
higher-order combinators such as \ha{(>>=)} and \ha{(>>)}.

One can formalize these informal specifications 
by augmenting the source program with \emph{ghost variables} 
\cite{OwickiGries,unno_automating_2013, vazou_bounded_2015} 
that represent ``departed quantities'' that characterize 
the state of the world.
In our file access example, \ha{read f} would have 
to additionally take in the set of files we have 
been \ha{grant}ed access to, and \ha{f} must be 
a member, thus allowing us to compose the computations \ha{grant f >>  read f}.
In our stream example, \mbox{\ha{next t}} must both take 
in an additional mapping of tokens to their validity, 
in which \ha{t} must be valid, and then give us back
an updated mapping, reflecting that \ha{next t} is 
the next valid token, thus allowing us to sequence the
computations as \ha{next t >>=  next}, but not \ha{next t >>  next t}.
The key idea in both of these examples is that \ha{(>>)} (``sequence'') 
and \ha{(>>=)} (``bind'') must \emph{relate} the output of a higher-order 
argument (\ha{grant f} and \ha{next t}) to the requirements of their other 
argument (\ha{read f} and \ha{next}).
For these stateful combinators, ghost variables allow us to lift the state of the world to relate the output of one computation to the input of another.

Note that while we start with these stateful examples for familiarity, we are interested in the more general problem of reasoning with higher-order programs.
For instance, Handley et al.~\cite{liquidate} introduce a \ha{Tick} datatype
that tracks resource usage of a computation.
In this setting, one may want a higher-order constant-resource combinator such as
\begin{code}
  mapA :: (a -> Tick n b) -> xs:List a -> Tick (n * len xs) (List b)
\end{code}
Informally, \ha{mapA} maps a computation that uses a fixed amount of resources $n$ (say, time) over a list and guarantees that it will take a fixed amount of resources equal to $n$ times the length of the list.
We can use ghost variables to lift the parameter \ha{n} such that both the input function and the output computation can refer to it.

\mypara{Ghosts of Past and Future}
Notice that there is a key distinction between 
the ghost variables in our informal specifications.
The ghost variable for \ha{read} summarizes the 
history of accesses granted and therefore 
corresponds to a \emph{history variable}~\cite{OwickiGries} 
used to refine the past. 
It is determined \emph{externally} by what accesses 
have been granted by the time \ha{read} is called.
Similarly, \ha{mapA} uses the ghost variable 
\ha{n} to track costs that will have been incurred
once the computation finishes.
In contrast, \ha{next t} yields the next valid token, 
and so the ghost variable for \ha{next} captures the 
\emph{future} value when the computation is run.
This corresponds to a \emph{prophecy variable}~\cite{abadi_existence_1988} 
we can use to refine the future.
That is, \ha{next} makes an \emph{internal} choice about 
what the next token \emph{will be}.
While external choice can be encoded as an additional 
parameter to \ha{read}, internal choice can be encoded 
as adding an additional \emph{return} value to \ha{next}.
However, in both cases, these ghost variables require 
polluting code with specification-level details and 
break APIs that don't accept or produce ghost variables.

\mypara{Implicit Refinement Types}
In this paper, we introduce \emph{Implicit Refinement Types} (IRT), 
a feature that allows us to capture the above specifications 
with implicit ghost variables, while preserving the automatic 
verification properties of prior refinement type systems.
Our approach takes inspiration from \emph{implicit parameters}~\cite{lewis_implicit_2000,
oliveira_type_2010, oliveira_implicit_2012}, a popular language feature
that is the foundation of the theory of typeclasses~\cite{wadler_how_1989};
models objects~\cite{oliveira_type_2010};
lends flexibility to module systems~\cite{white_modular_2015};
and features in C\#~\cite{emir_variance_2006}, C++~\cite{heinlein_implicit_2006}, 
and Scala.
To capture the distinction between internal and external choice 
of ghost parameters, we introduce two different notions of implicits: 
implicit dependent functions (\ie implicit-$\Pi$) and 
implicit dependent pairs (\ie implicit-$\Sigma$).
Implicit functions capture the notion of \emph{external choice}, 
where the ghost parameter is determined by the caller of a function.
Dually, implicit pairs correspond to \emph{internal choice}, where 
the ghost parameter is determined by the implementation of the function.
To sum up, we make the following contributions:

\begin{itemize}

\item 
  A declarative semantics for \emphbf{implicit refinement types}
  (\secref{sec:lang}), and an inference algorithm that is sound with
  respect to the declarative semantics (\secref{sec:constraint-generation}).
  Crucially, our semantics preserve \emph{subtyping} information, allowing us
  to take advantage of refinement information in inferring implicit parameters.

\item A notion of \emphbf{implicit pair types} that, when added to implicit
  function types, allows us to express higher-order specifications
  without requiring any changes to code (\secref{sec:overview}).

\item
  We implement our system, in a tool dubbed \toolname, and
  \emphbf{evaluate} our prototype on a range of case studies 
  to demonstrate where automatic verification with 
  Implicit Refinement Types bears fruit (\secref{sec:eval}).
\end{itemize}

\section{Overview}\label{sec:overview}

We start with an example-driven overview that walks through specification and verification with
implicit functions and pairs on minimal examples of higher-order functions, and then scales these examples to verify the
access control, token stream, and resource accounting examples from \secref{sec:intro}.

\subsection{Implicit Function Types}
\label{sec:overview:implicits}

While plain refinement type systems can employ extra
parameters to capture ghost variables, they must do so explicitly, requiring
programmers to divine their values and implement bookkeeping.
To illustrate this, consider the following example of a
higher-order function \ha{foo} that accepts a function parameter:
\begin{code}
  foo :: (Bool -> Int) -> ()
  foo f = assert (f True == f False)
\end{code}
The static \ha{assert}ion requires that the function
passed to \ha{foo} be constant.
We could formalize this informal
specification by adding an extra 
ghost parameter to \ha{foo} to 
represent the singleton return 
value of its argument:
\begin{code}
  foo :: n:Int -> (Bool -> SInt n) -> ()
\end{code}
But now we must not only rewrite \ha{foo} to take this ``unused'' ghost
parameter (\ha{foo n f = assert (f True == f False)}) but must also manually
modify the call site to \ha{foo 1 (\z $\rightarrow$ 1)}.
The \emph{implicit function} type lets us handle this automatically. 
Instead of changing the definition of \ha{foo} and its uses, the 
programmer changes \ha{foo} to make \ha{n} an \emph{implicit parameter} by surrounding it with square brackets:
\begin{code}
  foo :: [n:Int] -> (Bool -> SInt n) -> ()
\end{code}

\mypara{Verifying programs with implicits:}
To verify \ha{foo} and its call sites we must check:
\begin{enumerate*}[label=(\arabic*)]
\item
  that the \ha{assert}ion in the body of \ha{foo} is valid

\item
  that at the call-site \ha{foo (\z $\rightarrow$ 1)}, the argument meets the precondition.
\end{enumerate*}
We do so via a bidirectional traversal 
that checks terms against types to 
generate and solve \emph{Existential Horn Constraints} (\ehc{}).
When checking \emph{inside} the body of \ha{foo}, 
the implicit parameter \ha{[n:Int]} behaves 
just like a standard explicit, or ``corporeal'', 
parameter one might see in a $\Pi$-binder.
That is, calls to \ha{foo} will implicitly 
``pass'' an argument for \ha{n}, so the sequel 
must hold for \emph{all} values of \ha{n}.
This constrains the explicit argument \ha{f} 
to always return \ha{n}, from which the 
\ha{assert}ion then follows directly.
Checking the \emph{call-site} is more interesting: the implicit parameter
\mbox{\ha{[n:Int]}} says the application is valid iff there \emph{exists} a
witness \ha{n} such that the remainder of the specification holds.
For \ha{foo (\z $\rightarrow$ 1)}, we require that \ha{n} 
equals the return value of \ha{(\z $\rightarrow$ 1)}, \ie \ha{n = 1}.
Let's see how to automatically find such witnesses.

\mypara{Step 1: Templates}
First, we generate \emph{templates} to represent the types 
of terms whose refinements must be inferred. 
These templates are the base unrefined types for those 
terms, refined with \emph{predicate variables} $\kappa$ 
that represent the unknown refinements~\cite{rondon_liquid_2008}.
For example, we represent the as-yet unknown type
of\, \ha{\z $\rightarrow$ 1} with the template $\trfun{z}{\tbool}{\reftpv{\vv}{\tint}{\kappa(\vv)}}$.

\mypara{Step 2: Existential Horn Constraints}
Next, we traverse the term \ha{foo (\z $\rightarrow$ 1)} in a 
bidirectional, syntax directed fashion (\secref{sec:constraint-generation}) 
to generate the \ehc{} in \figref{fig:foobar:ehc}.
\begin{itemize}
\item Constraint~(\ref{constOneBody}) comes from the 
      body of the $\lambda$-term \ha{\z $\rightarrow$ 1} and says 
      the type of the returned value \ha{1}, \ie 
      \(\reftpv\vv\tint{\vv=1}\), must be a subtype 
      of the output type \reftpv{\vv}{\tint}{\kappa(\vv)}.

\item Constraint~(\ref{fooApp}) comes from applying
      \ha{foo} to the $\lambda$-term: the type 
      of which must be a subtype of 
      \ha{foo}'s input. Since function outputs are covariant, 
      this means the output type \reftpv{\vv}{\tint}{\kappa(\vv)} 
      must be a subtype of \reftpv{\vv}{\tint}{\vv = n}.
\end{itemize}

\mypara{Step 3: NNF Constraints}
We say that the \ehc{} in \figref{fig:foobar:ehc} 
is satisfiable iff there exists a \emph{predicate}
for $\kappa$ that, when substituted 
into the constraint, yields a \emph{valid} 
first-order formula.
To determine satisfiability, we transform the 
\ehc{} into a new constraint shown in \figref{fig:foobar:hc} 
comprising
\begin{itemize}
\item 
  An \emph{\nnf{} Constraint}~\cite{bjorner_horn_2015}, 
  where we replace each existentially quantified 
  $n$ with a universally quantified version bounded 
  by a predicate variable $\pvar{n}$, as shown 
  in Constraint~(\ref{pix}), and
\item 
  An \emph{Inhabitation Constraint} that each $\pvarsym$
  is non-empty, as shown in Constraint~(\ref{sidecon}).
\end{itemize}
We write $\pvarsym$ instead of $\kappa$ here solely 
to differentiate between predicate variables from the
original \ehc{} and those produced in this translation. 
Intuitively, if every $n$ satisfying $\pi(n)$ satisfies 
the \nnf{} constraint \emph{and} $\pi$ is non-empty,
\ie there exists $n$ such that $\pi(n)$, then we can 
conclude there exists some $n$ that satisfies the 
original \ehc{}.

\mypara{Step 4: Solution} 
We require an assignment to the 
variables $\pvar{x_1}, \kappa$ 
that makes the constraint in 
\figref{fig:foobar:hc} valid.
We first compute the \emph{strongest} 
valid solution~\cite{cosman_local_2017} for $\kappa$,
which yields
$\kappa(v) \doteq v = 1$.
As we require $\pvar{x_1}$ to be inhabited (Constraint~(\ref{sidecon})), we
assign to $\pvar{x_1}$ the \emph{weakest} predicate (as detailed in
\secref{sec:solving}) that makes valid the \nnf{} Constraint~(\ref{pix}), \ie, we
assign: $\pvar{x_1}(n) \doteq n = 1$.
The above assignments for $\kappa$ and $\pvar{x_1}$ 
make the \nnf{} constraint in \figref{fig:foobar:hc} 
valid, thus verifying \ha{foo (\z $\rightarrow$ 1)}.

\begin{figure}[t]
  \begin{subfigure}[t]{0.47\textwidth}
    \begin{alignat}{2}
      & \wedge\ \csbind z{\ttrue} \Rightarrow \csbind\vv{\h\vv = 1} \Rightarrow \kappa(\h\vv)         && \label{constOneBody} \\
      & \wedge\ \h{\exists}{\h{n}}.~ \csbind{\vv}{\kappa(\h\vv)} \Rightarrow \h\vv = \h{n} && \label{fooApp}
    \end{alignat}
    \caption{Existential Horn Clause}
    \label{fig:foobar:ehc}
  \end{subfigure}
  \hspace{0.03\textwidth}
  \begin{subfigure}[t]{0.47\textwidth}
    \begin{alignat}{2}
      & \wedge\ \csbind z{\ttrue} \Rightarrow \csbind\vv{\h\vv = 1} \Rightarrow \kappa(\h\vv) && \notag \\
      & \wedge \csbind{n}{\pvaro{n}{\h{n}}} \Rightarrow \csbind{\vv}{\kappa(\h\vv)} \Rightarrow \h\vv = \h{n} && \label{pix}\tag{$\Pi$}   \\
      & \wedge\ \h{\exists}{\h{n}} .\ \pvaro{n}{\h{n}} && \label{sidecon}\tag{$\Sigma$}
    \end{alignat}
    \caption{Horn Clause and Side Condition}
    \label{fig:foobar:hc}
  \end{subfigure}
  \caption{Verifying the application \ha{foo (\\z -> 1)}}
  \label{fig:foobar}
\end{figure}

\subsection{Implicit Pair Types}\label{sec:overview:coimplicits}

Implicit function parameters \mbox{(\ha{[n:Int] ->} $\---$)} 
represent \emph{external} choice 
where the implicit's value is resolved 
at the call site.
But external choice alone is insufficient 
to capture every use of ghost variables.
Consider the following example of a function that \emph{returns} a function:
\begin{code}
  bar :: () -> (Bool -> Int)
  bar _ = (\z -> 1)
  let f = bar () in assert (f True == f False)
\end{code}
Unlike \ha{foo}, \ha{bar} instead \emph{returns} 
a function \ha{f}: we want to verify 
that the two calls to \ha{f} return 
the \emph{same} value.
This simple example models our token stream API: while in fact \ha{bar} always returns 1 we do not in general know the exact return value of a computation --- for instance the token returned by \ha{next} may be retrieved over the network.
Suppose, as for \ha{foo}, we use an implicit function argument to type \ha{bar}: 
\begin{code}
  bar :: [n:Int] -> () -> (Bool -> SInt n)
\end{code}
Unfortunately, with this type we cannot 
verify the body of \ha{bar}:
for \emph{any} external instantiation 
of \ha{n},  it requires that \ha{bar} will return a function 
that returns \ha{n}, which is not true!
Instead, \ha{bar} is making an 
\emph{internal} choice (specifically, that \ha{n = 1}).
This motivates our notion of \emph{implicit pair types}, 
(written \ha{[n:Int].} $\---$), which add a ghost parameter 
in the \emph{return} position of a function. This lets 
us specify
\begin{code}
  bar :: () -> [n:Int].(Bool -> SInt n)
\end{code}

To verify that the code is safe using this specification, 
we are once again left with two tasks:
\begin{enumerate*}[label=(\arabic*)]
\item
  check that the \ha{assert}ion at the use site of \ha{bar} is valid
\item
  check at the definition site that the body of \ha{bar} meets the specified postcondition.
\end{enumerate*}
The first task is the easier one, though it requires an additional step that was
not needed in verifying the body of \ha{foo}.
\emph{Externally}, the type \mbox{(\ha{[n:Int].} $\---$)} acts as an extra return value
that behaves exactly like a standard (``corporeal'') dependent pair and is assumed to exist in
the body of the \ha{let}.
In order to account for this, we use a type-directed elaboration to automatically insert the appropriate
``unpacking'', giving a name to \ha{n} at the use site:
\begin{code}
  unpack (n, f) = bar () in assert (f True == f False)
\end{code}
Verifying the assertion then follows the exact same logic as verifying the body
of \ha{foo}, as \ha{f} is constrained to always return a value equal to \ha{n}.

Now, we turn to the second task:
\emph{internally}, the type \mbox{(\ha{[n:Int].} $\---$)} 
states that there \emph{exists} some ``ghost'' 
value \ha{n} that makes the remaining specification 
valid.
Here, \ha{n} names the return value of the function 
returned by \ha{bar ()}.
When checking the body of \ha{bar}, we will need 
to find an instantiation of \ha{n} and implicitly 
``pair'' or ``pack'' this value with the type of 
the returned function \ha{(\z $\rightarrow$ 1)}.
Here, clearly the instantiation is \ha{n = 1}.
Note how the process of checking the definition 
site of \ha{bar} mirrors that of checking the 
use site \ha{foo (\z $\rightarrow$ 1)}.
In fact, checking the body of \ha{bar} will 
generate the exact same constraints shown in 
\figref{fig:foobar}, which will be solved by 
the same process, producing the same solution, 
which suffices to verify that \ha{bar}
implements its specification.

\subsection{State}
\label{sec:overview:state}

\begin{figure}[t]
  \begin{code}
-- | A Hoare-style State Transformer --------------------------------------
data HST p    q     s a = State (s -> (s, a))
type SST w_in w_out s a = HST {w|w = w_in} {w|w = w_out} s a 

-- | Read and write the state ---------------------------------------------
get :: [w:s] -> SST w w s {v:s|v = w}
get = State (\s -> (s, s))

set :: w:s -> HST s {v:s|v = w} s ()
set w = State (\_ -> (w, ()))
 
-- | Monadic Interface for HST --------------------------------------------
pure  :: [w:s] -> x:a -> SST w w s a 
>>= :: [w1 w2 w3] -> SST w1 w2 s a -> (a -> SST w2 w3 s b) -> SST w1 w3 s b

-- | Client: Computing a "fresh" Int  -------------------------------------
fresh :: [n:Int] -> SST n (n + 1) Int n 
fresh = do { n <- get; set (n + 1); pure n }
\end{code}
  \caption{Typing Stateful Computations using Implicits}
  \label{fig:state}
\end{figure}

Next, we show how Implicit Refinement Types allow us 
to develop a new way of typing stateful 
computations, represented as higher-order 
state transformers. 
The key challenge here is to devise 
specification mechanisms that can \emph{relate} 
the state \emph{after} the transformation with 
the state \emph{before}. For example, to say that 
some function \emph{increments} a counter, we 
need a way to say that the value of the counter 
after the transformation is one greater than the value 
before. 
Previous methods do this either by typing the computation with \emph{two-state}
predicates (as in YNot~\cite{nanevski_ynot_2008}) or with a \emph{predicate
transformer} that computes the value of the input state in terms of the
output~\cite{Swamy:2013}.

Implicit Refinement Types enable a new way to relate the 
input and output states while still ensuring that each atomic 
component of the specification simply refers to a single value.
We will see that implicits are the crucial ingredient,
allowing us to \emph{name} --- and hence, reason about --- the output 
state, much as they did in the simplified \ha{foo} and \ha{bar} functions.

\mypara{A Hoare-Style State Transformer Monad}
First, as shown in \figref{fig:state}, we define a 
state transformer monad indexed by the type of the state \ha{s} 
and the computation's result \mbox{\ha{a}.}
We call this a \emph{Hoare State Monad} as it is also indexed with two
\emph{phantom} parameters \ha{p} and \ha{q} which will be refinement types
describing the \emph{input} and \emph{output} states of the transformer.
For convenience, we also define the singleton 
version \ha{SST i o s a} where the 
\ha{pre}-condition and \ha{post}-conditions 
are singleton types that say that the input 
(resp. output) state is exactly \ha{i} 
(resp. \ha{o}).
We can write a combinator to \ha{get} the 
``current state'', represented by the implicit 
parameter \ha{w} that is both the input and 
output state, and also used in the singleton 
type of the result of the computation.
Finally, we can write a combinator to \ha{set} the state to 
some new (explicitly passed) value \ha{w}, in which case, 
the input state can be any \ha{s}.

\mypara{A Monadic Interface using Implicit States}
Next, we develop a monadic interface for programming with \ha{HST}, by
implementing the \ha{pure} and \ha{>>=} combinators whose types use implicit
parameters to relate their input and output states.
The \ha{pure} combinator takes an implicit \ha{w} and returns a ``pure''
computation \ha{SST w w s a} whose result is the input \ha{x} and where the
state is \emph{unchanged}, \ie where the input and output states are both
\ha{w}.
The bind combinator \ha{(>>=)} combines transforms
from \ha{w1} to \ha{w2}, and
from \ha{w2} to \ha{w3} into a single transform
from \ha{w1} to \ha{w3}.
Here, one can think of \ha{w1} as a history variable summarizing the state of
the world before the transformed computation and \ha{w2} and \ha{w3} as prophecy
variables predicting what the respective sub-transformers will compute.
The implicit refinement specification then ensures that these ghost variables
all align appropriately.

\mypara{Verifying Clients}
We can use our interface to write a specification 
and implementation of the function \ha{fresh} that 
``increments'' a counter captured by the state parameter. 
This program \ha{get}s an integer state \ha{n}, and \ha{set}s it 
to \ha{n + 1}, and then returns \ha{n}. 
The \ha{do}-notation desugars into the monadic interface 
shown above in \figref{fig:state}.
The specification captures the fact that the ``counter'' 
is incremented by relating the input and output states via 
the implicit parameter \ha{n}.
Notice that the implicit parameters are doubly crucial:
first, they let us relate the input- and output-states, and 
second, they make programming pleasant by not requiring the 
programmer tediously spell out the intermediate states.

\subsection{Access Control} \label{subsec:examples:grant}

\begin{figure}[t!]
\begin{code}
-- | Access control policy State Transformer -----------------------
type AC p1 p2 a = SST p1 p2 (Set String) a

-- | Grant or revoke access to a file path -------------------------
grant :: [p:Set String] -> f:String -> AC p (p Set_cup single f) ()
grant f = State (\p -> (insert f p, ()))

revoke :: [p:Set String] -> f:String -> AC p {v|v = p - single f} ()
revoke f = State (\p -> (delete f p, ()))

read :: [p:Set String] -> {f:String|f Set_mem p} -> AC p p String
read f = State (\p -> (p, "file contents"))

-- | How one might safely read a file ------------------------------
main = runST {} (do grant "f.txt"; read "f.txt")

-- | Enabling dynamic access control policies ----------------------
canRead :: [p:Set String] -> f:String -> AC p p {v:Bool|v = f Set_mem p}
canRead f = State (\p -> (p, member f p))

safeRead :: [p:Set String] -> f:String -> AC p p (Maybe String)
safeRead f = do r <- canRead; if r then Just <$> read f else pure Nothing
\end{code}
\caption{Verifying Access Control Policies}
\label{fig:revoke}
\end{figure}

The refined Hoare State Transformer lets us 
specify and verify the access control and 
token stream examples from \secref{sec:intro}.
In \figref{fig:revoke} we show how we can 
instantiate it with refinements over the 
theory of sets to derive 
a stateful API representing the verified 
access control primitives of \ha{grant} 
and \ha{read}.
We first define \ha{AC} as a specialization 
of the \ha{SST} monad where we track file 
access permissions as a set of filenames 
both at the runtime level and---using implicits to relate the input and 
output states---at the type level.

\mypara{An API for safe file access}
We use the \ha{AC} monad to develop an API 
that statically enforces compliance with 
an access control list (\textsc{acl}).
The \ha{grant} primitive adds a file name 
to the access list, and \ha{read} statically 
checks that the file is in the \textsc{acl} and---for
simplicity---just returns the string \mbox{\ha{"file contents"}.}
In conjunction with the implementations 
of \ha{pure} and \ha{>>=} from \figref{fig:state}, 
the combinators can be used to verify \ha{main} which, 
running with an initial empty access control list, 
grants permission to read a file then reads the file.
If we accidentally tried to \ha{read} from, 
say, \ha{"secret-password.txt"}, the type 
checker would reject the program as unsafe.

\mypara{Dynamic policies}
Systems enforcing access control policies in 
practice~\cite{ChongOSDI14} are not necessarily 
limited to a static policy---instead, they 
define a \emph{checked read}, which checks 
at runtime whether a file is in the access 
control list, and then reads that file, 
returning a failure result if the permission 
check fails.
We enable this via \ha{canRead}, which 
determines if a file is in the dynamic \textsc{acl}.
This is reflected at the type level by 
passing in an implicit access control 
list \ha{p} and specifying that 
\ha{canRead} returns true iff 
its argument is in the \textsc{acl} \ha{p}.
We then use \ha{canRead} to define 
\ha{safeRead}, which can be called 
with no conditions on the \textsc{acl}.
Instead, it uses \ha{canRead} to 
dynamically check permissions, 
returning \ha{Nothing} on failure.
Crucially, to verify \ha{main} 
and \ha{safeRead} we do not need 
to tediously instantiate ghost 
variables, as Implicit Refinement Types 
automatically infer the suitable 
instantiations.

\subsection{Token Stream} \label{subsec:examples:tokens}

\begin{figure}[t!]
\begin{center}
\begin{code}
-- | Token State Transformer ---------------------------------------------
type TokM m1 m2 a = SST m1 m2 Tk a

-- | Start a stream, Get the next token unless the stream is done --------
done  :: Tk
init :: [t:Tk]. TokM {} (store {} t top) {v|v = t}
next :: [m:Map Tk Bool] -> t:{Tk | (select m t) /\ (not (t = done))}
        -> ([t':Tk]. TokM m (store (store m t bot) t' top) t')

-- | Looping through all tokens ------------------------------------------
client :: [m:Map Tk Bool] -> t:{select m t} -> TokM m {m'|select m' done} ()
client t = if t == done then pure () else do {t' <- next t; client t'}

-- | Starting the stream and running client ------------------------------
main :: TokM {} {m|select m done} ()
main = do {t <- init; client t}
\end{code}
\end{center}
\caption{Verifying a Token Stream API}
\label{fig:tokens}
\end{figure}

The Hoare State Transformer can be used 
to verify the token stream example of 
\secref{sec:intro}.
In \figref{fig:tokens} we show how to instantiate 
it with refinements over maps that track the 
validity of tokens.
First, we define a specialization of the \ha{SST} monad named \ha{TokM} whose
concrete state is a token (of type \ha{Tk}) that tracks the last token we sent.
\ha{TokM}'s ghost state is a map of the status 
of \emph{every} token.
That is, \ha{select m t} represents the proposition 
that the token \ha{t} is \emph{valid} (the next value 
to pass to the API): if \ha{select m t = top} 
(shorthand for \ha{true}), then \ha{t} is valid.
Otherwise, \ha{select m t = bot} 
(shorthand for \ha{false}) means 
that the token \ha{t} has been 
used and is now invalid.

\mypara{An API for streaming tokens}
The \ha{TokM} monad lets us specify and verify the 
token streaming API from \sectionref{sec:intro}.
First, we specify a special token \ha{done} 
that represents the last token of the stream.
On the other hand, \ha{init} represents a 
\emph{computation} that begins the stream 
using an implicit pair to capture that there 
is \emph{some} valid token \ha{t} resulting
from the computation of \ha{init}.
Moreover, \ha{init} starts with the empty map 
(with no stale tokens) to ensure that we may 
only begin the stream once.

The workhorse of this API is \ha{next}.
It first takes an implicit history parameter 
\ha{m} representing all of the past tokens.
We then check that the token \ha{t} passed 
in is not the \ha{done} token, and use \ha{m} 
to constrain \ha{t} to be valid (\ha{select m t}).
Finally, another implicit pair is used to 
produce the \emph{prophecy} variable \ha{t'}
which is both the next value returned by the 
computation, and the next valid token as noted 
by the ghost state, which also marks the old 
token \ha{t} invalid.

We can now develop the \mbox{\ha{client},} 
which recursively consumes the
remaining stream of tokens.
Were we to attempt to reuse the token \ha{t} 
in the recursive call the program will be 
correctly rejected as unsafe.
\ha{main} kicks off the stream with
\ha{init} and consumes it using \mbox{\ha{client}.}
Thus, implicit refinements eliminate 
the tedium of manually instantiating ghosts.

\begin{figure}[t!]
\begin{code}
-- | Singletons as resource counts ---------------------------------------
data Tick t a = Tick a
type T t a = Tick {v:Int | v = t} a 

-- | The Applicative Functor API -----------------------------------------
<*> :: [n:Int m:Int] -> T n (a -> b) -> T m a -> T (n + m) b
</> :: [n:Int m:Int] -> T n (a -> b) -> T m a -> T (n + m + 1) b
pure :: x -> T 0 a

-- | Appending two lists in a linear number of steps ---------------------
++ :: xs:(List a) -> ys:(List a) -> T (len xs) {v|len v = len xs + len ys}
[]      ++ ys = pure ys
(x:xs') ++ ys = pure (x:) </> (xs'++ys)

-- | Mapping a costly function over a list -------------------------------
mapA :: [n] -> (a -> T n b) -> xs:List a -> T (n * len xs) (List b)
mapA f []     = pure []
mapA f (x:xs) = pure (:) <*> f x <*> mapA f xs
\end{code}
\caption{Intrinsic Verification of Resource Usage}
\label{fig:tick}
\end{figure}

\subsection{Intrinsic Verification of Resource Usage}
\label{subsec:tick}
Next, we demonstrate how implicit refinement types can be used for specifying higher-order programs beyond the state monad: in particular, tracking resource usage.
\figref{fig:tick} defines an applicative functor for counting resource usage in the same style as Handley et al.~\cite{liquidate}.
This API has both the standard application operator \ha{<*>} and a resource-consuming application operator \ha{</>}.
The API counts the number of times we use the \ha{</>} operator, which allows us to apply a function \ha{f} with cost \ha{n} to an argument \ha{x} of cost \ha{m}, incurring a total cost of \ha{n + m + 1}.

Handley et al.~\cite{liquidate} show these combinators can be used to verify properties about
resource usage.
We adapt their example of counting the recursive steps in \ha{xs ++ ys}.
At each recursive step, we append another element \ha{x} to the beginning of the list using \ha{pure (x:) </>}.
Ultimately \ha{(++)} will use \ha{len xs} applications of this operator to build this list, verifying that \ha{(++)} is linear in the first argument.

Using this API we can further define the higher-order constant-resource combinator \ha{mapA}, which allows us to map a function \ha{f} of constant cost \ha{n} over a list \ha{xs}, and automatically verify that doing so costs \ha{n * len xs}.
This is easy to specify with implicit refinement types: the implicit argument
lifts the output cost of the function argument so that we may relate it with the
overall cost of calling \ha{mapA}.

\mypara{Two-State Specifications}
It is worth pausing here to recall the standard technique for specifying effectful programs like the ones we have shown: the two-state specifications that allow expressing the input and output requirements of a particular computation.
For instance, \ha{fresh} (\figref{fig:state}) would be given a specification such as \ha{requires (\s $\rightarrow$ top)} \ha{ensures (\s o s' $\rightarrow$ s' = s + 1 $\wedge$ o = s)} where \ha{s} and \ha{s'} represent the input and output state and \ha{o} represents the output of the computation.
Notably, even if our language included such two-state specifications, specifying \ha{mapA} would still require an extra parameter, as the relationship between the start and end ``states'' of \ha{mapA} depends on the relationship between the start and end ``states'' of the \ha{(a $\rightarrow$ T n b)} argument.

On the other hand, implicit refinements scale from capturing relations between the inputs and
outputs of a single computation to relating a higher-order computation to its
function argument(s) without having to ``hardwire'' some notion of two-states or
pre/post conditions.
Instead, they allow us to \emph{name} the input and output worlds
and \emph{lift them to the top level}, which allows assertions (refinements) that span
those states/worlds, including in examples such as \ha{mapA}.
The key contribution of Implicit Refinement Types then is that they work both
for classic two-state specifications \emph{and other} use-cases where two-state
specifications prove cumbersome and allow the necessary extra parameter of
functions like \ha{mapA} to be instantiated automatically.

\section{Programs}\label{sec:lang}

\begin{figure}[t!]
  \begin{center}
    \small
    \begin{tabular}{>{$}r<{$} >{$}c<{$} >{$}l<{$}}
      \multicolumn{3}{l}{\textbf{Types}} \\
      b & \bnfdef & \tint \spmid \tbool \spmid \cdots \\
      t & \bnfdef & \reftpv{x}{b}{r} \spmid \trfun{x}{t}{t} \spmid \trifun{x}{t}{t} \spmid \tcoifun{x}{t}{t} \\

      \multicolumn{3}{l}{\textbf{Terms}} \\
      c & \bnfdef & \tfalse, \ttrue \spmid 0, 1, \ldots \spmid \wedge, \vee, +, -, =, \leq, \ldots \\

      e & \bnfdef & c \spmid x \spmid \elambda{x}{t}{e} \spmid \app{e}{e} \spmid \eletin{\tb{x}{t}}{e}{e} \spmid \ilambda{x}{t}{e} \spmid \eunpack{x}{y}{e}{e} \\

      \multicolumn{3}{l}{\textbf{Contexts}} \\
      \Gamma & \bnfdef & \bullet \spmid \Gamma, \fbd{x}{t} \spmid \Gamma, \ebd{x}{t}
    \end{tabular}
  \end{center}

  \small
  \judgementHead{Type Checking}{\hastype{\Gamma}{e}{t}}
  \begin{mathpar}
    \mprset{sep=1em}
    \inferrule[\tAbsi]
      {\hastype{\Gamma, \ebd{x}{t_x}}{e}{t}}
      {\hastype{\Gamma}{\ilambda{x}{t_x}{e}}{(\trifun{x}{t_x}{t})}}

    \inferrule[\tApp]
      {\hastype{\Gamma}{e_1}{t}
      \\ \appjudge{\Gamma}{t}{e_2}{t'}}
      {\hastype{\Gamma}{\app{e_1}{e_2}}{t'}}

    \inferrule[\tVar]
      {\fbd{x}{t} \in \Gamma}
      {\hastype{\Gamma}{x}{t}} 

    \inferrule[\tLetBase]
      {\hastype{\Gamma}{e_x}{t_x}
      \\ \hastype{\Gamma,\fbd{x}{t_x}}{e}{t}
      \\\\ \isWellFormed{\Gamma}{t}
      \\ \isWellFormed{\Gamma}{t_x}
      \\ \forgetreft{t_x} = \tau}
      {\hastype{\Gamma}{\eletin{\tb{x}{\tau}}{e_x}{e}}{t}}

    \inferrule[\tUnpack]
      {\hastype{\Gamma}{e_1}{\tcoifun{x'}{t_1}{t_2}}
      \\\\ \hastype{\Gamma,\ebd{x}{t_1},\fbd{y}{\subst{t_2}{x'}{x}}}{e_2}{t}
      \\ \isWellFormed{\Gamma}{t}}
      {\hastype{\Gamma}{\eunpack{x}{y}{e_1}{e_2}}{t}}
  \end{mathpar}

  \judgementHead{Application Checking}{$\appjudge{\Gamma}{t_1}{e}{t_2}$}
  \begin{mathpar}
    \mprset{sep=1em}
    \inferrule[\tAppIFun]
      {\appjudge{\Gamma}{\subst{t}{x}{e'}}{e}{t'}
      \\ \hastype{\forgetimplicits{\Gamma}}{e'}{t_x}}
      {\appjudge{\Gamma}{\trifun{x}{t_x}{t}}{e}{t'}}

    \inferrule[\tAppFun]
      {\hastype{\Gamma}{e}{t_e}
      \\ \isSubType{\Gamma}{t_e}{t_x}
      \\\\ \isSubType{\Gamma, \fbd{y}{t_e}}{\subst{t}{x}{y}}{t'}
      \\ \isWellFormed{\Gamma}{t'}
      \\ y \ \text{fresh}}
      {\appjudge{\Gamma}{\trfun{x}{t_x}{t}}{e}{t'}}
  \end{mathpar}

  \judgementHead{Subtyping}{\isSubType{\Gamma}{t_1}{t_2}}
  \begin{mathpar}
    \mprset{sep=1em}

    \inferrule[\tSubBase]
      {\embed{\Gamma}(\forall v_1:b.r_1 \h{\Rightarrow} \subst{r_2}{v_2}{v_1}) \text{ is valid}}
      {\isSubType{\Gamma}{\reftpv{v_1}{b}{r_1}}{\reftpv{v_2}{b}{r_2}}}

    \inferrule[\tSubSigmaR]
      {\isSubType{\Gamma}{t_1}{\subst{t_2'}{x}{e}}
      \\ \hastype{\forgetimplicits{\Gamma}}{e}{t_2}}
      {\isSubType{\Gamma}{t_1}{\tcoifun{x}{t_2}{t_2'}}}
  \end{mathpar}
\caption{{Syntax and static semantics of \lang}}
\label{fig:lang}
\end{figure}

We start with a declarative 
static semantics for our elaborated core language $\lang{}$. 
Our discussion here omits polymorphism as it is orthogonal to adding implicit
types. (The full system can be found in \appref{sec:rules:lang}).

\subsection{Syntax}\label{sec:lang:syntax}

\figref{fig:lang} presents the syntax of our source language---a
lambda calculus with refinement types, extended with implicit function and dependent
pair types.

\mypara{Types} of $\lang{}$ begin with base types $\tint$ and $\tbool$, which
are \emph{refined} with a (boolean-valued) expression $r$ to form refined
base types $\reftpv{x}{b}{r}$.
Next, $\lang{}$ has dependent function types $\trfun{x}{t_1}{t_2}$.
Dependent function types are complemented by \emph{implicit} dependent function types
$\trifun{x}{t_1}{t_2}$, which are similar,
except that the parameter $x$ is passed \emph{implicitly}, and does not occur at runtime.
Dually, we have \emph{implicit} dependent pairs $\tcoifun{x}{t_1}{t_2}$, which
represent a pair of values: The first, named $x$, of type $t_1$, is
implicit (automatically determined) and does not occur at runtime.
Meanwhile, the second is of type $t_2$ which may refer to $x$.
We use $\tau$ to denote \emph{unrefined} types.

\mypara{Terms} of $\lang{}$ comprise
\emph{constants} (booleans, integers and primitive operations)
and \emph{expressions}.
Let binders are half-annotated with either a refinement type to be checked, or a
base type on which refinements are to be inferred.

In addition to explicit function abstraction, $\lang{}$ has the
implicit $\lambda$-former $\ilambda{x}{t}{e}$, where the parameter $x$ represents a
\emph{ghost} value that can only appear in refinement types.
Implicit functions are instantiated automatically, so there is no
syntax for eliminating them.
Similarly, implicit dependent pairs are introduced automatically,
and thus have no introduction form in $\lang{}$.
Instead, implicit dependent pairs have an \emph{elimination} form
$\eunpack{x}{y}{e_1}{e_2}$.
Here, if $e_1$ is of type $\tcoifun{x}{t_x}{t}$, then $x$ is bound at type $t_x$
and $y$ is bound at type $t$ in $e_2$.
Just like with implicit functions the $x$ represents a \emph{ghost} value that
may only appear in refinement types.

Though both implicit lambda and unpack terms are present in our model,
in practice we handle their insertion by elaboration: we discuss this
aspect of our implementation in \secref{sec:eval:implementation}.
In light of this, we develop the theory of Implicit Refinement Types in terms of
the fully elaborated expressions of the $\lang$ syntax.

\mypara{Contexts}
of $\lang{}$, written $\Gamma$,
comprise the usual ordered sequences of ``corporeal'' binders 
$\fbd{x}{t}$, where $x$ 
is visible in both terms and refinements, as well as 
\emph{ghost binders} $\ebd{x}{t}$, 
where $x$ is only visible in 
refinements.

\subsection{Static Semantics}\label{sec:lang:static}

\figref{fig:lang} provides an excerpt of the declarative typing 
rules for $\lang{}$ (the complete rules are in \appref{sec:rules:lang}).
Most of the rules are standard for refinement types~\cite{rondon_liquid_2008} so we focus our attention
on the novel rules regarding implicit types.

\mypara{Type Checking} judgments of the 
form $\hastype{\Gamma}{\expr}{\type}$ mean
``in context $\Gamma$, the term $\expr$ has 
type $\type$.''
The ghost binders in $\Gamma$, written $\ebd{x}{t}$, reflect
the ghostly, refinement-only nature of implicits.
This distinction is witnessed by the rule [\tVar{}] 
which types term-level variables using only the 
corporeal binders $\fbd{x}{t}$ in $\Gamma$.
This ensures that implicit variables are erasable: they can only appear in types
(specifications) and thus \emph{cannot} affect computation.

With the separation of ghostly (erasable) implicit binders from corporeal
(computationally relevant) binders, both the introduction rule for implicit
functions [\tAbsi{}] and the elimination rule for implicit pairs [\tUnpack{}]
are standard up to the ghostliness of binders.
Implicit pairs are eliminated through ``unpacking'' as is typical for dependent
pairs and existential types.
The rule [\tLetBase{}] demonstrates how we 
handle annotations at unrefined base types $\tau$: 
we pick some well-formed refinement type $t_x$
that erases to the base type $\forgetreft{t_x} = \tau$, 
and then use $t_x$ as the bound type of $x$.

Lastly, we split off type checking of applications 
into an application checking judgment,
in order to handle instantiations of implicit 
functions. 
In the rule [\tApp{}] we use this additional 
judgment to check that the argument $e_2$ 
is compatible with the input type of $e_1$.

\mypara{Application Checking} judgments of the form
{$\appjudge{\Gamma}{t}{e}{t'}$} mean
``when $e$ is the argument to a function of type $t$
the result has type $t'$''.
The (corporeal) application rule [\tAppFun{}], 
finds the type $t_e$ of $e$, checks
that this is consistent with the input type 
$t_x$ of the function, and then creates a new 
name $y$ to refer to $e$ so that we 
may substitute it into the return type.
However, $y$ is not bound in $\Gamma$ so we 
guess a type $t'$, well-formed under $\Gamma$, 
for the entire term, and check that the return 
type $\subst{t}{x}{y}$ is a subtype of $t'$.
The rule [\tAppIFun{}] checks implicit 
applications by \emph{guessing} an 
expression $e'$ at which to instantiate 
the implicit parameter. 
This $e$ is only used at the refinement-level 
and is thus allowed to range over both corporeal 
and ghost binders in $\Gamma$, as described by 
the antecedent $\hastype{\forgetimplicits{\Gamma}}{e}{t}$. 
($\forgetimplicits{\Gamma}$ replaces each ghost binder $\ebd{x}{t}$ in $\Gamma$
with a corresponding corporeal binder $\forgetimplicits{\ebd{x}{t}} =
\fbd{x}{t}$.)
The rule then continues along the spine of the application, further
instantiating implicit parameters as necessary until we can apply 
the rule [\tAppFun{}].

\mypara{Subtyping Judgments}
of the form ${\isSubType{{\Gamma}}{\type_1}{\type_2}}$ 
mean ``in context $\Gamma$, the values of $\type_1$ are 
a subset of the values of $\type_2$.''
Most of the rules are standard, with the base 
case [\tSubBase{}] reducing to a 
\emph{verification condition} (VC) 
that checks if one refinement 
\emph{implies} another.
[\tSubSigmaR{}] serves as the 
``introduction form'' for implicit pairs, 
and states that a term of type $t_1$ can 
be used as an implicit pair type if there 
is some expression $e$ of type $t_2$ such 
that $t_1$ is a subtype of $t_2'$ instantiated 
with that $e$.
Note that this is similar to explicitly constructing the dependent pair with $e$
as the first element.

\section{Logic}\label{sec:log}
\begin{figure}[t!]
  \begin{center}
  \begin{tabular}{r >{$}r<{$} >{$}r<{$} >{$}l<{$} >{$}l<{$}}
    \textit{Predicates} & r & \bnfdef & \text{\ldots varies \ldots} \\
    \textit{Types} & \tau & \bnfdef & \text{\ldots varies \ldots} \\
    \textit{Propositions} & p & \bnfdef & \applykvar{\kappa}{\overline{x}} \spmid r \\
    \textit{Existential Horn Clauses} & c & \bnfdef & \exists x : \tau. c \spmid c \wedge c \spmid \forall x: \tau . p \Rightarrow c \spmid p \\
    \textit{First Order Assignments} & \Psi & \bnfdef & \bullet \spmid \Psi, \fobind{r}{x}\\
    \textit{Second Order Assignments} & \Delta & \bnfdef & \bullet \spmid \Delta, \sobind{\lambda \overline{x}.r}{\kappa}
  \end{tabular}
  \end{center}
\caption{Syntax of \smtlang{}}
\label{fig:rules:smt}
\end{figure}

We define the syntax and semantics of \emph{verification conditions} (VCs) 
generated by rule [\tSubBase{}].
\figref{fig:rules:smt} summarizes the 
syntax of VCs, \emph{Existential Horn Clauses} (\textsc{ehc}), 
which extends the \textsc{nnf} Horn 
Clauses used in Cosman and Jhala~\cite{cosman_local_2017}
with \emph{existential} binders.
\emph{Predicates} $r$
range over a decidable
background theory.
\emph{Propositions} $\h{p}$
include predicates and second order 
\emph{predicate variables} $\applykvar{\kappa}{\overline{\h{x}}}$.
\emph{Clauses} $\h{c}$ comprise 
quantifiers, conjunction, and propositions.
In the syntax tree of clauses, there are two places a $\h{\kappa}$ variable may appear: as a leaf (\emph{head} position), or in an antecedent under a
universal quantifier $\h{\forall x:\tau.\applykvar{\kappa}{\overline{y}}
\Rightarrow c}$, (\emph{guard} position).

\emph{Dependencies} 
\label{subsec:deps}
of an \ehc{} constraint \h{c} are the set $E$ 
of pairs $(\kappa, \kappa')$ such that $\kappa$ appears 
in guard position in \h{c} and $\kappa'$ appears 
in head position under that guard.
When $(\kappa, \kappa') \in E$, we say that 
$\kappa$ \emph{depends on} $\kappa'$, or, dually, 
$\kappa'$ appears \emph{under} $\kappa$.
The \emph{cycles} in a constraint \h{c} are 
nonempty sets $S$ of predicate variables 
such that: for all $\kappa \in S$, there 
exists $\kappa' \in S$ such that $(\kappa,\kappa') \in E$.
A set of predicate variables $S$ is said 
to be \emph{acyclic} in \h{c}, if all cycles in \h{c} 
contain at least one predicate variable not in 
$S$.
A predicate variable ${\kappa}$ is said to be \emph{acyclic} 
in \h{c}, if no cycles in \h{c} contain $\kappa$.
A constraint $\h{c}$ is said to be \emph{acyclic} if 
there are no cycles in $\h{c}$.

\mypara{Semantics}
\label{subsec:logic-semantics}
[\tSubBase{}] checks the \emph{validity} of the formula obtained by interpreting
contexts and terms in our constraint logic, (formalized by $\embed{\cdot}$ in \appref{sec:rules:logic}).
Contexts $\Gamma$ yield a sequence
of universal quantifiers for all variables bound at 
(interpretable) basic types.
Recall that VCs do not contain predicate variables 
$\applykvar{\kappa}{\overline{\h{x}}}$.
The restricted grammar of the VCs
is designed to be amenable to SMT 
solvers, represented by an oracle $\smtvalid{\h{c}}$ that checks validity, defined at the end of this section.

We eliminate \emph{predicate variables} $\kvar$ via substitution,
$\applyso{\h{\Delta}}{\h{c}}$ (defined in \appref{sec:solving:complete}) that map them onto meta-level lambdas $\lambda \overline{\h{x}}.\h{p}$.
We represent solutions to
\emph{existential binders} with a substitution ($\Psi$) binding existential
variables to predicates.
We define an \emph{existential substitution}
$\fosubst{\h{c}}{\h{x}}{\h{r}}$ recursively over
$\h{c}$ which removes the corresponding existential binder, using standard
substitution to replace $\h{x}$ with its solution $\h{r}$:
$\fosubst{(\exists x:\tau.c)}{x}{r} \doteq \subst{c}{x}{r}$.

\emph{\ehc{} Validity} is then defined by the judgment
$\valid{\h{\Delta}}{\h{\Psi}}{\h{c}}$.
Intuitively, $\h{c}$ is valid under the substitutions 
$\h{\Delta}, \h{\Psi}$ if the result of applying the substitutions 
yields a VC that is valid.
We say that an \textsc{ehc} is \emph{satisfiable}, 
written $\vDash \h{c}$ if \emph{there exist} 
$\h{\Delta}$ and $\h{\Gamma}$ such that  
$\valid{\h{\Delta}}{\h{\Gamma}}{\h{c}}$.
We say that \h{c} and \h{c'} are \emph{equisatisfiable} 
when $\vDash \h{c}$ iff $\vDash \h{c'}$.

\section{Type Inference}
\label{sec:constraint-generation}

\begin{figure}[t]
  \small

  \judgementHead{Subtyping}{\subtyping{\Gamma}{t_1}{t_2}{c}}
  \begin{mathpar}
    \mprset{sep=0.6em}

    \inferrule
      {c = \forall x:\tau. \embed{r} \Rightarrow \embed{\subst{r'}{y}{x}}}
      {\subtyping{\Gamma}{\reftpv{x}{\tau}{r}}{\reftpv{y}{\tau}{r'}}{c}}

    \inferrule
      {\subtyping{\Gamma, \fbd{z}{t_2}}{t_1}{\subst{t_2'}{x}{z}}{c}
      \\ z \text{ fresh}}
      {\subtyping{\Gamma}{t_1}{\tcoifun{x}{t_2}{t_2'}}{\genexists{z}{t_2}{c}}}
  \end{mathpar}

  \judgementHead{Checking}{\checking{\Gamma}{e}{t}{c}}
  \begin{mathpar}
    \mprset{sep=0.6em}

    \inferrule[\cSub]
      {\synth{\Gamma}{e}{t'}{c}
      \\\\ \subtyping{\Gamma}{t'}{t}{c'}}
      {\checking{\Gamma}{e}{t}{c \wedge c'}}

    \inferrule[\cLetBase]
      {\checking{\Gamma}{e_1}{\hat{t}}{c_1}
      \\\\ \checking{\Gamma,\fbd{x}{\hat{t}}}{e_2}{t}{c_2}
      \\ \hat{t} = \fresh{\Gamma}{\tau}}
      {\checking{\Gamma}{\eletin{\tb{x}{\tau}}{e_1}{e_2}}{t}{c_1 \wedge (\genimp{x}{\hat{t}}{c_2})}}

    \inferrule
      {\synth{\Gamma}{e_1}{\tcoifun{x'}{t_1}{t_2}}{c_1}
      \\ \checking{\Gamma,\ebd{x}{t_1},\fbd{y}{t_2'}}{e_2}{t}{c_2}
      \\\\ t_2' = \subst{t_2}{x'}{x}
      \\ c = c_1 \wedge (\genimp{x}{t_1}{(\genimp{y}{t_2'}{c_2})})}
      {\checking{\Gamma}{\eunpack{x}{y}{e_1}{e_2}}{t}{c}}

    \inferrule
      {\synth{\Gamma}{e_1}{t'}{c_1}
      \\\\ \spinechecking{\Gamma}{t'}{e_2}{t}{c_2}}
      {\checking{\Gamma}{\app{e_1}{e_2}}{t}{c_1 \wedge c_2}}
  \end{mathpar}

  \judgementHead{Synthesis}{\synth{\Gamma}{e}{t}{c}}
  \begin{mathpar}
    \inferrule
    {\fbd{x}{t} \in \Gamma}
    {\synth{\Gamma}{x}{t}{\true{}}}

    \inferrule
    {\synth{\Gamma}{e_1}{t_1}{c_1}
      \\ \spinesynth{\Gamma}{t_1}{e_2}{t_2}{c_2}}
    {\synth{\Gamma}{\app{e_1}{e_2}}{t_2}{c_1 \wedge c_2}}
  \end{mathpar}

  \judgementHead{Application Checking}{\spinechecking{\Gamma}{t}{e}{t'}{c}}
  \begin{mathpar}
    \inferrule[\cAppFun]
    {\checking{\Gamma}{y}{t_x}{c_1}
      \\ \subtyping{\Gamma}{\subst{t}{x}{y}}{t'}{c_2}}
    {\spinechecking{\Gamma}{\trfun{x}{t_x}{t}}{y}{t'}{c_1 \wedge c_2}}

    \inferrule[\cAppIFun]
    {\spinechecking{\Gamma,\ebd{z}{t_x}}{\subst{t}{x}{z}}{y}{t'}{c}
      \\ z \text{ fresh}}
    {\spinechecking{\Gamma}{\trifun{x}{t_x}{t}}{e}{t'}{\genexists{z}{t_x}{c}}}
  \end{mathpar}
  \caption{Constraint Generation}
  \label{fig:consgen}
\end{figure}

The declarative semantics described in \figref{fig:lang} are
decidedly non-deterministic.
This is most evident in the rules for implicits, such as [\tAppIFun{}] and
[\tSubSigmaR{}], where, out of thin air, we generate an expression to instantiate
the implicit parameter.
Additional non-determinism appears in rules like \mbox{[\tUnpack{}],} where a
refinement type must be picked such that it is well-formed under the outer
context $\Gamma$.
This is required as the body of the unpack expression is checked under $\Gamma$
extended with the binders $x$ and $y$, but the type must be well-formed under
$\Gamma$ itself to ensure that these variables do not escape (since they may
appear in the type $t$).

We account for all of the non-determinism of the declarative semantics with an
algorithmic, bidirectional type inference system~\cite{pierce_local_2000,
dunfield_complete_2013}, excerpts of which are shown in \figref{fig:consgen}
(the full rules are in \appref{sec:rules:cgen}).
We split the declarative type checking judgments 
into two forms:
synthesis ($\synth{\Gamma}{e}{t}{c}$) and 
checking ($\checking{\Gamma}{e}{t}{c}$) along with corresponding 
application synthesis ($\spinesynth{\Gamma}{t}{y}{t'}{c}$) and 
application checking ($\spinechecking{\Gamma}{t}{y}{t'}{c}$) forms.
Synthesis forms produce the type as an output while checking forms take the
type as an input. 
We also introduce an algorithmic subtyping judgment
(\mbox{$\subtyping{\Gamma}{t_1}{t_2}{c}$}).
In addition to their other outputs, these judgments 
produce an \ehc{} $c$ as an output.
The core of our inference algorithm is precisely 
in extending the restricted grammar of verification 
conditions to an \textsc{ehc} that captures the 
constraints on the non-deterministic choices.
Inference then reduces to the satisfiability 
of the constraint $c$ (as checked in \secref{sec:solving}).

\subsection{Constraining Unknown Refinements}
\label{subsec:extending-predicates}

Consider the following \lang program from \secref{sec:overview:implicits}:
\[
  \textsc{example} = \eletin{\etb{y}{(\tfun{\tbool}{\tint})}}{\elambda{z}{\tbool}{1}}{\;\app{\text{foo}}{y}}.
\]
recalling that $\text{foo}$ has the type
$\trifun{n}{\tint}{\tfun{(\tfun{\tbool}{\reftpv{v}{\tint}{v = n}})}{\tunit}}$.
We wish to check that this program is safe by checking that it types with type
$\tunit$.

To type \textsc{example} in our declarative 
semantics, we first need to apply the rule 
[\tLetBase{}] which non-deterministically chooses 
a $t_x$, well-formed under $\Gamma$
($\isWellFormed{\Gamma}{t_x}$), 
such that $t_x$ is consistent with
the base type ($\forgetreft{t_x} = \tfun{\tbool}{\tint}$).
To capture picking this refinement type 
we employ predicate variables that represent 
unknown refinements, as is standard in the 
refinement type inference literature~\cite{Knowles07,rondon_liquid_2008}.
As we know the type we wish to give \textsc{example}, 
we will focus on the checking rule [\cLetBase{}].
We generate a fresh refinement type using $\freshsym$, which takes as input a
type $t$ and the current context $\Gamma$ and then produces a
refinement type, where each base type is refined by a \emph{fresh} predicate
variable $\applykvar{\kappa}{\overline{x}}$, where $\overline{x}$ are all of the
variables bound in the context, all of which can appear in refinements at this
location.
In our example, this would give the type
$\hat{t} = \trfun{z}{\reftpv{v}{\tbool}{\applykvar{\kappa_1}{v}}}{\reftpv{v}{\tint}{\applykvar{\kappa_2}{z, v}}}$.

We then check $\elambda{z}{\tbool}{1}$ 
at the refinement type $\hat{t}$. 
Now, we use the rule [\cSub{}] to synthesize 
a type for this term and then use subtyping 
to check that the synthesized type is subsumed 
by $\hat{t}$.
The synthesized type will be
$\trfun{z}{\reftpv{v}{\tbool}{\applykvar{\kappa_3}{v}}}{\reftpv{v}{\tint}{v = 1}}$
and the subtyping check will generate the following constraint which is
equivalent to \coderef{constOneBody} in \figref{fig:foobar:ehc} modulo the
administrative predicate variable $\kappa_3$ and a simplification of the
unconstrained $\kappa_1$:
\begin{alignat*}{2}
  & \wedge \csimps{v}{\applykvar{\kappa_1}{v}}{\applykvar{\kappa_3}{v}} \\
  & \wedge \csbind{z}{\applykvar{\kappa_1}{z}} \Rightarrow \csimps{v}{v = 1}{\applykvar{\kappa_2}{z, v}}
\end{alignat*}
There is a subtlety in how [\cLetBase{}] generates the quantified subformula
$\csbind{z}{\applykvar{\kappa_1}{z}} \Rightarrow \cdots$: this formula is
generated by the clause $\genimp{x}{\hat{t}}{c_2}$ where the double colon
represents a \emph{generalized implication} that drops any variables quantified
at non-base types (as only base types are interpreted into the refinement logic).
$$
  \genimp{x}{\reftpv{x}{b}{r}}{c} \doteq \forall \tb{x}{b}.r \Rightarrow c 
  \quad
  \quad
  \quad
  \genimp{x}{t}{c} \doteq c
$$

\subsection{Constraining Implicit Application}
\label{subsec:extending-existentials}

The predicate variables let us capture 
guessed refinement types as second order 
constraints.
Next we turn to checking the \emph{implicit} 
application $\app{\text{foo}}{y}$.
$\text{foo}$ has the type
$\trifun{n}{\tint}{\tfun{(\tfun{\tbool}{\reftpv{v}{\tint}{v = n}})}{\tunit}}$,
so the declarative semantics arbitrarily picks
an expression $e$ of type $\tint$ 
to instantiate $n$.
In the algorithmic type system, we capture 
the constraints on this choice with the 
existential quantifiers of our \ehc{}.
This is shown in the application checking rule [\cAppIFun{}] (the corresponding
application synthesis rule [\sAppIFun{}] appears in \appref{sec:rules:cgen}).

Recall that the judgment $\spinechecking{\Gamma}{t}{e}{t'}{c}$ 
says that we are checking an application of a term of type $t$ to an argument 
$e$ and require that the application has the type $t'$.
Here, we are checking $\app{\text{foo}}{y}$ 
against the type $\tunit$.
Our bidirectional rule [\cAppIFun{}] ``guesses'' 
the instantiation by generating a fresh 
variable $n$ and binding it at the 
type $\tint$.
This variable is added to the context as 
predicate variables may depend on it.
We then generate a constraint $\genexists{n}{\tint}{c}$ 
which says that our guessed $n$ must be consistent 
with $c$, the constraint generated by continuing 
to check down the abstract syntax tree (along the spine 
of the application of the revealed function type
${\subst{t}{x}{y}} = \tfun{(\tfun{\tint}{\reftpv{v}{\tint}{v = n}})}{\tunit}$).
$\genexists{n}{\tint}{c}$ is the existential counterpart to the generalized
implication $\genimp{n}{\tint}{c}$:
\[  
  \genexists{x}{\reftpv{x}{b}{r}}{c} \doteq \exists \tb{x}{b}.(r \wedge c) 
  \quad
  \quad
  \quad
  \genexists{x}{t}{c} \doteq c
\]
This is now a concrete function type and the 
rule [\cAppFun{}] will check that the argument $y$ 
has the type
$\tfun{\tint}{\reftpv{v}{\tint}{v = n}}$
and that the type $\tunit$ is a subtype of $\tunit$.
Checking the implicit and then concrete application 
thus generates the constraints
$\exists n.(\csimps{v}{\true}{\applykvar{\kappa_1}{v}} \wedge{} \csimps{v}{\applykvar{\kappa_2}{\_, v}}{v = n})$.

Combining these constraints with those generated 
from checking $\elambda{z}{\tbool}{1}$, we get 
constraints equivalent to those in \figref{fig:foobar:ehc}.
A satisfying assignment is:
\[
  \Delta = [\sobind{\lambda z, v. v = 1}{\kappa_1}, \sobind{\lambda v. \true}{\kappa_2}, \sobind{\lambda v. \true}{\kappa_3}]
  \quad
  \quad
  \quad
  \Psi = [\fobind{1}{v}]
\]

Satisfying solutions to the predicate 
variables and existential constraints give instantiations 
to the angelic choices of refinement types and 
implicit arguments respectively.
This gives us the following soundness theorem 
for our type inference algorithm (where the function
$\kvars{c}$ returns the set of predicate variables in $c$):

\begin{theorem}[Soundness of Type Inference]
  {~}
  If $\synth{\bullet}{e}{t}{c}$, $\valid{\Delta}{\Psi}{c}$, and $\kvars{t} \subseteq \domain{\Delta}$, then $\hastype{\bullet}{e}{\Delta(t)}$.
\end{theorem}

\section{Solving}
\label{sec:solving}

The constraints generated by algorithmic type checking 
have both predicate variables and alternating universal 
and existential quantifiers.
We must provide solutions to both predicate variables
and existential variables before we can use SMT solvers
to check the validity of a VC (\secref{subsec:logic-semantics}). 
We compute solutions in four steps:
\begin{enumerate}
\item We transform the \ehc{} by skolemization to replace existential
      variables with universally quantified predicate variables and 
      inhabitation side conditions.

\item We eliminate the original predicate variables.

\item We solve the skolem predicate variables.

\item Finally, we check the inhabitation side conditions.
\end{enumerate}

\mypara{Weakening and Strengthening}
A function $f$ on constraints is a \emph{strengthening} when $\forall \h{c}.~f(\h{c})\hrel{\Rightarrow} \h{c}$.
A function $f$ on constraints is a \emph{weakening} when $\forall \h{c}.~\h{c}\hrel{\Rightarrow} f(\h{c})$.
We prove that, if the transformations in steps 1 and 2 above are both weakening
and strengthening, our algorithm produces a verification condition that is
\emph{equisatisfiable} with the original constraint, \ie our algorithm is sound
and complete.

\mypara{Separable Constraints} 
An \ehc{} $\h{c}$ is \emph{separable} 
if it can be written as a conjunction 
$\h{c_1 \wedge c_2}$, where $\h{c_1}$ 
is an \nnf{} Horn clause and $\h{c_2}$ 
is an acyclic \ehc{}.
The following theorem exactly characterizes 
separable \ehc{}s:
\begin{theorem}\label{thm:separable}
  \h{c} is separable iff there are no 
  cyclic $\h{\kappa}$s under existential binders.
\end{theorem}
There are standard partial techniques 
for solving cyclic \nnf{} Horn clauses~\cite{bjorner_horn_2015, cosman_local_2017}
so the task of solving a separable \ehc{} 
can be split into applying one of these 
existing techniques and then solving the 
acyclic \ehc{}.
Consequently, all a programmer must do to make 
constraints separable is provide either a 
local solution to an implicit variable via 
an explicit value or a solution to a cyclic 
predicate variable (\eg by providing a type 
signature for a recursive function.)

Thus, in the sequel, we focus on the remaining 
problem: solving an acyclic \ehc{}. 
We use the acyclic \ehc{} from \figref{fig:foobar:ehc} 
(reproduced below) as a running example.
Recall that this is a simplified version 
of the constraints generated during type 
inference on the program \textsc{example} in
\secref{subsec:extending-predicates} and 
\secref{subsec:extending-existentials}.
\begin{align}
  & \wedge\ \csbind z{\true} \Rightarrow \csbind\vv{\h\vv = 1} \Rightarrow \kappa(\h\vv) \\
  & \wedge\ \h{\exists}{\h{n}}. \csbind{\vv}{\kappa(\h\vv)} \Rightarrow \h\vv = \h{n}\label{ehc:existential-binder}
\end{align}

\mypara{Step 1: Skolemization}
We use the function \pokesym{} to transform the \ehc{} $\h{c}$
to a conjunction of an \nnf{} Horn Clause
$\noside{\poke{\emptyset}{\h{c}}}$ and side conditions $\sides{\poke{\emptyset}{\h{c}}}$.
This differs from textbook skolemization in two important ways:
We replace each existential quantifier $\exists n. c$ with Skolem predicates
$\h{\forall n . \pi_n(n,\overline{x}) \Rightarrow c}$ rather than Skolem
functions, so that we can synthesize a (declarative) relation
rather than a function.
As a result, we must still check to make sure that this relation is inhabited.
We do so by producing the side condition $\h{\exists n . \pi_n(n,\overline{x})}$.
Our transformation Skolemizes the existential binding \coderef{ehc:existential-binder} of our example as follows:
\begin{align}
  & \wedge\ \csbind z{\true} \Rightarrow \csbind\vv{\h\vv = 1} \Rightarrow \kappa(\h\vv) \label{ehc:head-kappa}\\
  & \wedge\ \forall n. \pi(n) \Rightarrow \csbind{\vv}{\kappa(\h\vv)} \Rightarrow \h\vv = \h{n} \\
  & \wedge\ \exists n. \pi(n)
\end{align}

This transformation will be crucial later: giving a name to $\pi$ allows us to
separate the inhabitation and sufficiency constraints on $n$.

\pokesym{} yields an \nnf{} that has two classes of predicate variables: Skolem
predicates corresponding to existential binders (written $\h{\pi_n}$) that have
an inhabitation side condition, and predicate variables corresponding to unknown
refinements (written $\kappa$).
The $\h{\pi_n}$ only appear negatively, so the standard technique of finding the
least fixed point solution~\cite{rondon_liquid_2008,cosman_local_2017} would
simply return $\h{\false{}}$, which will fail the inhabitation
side conditions.
Instead, we would like to compute the 
\emph{greatest fixed point} solution 
for each $\h{\pi_n}$, but, for efficiency 
reasons, do not wish to compute the 
greatest fixed point solution of every 
predicate variable. 
Fortunately Cosman and Jhala~\cite{cosman_local_2017}
show that acyclic predicate variables 
can be eliminated one by one. 
We explain first how to eliminate $\h\kappa$ variables and then how to eliminate
$\h\pi$ variables.

\mypara{Step 2: Eliminating $\h{\kappa}$-Variables from \h{c}}
Procedure $\elimksym$ of \cite{cosman_local_2017} 
eliminates each individual acyclic predicate variable, 
$\kappa$, in an \nnf{} Horn Clause.
Briefly, given a predicate variable $\kappa$ 
in an \nnf{} Horn Clause $c$, the procedure 
computes the strongest solution for $\kappa$:
$\solk{\kappa}{c}$, and then substitutes the 
solution into $c$.
For the single $\kappa$ in our example this solution
$\solk{\kappa}{c}$ is $\lambda x.(\exists z'.\true \wedge (\exists v'.v' = 1 \wedge v' = x))$.
After substitution and simplification we get 
\begin{align*}
  & \wedge\ \forall n. \pi(n) \Rightarrow \forall v. \true \Rightarrow \forall z'. \true \Rightarrow \forall v'. v' = 1 \Rightarrow v = v' \wedge v = n \\
  & \wedge\ \exists n. \pi(n)
\end{align*}

\begin{figure}[t]
\begin{center}
  \[\begin{array}{lcl}
    \toprule
    \elimsym_{qe} & : & P^C \rightarrow \overline{\Pi} \times C^{\Pi} \times C \rightarrow C\\
    \midrule

    \elimqe{qe}{[]}{[]}{\sigma}{c} & \doteq & c\\

    \elimqe{qe}{[]}{\pi_n:\overline{\pi}}{\sigma}{c} & \doteq &
    \elimqe{qe}{[]}{\overline{\pi}}{\sigma}{\sosubst{c}{\pi_n}{\lambda \overline{x}.p}}\\
    \quad \text{where} \quad p & = & qe(\solp{qe}{\{\pi_n\}}{\sigma}{\sigma(\pi_n)})\\

    \midrule
    \solpsym_{qe} & : & P^C \rightarrow \overline{\Pi} \times C^{\Pi} \times C \rightarrow C\\
    \midrule
    \multicolumn{2}{l}{$\solp{qe}{\overline{\pi}}{\sigma}{\forall n.\pi_n(n,\overline{x}) \Rightarrow c}$}\\
    \quad \mid \pi_n \in \overline{\pi} & \doteq &  \solp{qe}{\overline{\pi}}{\sigma}{c} \\
    \quad \mid \pi_n \notin \overline{\pi} & \doteq & \forall n.p \Rightarrow \solp{qe}{\overline{\pi}}{\sigma}{c}\\
    \quad \quad \text{where } p & =  & qe(\solp{qe}{\overline{\pi} \cup \{\pi_n\}}{\sigma}{\sigma(\pi_n)})\\

    \solp{qe}{\overline{\pi}}{\sigma}{\forall x.p \Rightarrow c} & \doteq &
    \forall x.p \Rightarrow \solp{qe}{\overline{\pi}}{\sigma}{c}\\

    \solp{qe}{\overline{\pi}}{\sigma}{c_1 \wedge c_2} & \doteq &
    \solp{qe}{\overline{\pi}}{\sigma}{c_1} \wedge  \solp{qe}{\overline{\pi}}{\sigma}{c_2}\\

    \solp{qe}{\overline{\pi}}{\sigma}{p} & \doteq & p\\
    \bottomrule
  \end{array}\]
\end{center}
  \caption{Eliminating $\pi$ Variables and their Side Conditions}
  \label{fig:elimE}
\label{fig:elimp}
\end{figure}

\mypara{Step 3: Eliminating Skolem Variables}
$\elimksym$ removes all the $\kappa$ predicate
variables leaving only the $\pi_n$ variables 
inserted by $\pokesym$.
The inhabitation side conditions require we 
find the greatest fixed point (\textsc{gfp}) solution to 
these variables to ensure we do not spuriously
eliminate witnesses.
As $\pi_n$ only appears negatively (in guards), the \textsc{gfp} is the conjunction of
every $c$ appearing as $\forall n.\pi_n(n, \overline{x}) \Rightarrow c$.
For a given $\pi_n$ appearing in the constraint 
$c'$, we write this \textsc{gfp} as $\defconstr{\pi_n}{c'}$, 
the \emph{defining constraint} of $\pi_n$.

A first challenge arises in that we wish to use 
these solutions to \emph{eliminate} the Skolem 
predicate variables, but $\defconstr{\pi_n}{c'}$ 
is a conjunction of clauses featuring quantifiers.
As the Skolem variables appear in guards, we must 
transform these $c$ into an equisatisfiable predicate 
$p$ before we may substitute, so we parameterize our 
elimination algorithms with a quantifier elimination 
algorithm $qe$ that handles this task.
$qe$ can vary with the particular domain and, for 
domains that do not admit quantifier elimination 
or where it is infeasible, we instead use an approximation
described in \secref{sec:solving:special-case}. 

A second technical challenge arises en route 
to our solving algorithm:
$\defconstr{\pi_n}{c'}$ may contain \emph{other} 
$\pi$ variables and may contain cycles involving 
$\pi$ variables.
Fortunately, recursive Skolem variables are redundant:
\begin{lemma} \label{lem:pi-no-mind}
  If $\pi_n$ is a predicate variable inserted by $\pokesym$, then
  $\vDash \forall n.\pi_n(n,\overline{x}) \Rightarrow c$ iff $\vDash\forall n.\pi(n,\overline{x})\Rightarrow c[\lambda \_ . \true/\pi]$.
\end{lemma}
\noindent
This lets us to break cycles by ignoring recursive occurrences of Skolem variables.

With this result in hand, we develop the procedure $\solpsym_{qe}$, as shown in \figref{fig:elimp}.
$\solpsym_{qe}$ recursively eliminates Skolem variables from the defining constraint of $\pi_n$, using $qe$ to transform a nested Skolem variable's defining constraint (tracked in the map $\sigma$) into a substitutable predicate.
The first argument $\overline{\pi}$ tracks the seen Skolem variables, treating them as if they were solved to $\true$ when they next occur.

$\solpsym_{qe}$ is used in the procedure $\elimsym_{qe}$ which eliminates all Skolem variables from the constraint one $\pi_n$ at a time.
$\elimsym_{qe}$ simply calls $\solpsym_{qe}$ on the defining constraint of $\pi_n$ and then $qe$'s the result (now free of Skolem variables) before substituting the returned predicate as the solution for $\pi_n$.

Regardless of the properties of $qe$ we have the following soundness theorem:
\begin{theorem} \label{thm:elimqe-sound}
  If $c' = \poke{\emptyset}{c}$, 
     $\overline{\pi}$ is the set of all Skolem variables in $c'$, 
     $c'$ has no other predicate variables, 
     $\sigma(\pi_n) = \defconstr{\pi_n}{c'}$ for all $\pi_n \in \overline{\pi}$, and
     $\vDash \elimqe{qe}{}{\overline{\pi}}{\sigma}{c'}$ 
  then $\vDash c$.
\end{theorem}
If $qe$ is a strengthening and a weakening, the converse also holds and our algorithm is complete.

\mypara{Step 4: Putting It All Together}
The procedure $\safe$ (\figref{fig:algo:safe} in \appref{sec:solving:complete}) 
puts together all the pieces to automatically check the validity of acyclic \ehc{} 
constraints $c$.
Like $\solpsym$ and $\elimsym$, $\safe$ is parameterized by the quantifier 
elimination procedure $qe$.
$\safe$ first Skolemizes the constraint $c$ and then solves the original 
predicate variables by repeated applications of $\elimksym$.
Next, $\safe$ collects the defining constraints of the Skolem 
predicate variables and uses $\elimsym_{qe}$ to solve and eliminate 
the Skolem variables.
This leaves us with a predicate-variable-free constraint, 
but still featuring nested existential quantifiers from 
the inhabitation side conditions.
Here, we use $qe$ once again to eliminate the existential 
quantifiers, leaving us with a VC ripe to be passed to an 
automated theorem prover to verify validity.
If this VC is valid, we deem $c$ \emph{safe}.
If $qe$ is a strengthening then $\safe$ is sound
and if $qe$ is additionally a weakening then 
$\safe$ is complete, proved in Appendix E.

\begin{theorem} 
  \label{thm:safe-sound}
  \label{thm:safe-complete}
  Let $c$ be an acyclic constraint. If $qe$ is a strengthening then $\safe(c)$ implies $\vDash c$.
  If $qe$ is also a weakening then $\safe(c)$ iff $\ \vDash c$.
\end{theorem}

\subsection{A Theory-Agnostic Approximation to $qe$}
\label{sec:solving:special-case}

We cannot, in general, provide a quantifier 
elimination that is a strengthening and a 
weakening (since we are agnostic to the set 
of theories) especially as some theories 
do not admit a decidable quantifier instantiation! 
However, since SMT theories work by equality 
propagation, we can make use of equalities 
between theory terms without making any 
additional assumptions about the theories 
themselves.
Therefore, we include a theory-agnostic
quantifier elimination strategy over equalities.

Given the defining constraint $c$ of some 
$\pi_n(n, \overline{x})$, our strategy
computes the (well-scoped) congruence 
closure of the variables $n$ and $\overline{x}$
using the body of $c$.
This set of equalities is then used as 
the solution to $\pi_n$.
To eliminate the existential 
side condition $\exists n.\pi_n(n, \overline{x})$ 
we note that a sound approximation 
is to find a solution for $n$.
We search for this solution within 
the set of equalities.
If there is not one, we return $\false$ 
and verification fails.
In \sectionref{sec:eval} we evaluate this 
incomplete quantifier elimination on a 
number of benchmarks and demonstrate
its real-world effectiveness.

\section{Evaluation}

\label{sec:eval}
\begin{table}[t]
  \begin{tabular}{|l|l|l|l|l|l|l|l|l|}
    \hline
    \multirow{2}{*}{\diagbox{\footnotesize{Case Study}}{\footnotesize{Tool}}} & \multicolumn{4}{l|}{\toolname} & \multicolumn{2}{l|}{\mochi~\cite{unno_automating_2013}} & \multicolumn{2}{l|}{F* Implicits~\cite{fstar-implicits}} \\
    \cline{2-9}
    & LOC & Time (s) & Spec & Check & Spec & Check & Spec & Check \\
    \hline
    \textsc{incr}                     & 8    & 0.01  & \cmark & \cmark & \cmark & \xmark & \cmark & \xmark \\
    \textsc{sum}                      & 5    & 0.01  & \cmark & \cmark & \cmark & \cmark & \cmark &  \xmark \\
    \textsc{repeat}                   & 86   & 0.28  & \cmark & \cmark & \cmark & \xmark & \xmark & \xmark \\
    $d_2$~\cite{unno_automating_2013} & 8    & 0.07  & \cmark & \xmark & \cmark & \cmark & \cmark & \xmark \\
    \hline
    \multicolumn{6}{l}{\textbf{Resources}} \\
    \hline
    \textsc{incrState}                & 28   & 0.65  & \cmark & \cmark & \cmark & \cmark & \cmark & \cmark \\
    \textsc{accessControl}            & 49   & 0.36  & \cmark & \cmark & \cmark & \xmark & \cmark & \cmark\footnotemark \\
    \textsc{tick}~\cite{liquidate}    & 85   & 0.03  & \cmark & \cmark & \xmark & \xmark & \cmark & \cmark \\
    \textsc{linearDSL}                & 20   & 0.03  & \cmark & \cmark & \xmark  & \xmark & \cmark & \cmark\footnotemark[1] \\
    \hline
    \multicolumn{6}{l}{\textbf{State Machines}} \\
    \hline
    \textsc{pagination}                              & 79   & 0.3    & \cmark & \cmark & \xmark & \xmark & \xmark & \xmark \\
    \textsc{login}~\cite{BradyStateMachinesAll2016}  & 121  & 0.09   & \cmark & \cmark & \xmark & \xmark & \xmark & \xmark \\
    \textsc{twoPhase}                                & 135  & 0.86   & \cmark & \cmark & \xmark & \xmark & \xmark & \xmark \\
    \textsc{tickTock}                                & 112  & 0.14   & \cmark & \cmark & \xmark & \xmark & \xmark & \xmark \\
    \textsc{tcp}                                     & 332  & 115.73 & \cmark & \cmark & \xmark  & \xmark & \xmark & \xmark \\
    \hline
  \end{tabular}
  \caption{Comparison of systems: \cmark on ``Spec'' and ``Check'' respectively 
           indicate that the specification can be written and that the code can 
           be verified by automatically instantiating the implicit parameters.}
  \label{tab:eval}
\end{table}

We implement our system of Implicit Refinement Types in a tool dubbed \toolname, evaluate \toolname using a set of illustrative examples and case studies (links to full examples elided for DBR), and compare \toolname against the existing state of the art, in order to answer the following questions:
\begin{description}
  \item[Q1: Lightweight Verification] Are implicits in conjunction with the theory-agnostic instantiation procedure sufficient to verify programs with IRTs, even without using heavyweight instantiation techniques such as domain-specific solvers and synthesis engines?
  \item[Q2: Expressivity] Can we use implicits to encode specifications that would've otherwise required the use of additional language features?
  \item[Q3: Flexibility] Do they allow \emph{automated verification} in places where unification-based implicits and CEGAR-based extra parameters do not?
\end{description}

\footnotetext{\label{fn:fstar}Verification requires annotating a simple fact about sets as \fstar does not include a native theory of sets.}

\mypara{Implementation}
\label{sec:eval:implementation}
\toolname extends our language (\secref{sec:lang}) with polymorphism
and type constructors, and omits concrete syntax for implicit lambdas and
unpacks.
Instead, implicit parameters appear solely in specifications---that is,
refinement \emph{types}---and do not require implementation changes to the code.

Inserting implicit lambdas is straightforward: when checking
that a term has an implicit function type, insert a corresponding implicit
lambda.
Inserting implicit unpacks is more interesting: we need to unpack any term of
implicit pair type before we use it.
For instance, if $e$ has the type $\tcoifun{x}{t_1}{t_2}$,
then we must unpack $e$ to extract the
corporeal ($t_2$) component of the pair before we can apply a function to it:
$\app{e_f}{e}$ becomes $\eunpack{x}{y}{e}{\app{e_f}{y}}$.
This transformation ensures that the ghost parameter ($x$ in the above term) is
in context at the use site of the implicit pair.
To automate this procedure, we use an \textsc{Anf}-like
transformation that restricts A-normalization to terms of implicit pair type.

As is common in similar research tools, \toolname handles datatypes by axiomatizing
their constructors. A production implementation would treat surface datatype
declarations as sugar over these axioms.

\mypara{Polymorphism}
\toolname also includes support for limited refinement polymorphism.
For simplicity, we left refinement
polymorphism out of our formalism in \secref{sec:lang:static} --- it is
orthogonal to the addition of implicit parameters and pairs --- but we did
include it in our implementation.

Refinement polymorphism alleviates some issues with phantom type parameters. 
First, when using phantom type parameters, core constructors cannot be directly 
verified and must be assumed to have the given type. Moreover, we can specify 
the semantics of our stateful APIs in terms of \eg the \texttt{HST} type, 
but the meaning of the arguments to the HST type operator are determined 
only by its use.
In contrast, refinement polymorphism brings the intended semantics
of HST from the world of phantom parameters into the semantics
of the language itself. The semantics of HST are reflected with the type:
\begin{code}
  type HST = rforall p q. forall a. p -> (q, a)
\end{code}
where \texttt{rforall} ranges over refinement types and \texttt{forall}
ranges over base types. The more sophisticated refinement
polymorphism~\cite{vazou_bounded_2015} present in
existing refinement systems would use the type:
\begin{code}
  type HST p q s a = {w:s | p w} -> ({w':s | q w'}, a)
\end{code}
Here, \texttt{p} and \texttt{q} are predicates on the state type
\texttt{s}.

Refinement polymorphism further makes core constructors and API primitives
themselves subject to verification: the \texttt{rforall} version of HST above
allows the direct verification of \texttt{get} as
\begin{code}
  get : forall s. [w:s] -> HST {v:s | v = w} {v:s | v = w} s
  get = \s -> (s, s)
\end{code}

\mypara{Comparsion}
\label{sec:eval:comparison}
We compare \toolname to higher-order model checker
\mochi~\cite{unno_automating_2013}, and
F*'s support for implicit parameters\footnote{%
Note that F*'s main verification mechanism is Dijkstra Monads~\cite{Swamy:2013}, which we do not compare against in our evaluation.
We discuss trade offs between Dijkstra Monads and Implicit Refinement Types in \secref{sec:related}.
}~\cite{fstar, fstar-implicits}.
Both systems, like \toolname and unlike foundational verifiers
 such as Idris~\cite{brady_idris_2013} and Coq~\cite{the_coq_development_team_coq_2009}, are
designed for lightweight, automatic verification.
\mochi aims to provide complete verification of higher-order programs by
automatically inserting extra (implicit) parameters.
Whereas \toolname has users write refinement type specifications (with
\emph{explicit} reference to implicit types), \mochi's specifications are implemented as
assertions within the code.
F*'s type system is a mix of a Martin-L\"of style dependent type system with
SMT-backed automatic verification of refinement type specifications.
There is no formal specification of implicit parameters in F*. They are
implemented by unification as part of Martin-L\"of typechecking. This is in
contrast to IRTs' integrated approach, that uses information from refinement subtyping constraints for
instantiation.

Our comparison, summarized in \tabref{tab:eval}, illustrates, via a series of
case studies, the specifications that can be written in
each system (the Spec column) and whether they can find the necessary
implicit parameter instantiations (the Check column).
As each tool is designed to be used for lightweight verification, we
write the implementation and then separately write the specifications.
We do not rewrite the implementations to better accommodate specifications,
though for \mochi we insert assertions within the code as necessary.

\subsection{Q1: Lightweight Verification}

We evaluate whether IRTs allow
for modular specifications that would otherwise 
be inexpressible with plain refinement types.
We do this via a series of higher-order programs that use implicits for
lightweight verification.
These programs are designed to capture core aspects of 
various APIs and how IRTs permit specifications 
that can be automatically verified in a representative 
client program using the API.
Notably, for all of these examples \toolname only 
uses the theory-agnostic instantiation procedure 
from \secref{sec:solving:special-case} and does not
employ heavyweight instantiation techniques such 
as domain-specific solvers or synthesis engines.

\mypara{Higher-Order Loops}
\textsc{repeat}, defines a 
loop combinator \ha{repeat} that takes 
an increasing stateful computation \ha{body} 
and produces a stateful computation that loops 
\ha{body} \ha{count} times.
Its signature in \toolname is:
\begin{code}
  repeat :: body:([x:Int] -> ([y:{v:Int | v > x}]. SST x y Int Int))
         -> count:{v:Int | v > 0}
         -> ([q:Int] -> ([r:{v:Int | v > q}]. SST q r Int Int))
\end{code}
which says that the input and output are stateful computations whose history
and prophecy variables together guarantee that the computation will leave the
state larger than it started.
In the first line the implicit function argument \ha{x} is externally determined
and captures the state of the world before the \ha{body} computation, while the
implicit pair component \ha{y} captures that \ha{body} updates the state of the
world to some (unknown) larger value.

The implicit pairs are necessary to specify \ha{repeat},
as the updated state is determined by an internal choice 
in \ha{body}.
\textsc{repeat} also demonstrates how implicit pairs can 
specify loop invariants (here, upward-closure) on higher-order 
stateful programs.
Though \ha{repeat} is specified using implicit pairs, 
\toolname does \emph{not} require that arguments passed 
to \ha{repeat} be specified using implicit pairs.
In \textsc{repeat} we define a function \ha{incr} with 
the type \ha{[x:Int] $\rightarrow$ SST x (x + 1) Int Int} 
that increments the state.
\toolname will appropriately type check and verify \ha{repeat incr}
as Implicit Refinement Types automatically account for the necessary subtyping 
constraints to ensure that \ha{incr} meets the conditions of \ha{repeat}.

\mypara{State Machines}
The next set of examples demonstrate that IRTs enable specification verification
of state-machine based protocols.
These are ubiquitous, spanning appications from  
networking protocols to device driver and operating 
system invariants~\cite{BradyStateMachinesAll2016}.
IRTs let us encode state machines as transitions 
allowed from a ghost state, using the Hoare State Monad (\ref{sec:overview:state}).
\textsc{login}~\cite{BradyStateMachinesAll2016}, 
which models logging into a remote server, served as our 
initial inspiration for studying this class of problems.
We verify that a client respects the sequence of 
connect, login, and only then accesses information.

\textsc{tickTock} verifies that two \ha{ticker} and
\ha{tocker} processes obey the specification standard to the concurrency 
literature~\cite{scalas_verifying_2019}.
Here we show the implementation of \ha{tocker}.
\begin{code}
  tocker = \c -> do
  msg <- recv c
  if msg = tick then send c tock else assert False
\end{code}
Ghost parameters on \ha{send} and \ha{recv} track 
the state machine and ensure messages follow the 
tick-tock protocol.
If \mbox{\ha{send c tock}} is changed to \mbox{\ha{send c tick}}
the program is appropriately rejected.

\textsc{twoPhase} is a verified implementation of 
one side of a two-phase commit process. 
This example serves as a scale model for \textsc{tcp}, 
a verified implementation of a model of a TCP client 
performing a 3-way handshake, using the TCP state 
machine\cite{RFC0793}.

\textsc{pagination} is an expanded version of our 
stream example from \sectionref{sec:overview}, and 
models the AWS S3 pagination API~\cite{RetrievingPaginatedResults}.
This example shows that our protocol state machine 
need not be finite, as it is specified with respect 
to an \emph{unbounded} state machine.

\mypara{Quantitative Resource Tracking}
The next set of examples reflect various patterns 
of specifying and verifying \emph{resources} and 
demonstrate that Implicit Refinement Types enable
lightweight verification of these patterns.
The examples show how \toolname handles specification 
and automatic verification when dealing with resources 
such as the state of the heap and quantitative
resource usage.
In \secref{subsec:examples:grant} we saw how Implicits enable
specifying the access-control API from \figref{fig:revoke}; 
the \textsc{accessControl} example verifies clients of this API.
In contrast, \textsc{tick} (as shown in \secref{subsec:tick}) follows Handley et
al.~\cite{liquidate} in defining an applicative functor that tracks quantitative
resource usage.
The key distinction from Handley et al.~\cite{liquidate} is that the resource
count exists only at the type level.
This example generalizes: we can port any of 
the \emph{intrinsic verification} examples 
from that work to use implicit parameters 
instead of explicitly passing
around resource bounds.

\textsc{linearDSL} embeds a simple linearly typed DSL in \toolname.
It allows us to embed linear terms in \toolname, with linear
usage of variables statically checked by our refinement type system.
The syntactic constructs of this DSL are smart constructors that take the typing
environments as implicit parameters, enforcing the appropriate linear typing
rules.

\toolname enables specification and automatic verification of this diverse range
of examples, and in fact only requires the theory-agnostic instantiation
procedure to find the correct implicit instantiations.
This demonstrates that IRTs enable lightweight
verification of a variety of higher-order programs, and that they are
widely useful even without a domain specific solver or heavyweight 
synthesis algorithm.

\subsection{Q2: Expressivity}

We compare the expressivity of Implicit Refinement Types as implemented in
\toolname to the assertion-based verification of \mochi and the implicit
parameters of \fstar. 
First, merely the fact that we allow users to explicity write specifications
using implicit parameters allows us to write specifications we otherwise
could not. In particular, this is witnessed by the \textsc{tick} and
\textsc{linearDSL} examples that specify resource tracking, a
non-functional property. \mochi cannot express the specifications
for these because there is no combination of program variables that
computes the (non-functional) \emph{usage} properties used in these specifications.

Second, Implicit Pair Types add expressivity by allowing us to write
specifications against choices made \emph{internal} to functions that we
wish to reason about.
For example, \textsc{repeat} uses implicit pairs to specify 
a loop invariant. 
This example can't be done with \fstar's implicits, as, absent
implicit pairs, we cannot bind to the return value of the loop 
body.
As a result, in \fstar we cannot verify this example with 
implicits alone: we would have to bring in additional 
features such as the Dijkstra monad~\cite{Swamy:2013}.

Similarly, \emph{none} of the protocol state machine examples 
can be encoded with implicit functions alone. Implicit pairs 
are required to write specifications against choices made by 
other actors on the protocol channel. 
As a result, these examples also cannot be encoded 
in \fstar's implicits. These state machine specifications 
cannot even be expressed in \mochi as they crucially 
require access to ghost state, which \mochi does not support.

\subsection{Q3: Flexibility}

\begin{figure}[t!]
\begin{code}
incr :: [n:Int] -> (Int -> SInt n) -> SInt (n + 1)
incr f = (f 0) + 1

test1 :: SInt 11                      test2 :: m:Int -> SInt (m + 1)
test1 = incr (\x -> 10)               test2 m = incr (\x -> m)
\end{code}
\caption{A higher order increment function}
\label{fig:incr}
\end{figure}
 
We compare the flexibility that our refinement-integrated 
approach gives us to solve for implicit parameters
relative to that of other systems.
We focus on the features of our semantics and
our abstract solving algorithm from~\sectionref{sec:solving:special-case} 
independently of the choice of quantifier
elimination procedure, which we examined above.
We illustrate these differences with several examples.
\textsc{incr} is the program from \figref{fig:incr}.
In \mochi the specification is given by assertions that 
\ha{test1} and \ha{test2} are equal to \ha{11} and \ha{m + 1} 
respectively.
\textsc{incrState} generalizes \textsc{incr} to track 
the integer in the singleton state monad instead of a 
closure.
\textsc{sum} is similar to \ha{incr} except that it takes 
two implicit arguments and two function arguments, returning 
the sum of the two returned values:
\begin{code}
  sum :: [n, m] -> (Int -> SInt n) -> (Int -> SInt m) -> SInt (n + m)
  sum f g = (f 0) + (g 0)
  
  test :: SInt 11
  test = sum (\x -> 10) (\y -> 1)
\end{code}
All three tools can specify the program. 
\toolname and \mochi successfully instantiate 
the implicit parameters needed for verification.
\textsc{D2} is an example from the 
\mochi~\cite{unno_automating_2013} 
benchmarks that loops a nondeterministic 
number of times and adds some constant each time.
Unlike \toolname, \mochi is unable 
to solve \textsc{accessControl}, as 
CEGAR is notoriously brittle on 
properties over the theory of 
set-operations.

\mypara{Comparsion to \mochi}
\mochi fails to automatically verify \textsc{incr}, 
but it succeeds if either \ha{test1} or \ha{test2} 
appear individually.
This is partly due to the unpredictability of \mochi's 
CEGAR loop~\cite{JhalaMcMillan06} (as it succeeds in 
verifying \textsc{incrState}), and partly due to the fact that \mochi attempts to infer specifications
(and implicit arguments) 
\emph{globally}, which can lead to the anti-modular behavior seen here.
In contrast, by introducing implicit arguments as a user-level specification
technique, \toolname permits modular verification: \ha{test1} and \ha{test2}
generate two separate quantifier instantiation problems that \toolname solves
locally.

In \textsc{D2}, \toolname \emph{fails} 
to solve the constraints generated by 
implicit instantiation as they involve
systems of inequalities, which our theory-agnostic 
quantifier elimination does not attempt.
However, \toolname captures the full set of constraints 
and could, with a theory specific solver like \mochi's, 
verify \textsc{D2}.
Concretely, \textsc{D2} yields a constraint 
of the form $\exists x. true \Rightarrow x > 3$.
\toolname could discharge this obligation by instantiating the 
abstract algorithm of \sectionref{sec:solving}
with either the solver from \mochi\cite{unno_automating_2013} 
or EHSF\cite{RybalchenkoEHSF} instead of the theory-agnostic one above.

\mypara{Comparsion to Unification}
The results of the comparison to \fstar's unification-based implicits are
summarized in Table~\ref{tab:eval}.
Implicit Refinements are more flexibile as \fstar attempts to solve the
implicits purely through unification, \ie without accounting for refinements. In
contrast, \toolname generates subtyping (implication) constraints that are
handled separately from unification.
When the unification occurs under a type 
constructor, then unification can succeed.
For example, \fstar verifies \textsc{incrState}
because elevating the ghost state to a parameter 
of the singleton state type constructor makes 
it ``visible'' to \fstar's unification algorithm.
However, if there is \emph{no} type constructor
to guide the unification, \fstar must rely on
higher-order unification, which is undecidable 
in general and difficult in practice.
For this reason, in the \textsc{incr} example 
from Figure~\ref{fig:incr}, \fstar fails to 
instantiate the implicit arguments needed to 
verify \ha{test1} and \ha{test2}, even if we 
aid it with precise type annotations on 
\ha{\x $\rightarrow$ 10} and \ha{\x $\rightarrow$ m}, and
\fstar fails to verify the \textsc{sum} example 
for the same reason.
By contrast, \toolname can take advantage 
of the refinement information to
solve for the implicit parameters. 
Finally, \fstar's unification fundamentally 
cannot handle an example like \mochi's 
\textsc{D2}, as unification will not be 
able to find an implicit instantiation
when it is only constrained by inequalities.

\section{Related Work}\label{sec:related}

\mypara{Verification of Higher-Order Programs}
As discussed in \sectionref{sec:eval}, \fstar~\cite{swamy_dependent_2016} is
another SMT-aided higher-order verification language.
Its main support for verifying stateful programs is baked in
via Dijkstra Monads~\cite{Swamy:2013}.
When it comes to verifying higher-order stateful programs, Dijkstra Monads
enable \fstar to scale automatic verification up to complex invariants.
However, like other two-state specification
techniques, they are not on their own flexible enough to handle higher-order
computations such as
\ha{mapA}, when higher-order computations are composed in richer ways than
simple Kleisli composition.
The key difference is that Implicit Refinement Types work for both
classic two-state specifications \emph{and other} use-cases where two-state
specifications prove cumbersome --- IRTs add value even to a system that already includes
two-state specifications by allowing the necessary extra parameter of
functions like \ha{mapA} to be instantiated automatically. Moreover, IRTs do so
while maintaining the key properties of refinement type systems: (1) SMT-driven
automatic verification and (2) refinement type specifications are added atop an
existing program without requiring code changes.  In this way IRTs also differ
from systems that use ``full-strength'' two-state specifications with dependent
types such as VST\cite{Appel12} or Bedrock\cite{Chlipala13}.

\fstar has both implicit parameters and refinement types, but \fstar{}'s
implicits have no formal
description, and instantiation is independent of refinement information.
On the other hand, we lack first-class invariants, but future work may be able
to alleviate this by abstracting over refinements \emph{a la} Vazou et
al.~\cite{vazou_abstract_2013}.

Dafny~\cite{leino_dafny_2010} supports specifications with ghost variables, but
the user must explicitly craft triggers to perform quantifier instantiation when
the backing SMT solver's heuristics cannot, and must manually pass and update
ghost variables.

\mypara{Implicit Parameter Instantiation}%
Since finding an instantiation of implicit parameters is, in general,
undecidable, different systems make different tradeoffs:
While Haskell and Scala only perform type-directed lookup,
Idris~\cite{brady_idris_2013} resolves implicits via first-order unification,
with a default value or by a fixed-depth enumerative program synthesis.
This form of implicit resolution system can be very powerful, but can't be used
in conjunction with solver-automated reasoning so, for example, Idris would not be
able to handle our \textsc{sum} example which requires automated reasoning
about arithmetic.
Agda~\cite{norell_dependently_2009} combines unification with a second kind of implicit
parameter, \emph{instance arguments}~\cite{devriese_bright_2011}, that uses a
separate, specialized, implicit resolution mechanism for implicit arguments
that relate to typeclasses.
Coq~\cite{the_coq_development_team_coq_2009} uses a mechanism called canonical
structures~\cite{mahboubi_canonical_2013}, which uses a programmable hint system, but is intimately tied to the
specifics of Coq's implementation, and its use for implicit parameter
instantiation lacks a formal description, as Devriese et
al.~\cite{devriese_bright_2011} lament.

As Devriese et al.~\cite{devriese_bright_2011} note, the complexity of dependent type
systems make even achieving similar functionality as in Haskell and Scala a
significant task.
These implicit parameters are
designed to accomplish similar tasks as in the non-dependent Haskell and Scala:
automating the instantiation of repetitive arguments or automatically searching
the context for relevant arguments.
Refinement types allow us to sidestep this issue as the base type system can separately
use implicit parameters to automate typeclasses or other programming tasks,
allowing our technique of Implicit Refinement Types to entirely focus on
using dependent type information.
Here, this means focusing on allowing us to use an SMT to simplify ghost specifications
without worrying about interactions with the base type system.

While there is much work on implicit parameters for \emph{dependent types}, we believe
that we provide the first formal description of a system combining implicit
parameters with refinement types.
\fstar is the only other example of a language that has implicit parameters and
refinement types, but we discuss how \fstar's implicits cannot take advantage of
refinement type information in \secref{sec:eval}.

\mypara{Horn Constraints}
Horn Clauses have emerged as a lingua franca 
of verification tools~\cite{bjorner_horn_2015}
as they offer a straightforward encoding of assertions.
Z3~\cite{de_moura_z3_2008} includes a fixpoint 
solver for Horn Constraints including many 
quantifier elimination heuristics.
Rybalchenko et al.~\cite{RybalchenkoEHSF} present a semi-decision 
procedure for solving existential Horn clauses 
using a template-based CEGAR loop.
Cosman and Jhala~\cite{cosman_local_2017} use NNF Horn constraints
to preserve \emph{scope}. We extend this framework 
to synthesize refinements for implicit program 
variables that are \emph{existentially} quantified, 
in addition to the usual \emph{universally} quantified 
binders.

Unno et al.~\cite{unno_automating_2013} show that
it is sufficient to add one extra
parameter for each higher order argument 
to any given function to achieve complete 
higher-order verification with first-order 
refinements.
Their treatment utilizes more automation, \eg interpolation and Farkas' lemma,
but is limited to arithmetic specifications, precluding programs manipulating
data structures like sets and maps, as demonstrated in \secref{sec:eval}.
\bibliography{bib}

\begin{thebibliography}{10}

\bibitem{RetrievingPaginatedResults}
Retrieving paginated results - {AWS} {SDK} for java version 2, 2019.
\newblock URL:
  \url{https://docs.aws.amazon.com/sdk-for-java/v2/developer-guide/examples-pagination.html}.

\bibitem{abadi_existence_1988}
Mart{\'{\i}}n Abadi and Leslie Lamport.
\newblock The existence of refinement mappings.
\newblock In {\em Proceedings of the Third Annual Symposium on Logic in
  Computer Science {(LICS} 1988), Edinburgh, Scotland, UK, July 5-8, 1988},
  pages 165--175. {IEEE} Computer Society, 1988.
\newblock \href {https://doi.org/10.1109/LICS.1988.5115}
  {\path{doi:10.1109/LICS.1988.5115}}.

\bibitem{Appel12}
Andrew~W. Appel.
\newblock Verified software toolchain.
\newblock In Alwyn Goodloe and Suzette Person, editors, {\em {NASA} Formal
  Methods - 4th International Symposium, {NFM} 2012, Norfolk, VA, USA, April
  3-5, 2012. Proceedings}, volume 7226 of {\em Lecture Notes in Computer
  Science}, page~2. Springer, 2012.
\newblock \href {https://doi.org/10.1007/978-3-642-28891-3\_2}
  {\path{doi:10.1007/978-3-642-28891-3\_2}}.

\bibitem{GordonRefinement09}
Jesper Bengtson, Karthikeyan Bhargavan, C{\'{e}}dric Fournet, Andrew~D. Gordon,
  and Sergio Maffeis.
\newblock Refinement types for secure implementations.
\newblock In {\em Proceedings of the 21st {IEEE} Computer Security Foundations
  Symposium, {CSF} 2008, Pittsburgh, Pennsylvania, USA, 23-25 June 2008}, pages
  17--32. {IEEE} Computer Society, 2008.
\newblock \href {https://doi.org/10.1109/CSF.2008.27}
  {\path{doi:10.1109/CSF.2008.27}}.

\bibitem{RybalchenkoEHSF}
Tewodros~A. Beyene, Corneliu Popeea, and Andrey Rybalchenko.
\newblock Solving existentially quantified horn clauses.
\newblock In Natasha Sharygina and Helmut Veith, editors, {\em Computer Aided
  Verification - 25th International Conference, {CAV} 2013, Saint Petersburg,
  Russia, July 13-19, 2013. Proceedings}, volume 8044 of {\em Lecture Notes in
  Computer Science}, pages 869--882. Springer, 2013.
\newblock \href {https://doi.org/10.1007/978-3-642-39799-8\_61}
  {\path{doi:10.1007/978-3-642-39799-8\_61}}.

\bibitem{bjorner_horn_2015}
Nikolaj Bj{\o}rner, Arie Gurfinkel, Kenneth~L. McMillan, and Andrey
  Rybalchenko.
\newblock Horn clause solvers for program verification.
\newblock In Lev~D. Beklemishev, Andreas Blass, Nachum Dershowitz, Bernd
  Finkbeiner, and Wolfram Schulte, editors, {\em Fields of Logic and
  Computation {II} - Essays Dedicated to Yuri Gurevich on the Occasion of His
  75th Birthday}, volume 9300 of {\em Lecture Notes in Computer Science}, pages
  24--51. Springer, 2015.
\newblock \href {https://doi.org/10.1007/978-3-319-23534-9\_2}
  {\path{doi:10.1007/978-3-319-23534-9\_2}}.

\bibitem{BradyStateMachinesAll2016}
Edwin Brady.
\newblock State machines all the way down.
\newblock 2016.

\bibitem{brady_idris_2013}
Edwin~C. Brady.
\newblock Idris, a general-purpose dependently typed programming language:
  Design and implementation.
\newblock {\em Journal of Functional Programming}, 23(5):552--593, 2013.
\newblock \href {https://doi.org/10.1017/S095679681300018X}
  {\path{doi:10.1017/S095679681300018X}}.

\bibitem{Chlipala13}
Adam Chlipala.
\newblock The bedrock structured programming system: combining generative
  metaprogramming and hoare logic in an extensible program verifier.
\newblock In Greg Morrisett and Tarmo Uustalu, editors, {\em {ACM} {SIGPLAN}
  International Conference on Functional Programming, ICFP 2013, Boston, MA,
  {USA} - September 25 - 27, 2013}, pages 391--402. {ACM}, 2013.
\newblock \href {https://doi.org/10.1145/2500365.2500592}
  {\path{doi:10.1145/2500365.2500592}}.

\bibitem{cosman_local_2017}
Benjamin Cosman and Ranjit Jhala.
\newblock Local refinement typing.
\newblock {\em Proceedings of the {ACM} on Programming Languages},
  1({ICFP}):26:1--26:27, 2017.
\newblock \href {https://doi.org/10.1145/3110270} {\path{doi:10.1145/3110270}}.

\bibitem{oliveira_type_2010}
Bruno~C. d.~S.~Oliveira, Adriaan Moors, and Martin Odersky.
\newblock Type classes as objects and implicits.
\newblock In William~R. Cook, Siobh{\'{a}}n Clarke, and Martin~C. Rinard,
  editors, {\em Proceedings of the 25th Annual {ACM} {SIGPLAN} Conference on
  Object-Oriented Programming, Systems, Languages, and Applications, {OOPSLA}
  2010, October 17-21, 2010, Reno/Tahoe, Nevada, {USA}}, pages 341--360. {ACM},
  2010.
\newblock \href {https://doi.org/10.1145/1869459.1869489}
  {\path{doi:10.1145/1869459.1869489}}.

\bibitem{oliveira_implicit_2012}
Bruno~C. d.~S.~Oliveira, Tom Schrijvers, Wontae Choi, Wonchan Lee, and
  Kwangkeun Yi.
\newblock The implicit calculus: a new foundation for generic programming.
\newblock In Jan Vitek, Haibo Lin, and Frank Tip, editors, {\em {ACM} {SIGPLAN}
  Conference on Programming Language Design and Implementation, {PLDI} 2012,
  Beijing, China - June 11 - 16, 2012}, pages 35--44. {ACM}, 2012.
\newblock \href {https://doi.org/10.1145/2254064.2254070}
  {\path{doi:10.1145/2254064.2254070}}.

\bibitem{de_moura_z3_2008}
Leonardo~Mendon{\c{c}}a de~Moura and Nikolaj Bj{\o}rner.
\newblock {Z3:} an efficient {SMT} solver.
\newblock In C.~R. Ramakrishnan and Jakob Rehof, editors, {\em Tools and
  Algorithms for the Construction and Analysis of Systems, 14th International
  Conference, {TACAS} 2008, Held as Part of the Joint European Conferences on
  Theory and Practice of Software, {ETAPS} 2008, Budapest, Hungary, March
  29-April 6, 2008. Proceedings}, volume 4963 of {\em Lecture Notes in Computer
  Science}, pages 337--340. Springer, 2008.
\newblock \href {https://doi.org/10.1007/978-3-540-78800-3\_24}
  {\path{doi:10.1007/978-3-540-78800-3\_24}}.

\bibitem{devriese_bright_2011}
Dominique Devriese and Frank Piessens.
\newblock On the bright side of type classes: instance arguments in agda.
\newblock In Manuel M.~T. Chakravarty, Zhenjiang Hu, and Olivier Danvy,
  editors, {\em Proceeding of the 16th {ACM} {SIGPLAN} international conference
  on Functional Programming, {ICFP} 2011, Tokyo, Japan, September 19-21, 2011},
  pages 143--155. {ACM}, 2011.
\newblock \href {https://doi.org/10.1145/2034773.2034796}
  {\path{doi:10.1145/2034773.2034796}}.

\bibitem{dunfield_complete_2013}
Jana Dunfield and Neelakantan~R. Krishnaswami.
\newblock Complete and easy bidirectional typechecking for higher-rank
  polymorphism.
\newblock In Greg Morrisett and Tarmo Uustalu, editors, {\em {ACM} {SIGPLAN}
  International Conference on Functional Programming, ICFP 2013, Boston, MA,
  {USA} - September 25 - 27, 2013}, pages 429--442. {ACM}, 2013.
\newblock \href {https://doi.org/10.1145/2500365.2500582}
  {\path{doi:10.1145/2500365.2500582}}.

\bibitem{emir_variance_2006}
Burak Emir, Andrew Kennedy, Claudio~V. Russo, and Dachuan Yu.
\newblock Variance and generalized constraints for c\({}^{\mbox{{\#}}}\)
  generics.
\newblock In Dave Thomas, editor, {\em {ECOOP} 2006 - Object-Oriented
  Programming, 20th European Conference, Nantes, France, July 3-7, 2006,
  Proceedings}, volume 4067 of {\em Lecture Notes in Computer Science}, pages
  279--303. Springer, 2006.
\newblock \href {https://doi.org/10.1007/11785477\_18}
  {\path{doi:10.1007/11785477\_18}}.

\bibitem{FournetCCS11}
C{\'{e}}dric Fournet, Markulf Kohlweiss, and Pierre{-}Yves Strub.
\newblock Modular code-based cryptographic verification.
\newblock In Yan Chen, George Danezis, and Vitaly Shmatikov, editors, {\em
  Proceedings of the 18th {ACM} Conference on Computer and Communications
  Security, {CCS} 2011, Chicago, Illinois, USA, October 17-21, 2011}, pages
  341--350. {ACM}, 2011.
\newblock \href {https://doi.org/10.1145/2046707.2046746}
  {\path{doi:10.1145/2046707.2046746}}.

\bibitem{SwamyOAKLAND11}
Arjun Guha, Matthew Fredrikson, Benjamin Livshits, and Nikhil Swamy.
\newblock Verified security for browser extensions.
\newblock In {\em 32nd {IEEE} Symposium on Security and Privacy, S{\&}P 2011,
  22-25 May 2011, Berkeley, California, {USA}}, pages 115--130. {IEEE} Computer
  Society, 2011.
\newblock \href {https://doi.org/10.1109/SP.2011.36}
  {\path{doi:10.1109/SP.2011.36}}.

\bibitem{liquidate}
Martin A.~T. Handley, Niki Vazou, and Graham Hutton.
\newblock Liquidate your assets: reasoning about resource usage in liquid
  haskell.
\newblock {\em Proceedings of the {ACM} on Programming Languages},
  4({POPL}):24:1--24:27, 2020.
\newblock \href {https://doi.org/10.1145/3371092} {\path{doi:10.1145/3371092}}.

\bibitem{heinlein_implicit_2006}
Christian Heinlein.
\newblock Implicit and dynamic parameters in {C++}.
\newblock In David~E. Lightfoot and Clemens~A. Szyperski, editors, {\em Modular
  Programming Languages, 7th Joint Modular Languages Conference, {JMLC} 2006,
  Oxford, UK, September 13-15, 2006, Proceedings}, volume 4228 of {\em Lecture
  Notes in Computer Science}, pages 37--56. Springer, 2006.
\newblock \href {https://doi.org/10.1007/11860990\_4}
  {\path{doi:10.1007/11860990\_4}}.

\bibitem{JhalaMcMillan06}
Ranjit Jhala and Kenneth~L. McMillan.
\newblock A practical and complete approach to predicate refinement.
\newblock In Holger Hermanns and Jens Palsberg, editors, {\em Tools and
  Algorithms for the Construction and Analysis of Systems, 12th International
  Conference, {TACAS} 2006 Held as Part of the Joint European Conferences on
  Theory and Practice of Software, {ETAPS} 2006, Vienna, Austria, March 25 -
  April 2, 2006, Proceedings}, volume 3920 of {\em Lecture Notes in Computer
  Science}, pages 459--473. Springer, 2006.
\newblock \href {https://doi.org/10.1007/11691372\_33}
  {\path{doi:10.1007/11691372\_33}}.

\bibitem{LiquidPLDI09}
Ming Kawaguchi, Patrick~Maxim Rondon, and Ranjit Jhala.
\newblock Type-based data structure verification.
\newblock In Michael Hind and Amer Diwan, editors, {\em Proceedings of the 2009
  {ACM} {SIGPLAN} Conference on Programming Language Design and Implementation,
  {PLDI} 2009, Dublin, Ireland, June 15-21, 2009}, pages 304--315. {ACM}, 2009.
\newblock \href {https://doi.org/10.1145/1542476.1542510}
  {\path{doi:10.1145/1542476.1542510}}.

\bibitem{Knowles07}
Kenneth Knowles and Cormac Flanagan.
\newblock Type reconstruction for general refinement types.
\newblock In {\em Programming Languages and Systems, 16th European Symposium on
  Programming, {ESOP} 2007, Held as Part of the Joint European Conferences on
  Theory and Practics of Software, {ETAPS} 2007, Braga, Portugal, March 24 -
  April 1, 2007, Proceedings}, pages 505--519, 2007.

\bibitem{leino_dafny_2010}
K.~Rustan~M. Leino.
\newblock Dafny: An automatic program verifier for functional correctness.
\newblock In Edmund~M. Clarke and Andrei Voronkov, editors, {\em Logic for
  Programming, Artificial Intelligence, and Reasoning - 16th International
  Conference, LPAR-16, Dakar, Senegal, April 25-May 1, 2010, Revised Selected
  Papers}, volume 6355 of {\em Lecture Notes in Computer Science}, pages
  348--370. Springer, 2010.
\newblock \href {https://doi.org/10.1007/978-3-642-17511-4\_20}
  {\path{doi:10.1007/978-3-642-17511-4\_20}}.

\bibitem{lewis_implicit_2000}
Jeffrey~R. Lewis, John Launchbury, Erik Meijer, and Mark Shields.
\newblock Implicit parameters: Dynamic scoping with static types.
\newblock In Mark~N. Wegman and Thomas~W. Reps, editors, {\em {POPL} 2000,
  Proceedings of the 27th {ACM} {SIGPLAN-SIGACT} Symposium on Principles of
  Programming Languages, Boston, Massachusetts, USA, January 19-21, 2000},
  pages 108--118. {ACM}, 2000.
\newblock \href {https://doi.org/10.1145/325694.325708}
  {\path{doi:10.1145/325694.325708}}.

\bibitem{fstar-implicits}
Tomer Libal.

\bibitem{mahboubi_canonical_2013}
Assia Mahboubi and Enrico Tassi.
\newblock Canonical {Structures} for the working {Coq} user, April 2013.
\newblock URL: \url{https://hal.inria.fr/hal-00816703/document}.

\bibitem{ChongOSDI14}
Scott Moore, Christos Dimoulas, Dan King, and Stephen Chong.
\newblock {SHILL:} {A} secure shell scripting language.
\newblock In Jason Flinn and Hank Levy, editors, {\em 11th {USENIX} Symposium
  on Operating Systems Design and Implementation, {OSDI} 2014, Broomfield, CO,
  USA, October 6-8, 2014}, pages 183--199. {USENIX} Association, 2014.
\newblock URL:
  \url{https://www.usenix.org/conference/osdi14/technical-sessions/presentation/moore}.

\bibitem{nanevski_ynot_2008}
Aleksandar Nanevski, Greg Morrisett, Avraham Shinnar, Paul Govereau, and Lars
  Birkedal.
\newblock Ynot: dependent types for imperative programs.
\newblock In James Hook and Peter Thiemann, editors, {\em Proceeding of the
  13th {ACM} {SIGPLAN} international conference on Functional programming,
  {ICFP} 2008, Victoria, BC, Canada, September 20-28, 2008}, pages 229--240.
  {ACM}, 2008.
\newblock \href {https://doi.org/10.1145/1411204.1411237}
  {\path{doi:10.1145/1411204.1411237}}.

\bibitem{norell_dependently_2009}
Ulf Norell.
\newblock Dependently typed programming in agda.
\newblock In Andrew Kennedy and Amal Ahmed, editors, {\em Proceedings of TLDI
  2009: 2009 {ACM} {SIGPLAN} International Workshop on Types in Languages
  Design and Implementation, Savannah, GA, USA, January 24, 2009}, pages 1--2.
  {ACM}, 2009.
\newblock \href {https://doi.org/10.1145/1481861.1481862}
  {\path{doi:10.1145/1481861.1481862}}.

\bibitem{OwickiGries}
Susan~S. Owicki and David Gries.
\newblock An axiomatic proof technique for parallel programs {I}.
\newblock {\em Acta Informatica}, 6:319--340, 1976.
\newblock \href {https://doi.org/10.1007/BF00268134}
  {\path{doi:10.1007/BF00268134}}.

\bibitem{pierce_local_2000}
Benjamin~C. Pierce and David~N. Turner.
\newblock Local type inference.
\newblock {\em {ACM} Transactions on Programming Language and Systems},
  22(1):1--44, 2000.
\newblock \href {https://doi.org/10.1145/345099.345100}
  {\path{doi:10.1145/345099.345100}}.

\bibitem{RFC0793}
Jon Postel.
\newblock Transmission control protocol.
\newblock STD~7, RFC Editor, September 1981.
\newblock \url{http://www.rfc-editor.org/rfc/rfc793.txt}.
\newblock URL: \url{http://www.rfc-editor.org/rfc/rfc793.txt}.

\bibitem{rondon_liquid_2008}
Patrick~Maxim Rondon, Ming Kawaguchi, and Ranjit Jhala.
\newblock Liquid types.
\newblock In Rajiv Gupta and Saman~P. Amarasinghe, editors, {\em Proceedings of
  the {ACM} {SIGPLAN} 2008 Conference on Programming Language Design and
  Implementation, Tucson, AZ, USA, June 7-13, 2008}, pages 159--169. {ACM},
  2008.
\newblock \href {https://doi.org/10.1145/1375581.1375602}
  {\path{doi:10.1145/1375581.1375602}}.

\bibitem{scalas_verifying_2019}
Alceste Scalas, Nobuko Yoshida, and Elias Benussi.
\newblock Verifying message-passing programs with dependent behavioural types.
\newblock In Kathryn~S. McKinley and Kathleen Fisher, editors, {\em Proceedings
  of the 40th {ACM} {SIGPLAN} Conference on Programming Language Design and
  Implementation, {PLDI} 2019, Phoenix, AZ, USA, June 22-26, 2019}, pages
  502--516. {ACM}, 2019.
\newblock \href {https://doi.org/10.1145/3314221.3322484}
  {\path{doi:10.1145/3314221.3322484}}.

\bibitem{fstar}
Nikhil Swamy, Juan Chen, C{\'{e}}dric Fournet, Pierre{-}Yves Strub, Karthikeyan
  Bhargavan, and Jean Yang.
\newblock Secure distributed programming with value-dependent types.
\newblock In Manuel M.~T. Chakravarty, Zhenjiang Hu, and Olivier Danvy,
  editors, {\em Proceeding of the 16th {ACM} {SIGPLAN} International Conference
  on Functional Programming}, pages 266--278. {ACM}, 2011.
\newblock \href {https://doi.org/10.1145/2034773.2034811}
  {\path{doi:10.1145/2034773.2034811}}.

\bibitem{swamy_dependent_2016}
Nikhil Swamy, Catalin Hritcu, Chantal Keller, Aseem Rastogi, Antoine
  Delignat{-}Lavaud, Simon Forest, Karthikeyan Bhargavan, C{\'{e}}dric Fournet,
  Pierre{-}Yves Strub, Markulf Kohlweiss, Jean~Karim Zinzindohoue, and
  Santiago~Zanella B{\'{e}}guelin.
\newblock Dependent types and multi-monadic effects in {F}.
\newblock In Rastislav Bod{\'{\i}}k and Rupak Majumdar, editors, {\em
  Proceedings of the 43rd Annual {ACM} {SIGPLAN-SIGACT} Symposium on Principles
  of Programming Languages, {POPL} 2016, St. Petersburg, FL, USA, January 20 -
  22, 2016}, pages 256--270. {ACM}, 2016.
\newblock \href {https://doi.org/10.1145/2837614.2837655}
  {\path{doi:10.1145/2837614.2837655}}.

\bibitem{Swamy:2013}
Nikhil Swamy, Joel Weinberger, Cole Schlesinger, Juan Chen, and Benjamin
  Livshits.
\newblock Verifying higher-order programs with the dijkstra monad.
\newblock pages 387--398, 2013.
\newblock \href {https://doi.org/10.1145/2491956.2491978}
  {\path{doi:10.1145/2491956.2491978}}.

\bibitem{the_coq_development_team_coq_2009}
{The Coq Development Team}.
\newblock The {Coq} {Reference} {Manual}.
\newblock page 522, 2009.

\bibitem{unno_automating_2013}
Hiroshi Unno, Tachio Terauchi, and Naoki Kobayashi.
\newblock Automating relatively complete verification of higher-order
  functional programs.
\newblock In Roberto Giacobazzi and Radhia Cousot, editors, {\em The 40th
  Annual {ACM} {SIGPLAN-SIGACT} Symposium on Principles of Programming
  Languages, {POPL} 2013, Rome, Italy - January 23 - 25, 2013}, pages 75--86.
  {ACM}, 2013.
\newblock \href {https://doi.org/10.1145/2429069.2429081}
  {\path{doi:10.1145/2429069.2429081}}.

\bibitem{vazou_bounded_2015}
Niki Vazou, Alexander Bakst, and Ranjit Jhala.
\newblock Bounded refinement types.
\newblock In Kathleen Fisher and John~H. Reppy, editors, {\em Proceedings of
  the 20th {ACM} {SIGPLAN} International Conference on Functional Programming,
  {ICFP} 2015, Vancouver, BC, Canada, September 1-3, 2015}, pages 48--61.
  {ACM}, 2015.
\newblock \href {https://doi.org/10.1145/2784731.2784745}
  {\path{doi:10.1145/2784731.2784745}}.

\bibitem{vazou_abstract_2013}
Niki Vazou, Patrick~Maxim Rondon, and Ranjit Jhala.
\newblock Abstract refinement types.
\newblock In Matthias Felleisen and Philippa Gardner, editors, {\em Programming
  Languages and Systems - 22nd European Symposium on Programming, {ESOP} 2013,
  Held as Part of the European Joint Conferences on Theory and Practice of
  Software, {ETAPS} 2013, Rome, Italy, March 16-24, 2013. Proceedings}, volume
  7792 of {\em Lecture Notes in Computer Science}, pages 209--228. Springer,
  2013.
\newblock \href {https://doi.org/10.1007/978-3-642-37036-6\_13}
  {\path{doi:10.1007/978-3-642-37036-6\_13}}.

\bibitem{vazou_refinement_2014}
Niki Vazou, Eric~L. Seidel, Ranjit Jhala, Dimitrios Vytiniotis, and Simon
  L.~Peyton Jones.
\newblock Refinement types for haskell.
\newblock In Johan Jeuring and Manuel M.~T. Chakravarty, editors, {\em
  Proceedings of the 19th {ACM} {SIGPLAN} International Conference on
  Functional programming, Gothenburg, Sweden, September 1-3, 2014}, pages
  269--282. {ACM}, 2014.
\newblock \href {https://doi.org/10.1145/2628136.2628161}
  {\path{doi:10.1145/2628136.2628161}}.

\bibitem{vazou_refinement_2018}
Niki Vazou, Anish Tondwalkar, Vikraman Choudhury, Ryan~G. Scott, Ryan~R.
  Newton, Philip Wadler, and Ranjit Jhala.
\newblock Refinement reflection: complete verification with {SMT}.
\newblock {\em Proceedings of the {ACM} on Programming Languages},
  2({POPL}):53:1--53:31, 2018.
\newblock \href {https://doi.org/10.1145/3158141} {\path{doi:10.1145/3158141}}.

\bibitem{wadler_how_1989}
Philip Wadler and Stephen Blott.
\newblock How to make ad-hoc polymorphism less ad-hoc.
\newblock In {\em Conference Record of the Sixteenth Annual {ACM} Symposium on
  Principles of Programming Languages, Austin, Texas, USA, January 11-13,
  1989}, pages 60--76. {ACM} Press, 1989.
\newblock \href {https://doi.org/10.1145/75277.75283}
  {\path{doi:10.1145/75277.75283}}.

\bibitem{white_modular_2015}
Leo White, Fr{\'{e}}d{\'{e}}ric Bour, and Jeremy Yallop.
\newblock Modular implicits.
\newblock 198:22--63, 2014.
\newblock \href {https://doi.org/10.4204/EPTCS.198.2}
  {\path{doi:10.4204/EPTCS.198.2}}.

\bibitem{xi_eliminating_1998}
Hongwei Xi and Frank Pfenning.
\newblock Eliminating array bound checking through dependent types.
\newblock In Jack~W. Davidson, Keith~D. Cooper, and A.~Michael Berman, editors,
  {\em Proceedings of the {ACM} {SIGPLAN} 1998 Conference on Programming
  Language Design and Implementation (PLDI), Montreal, Canada, June 17-19,
  1998}, pages 249--257. {ACM}, 1998.
\newblock \href {https://doi.org/10.1145/277650.277732}
  {\path{doi:10.1145/277650.277732}}.

\end{thebibliography}

\appendix
\section{Full Declarative Static Semantics}
\label{sec:rules:lang}

\judgementHead{Well-Formedness}{\isWellFormed{\Gamma}{t}}
\begin{mathpar}
\inferrule*[right=\wtBase]
  {\hasSort{\forgetreft{\forgetimplicits{\Gamma}},\fbd{v}{b}}{r}{\tbool}}
  {\isWellFormed{\Gamma}{\reftpv{v}{b}{r}}}

\inferrule*[right=\wtFun]
  {\isWellFormed{\Gamma}{t_x}
  \\ \isWellFormed{\Gamma,\fbd{x}{t_x}}{t}}
  {\isWellFormed{\Gamma}{\trfun{x}{t_x}{t}}}

\inferrule*[right=\wtVar]
  {\alpha \in \Gamma}
  {\isWellFormed{\Gamma}{\alpha}}

\inferrule*[right=\wtForall]
  {\isWellFormed{\Gamma,\alpha}{t}}
  {\isWellFormed{\Gamma}{\ttabs{\alpha}{t}}}

\inferrule*[right=\wtIFun]
  {\isWellFormed{\Gamma}{\reftpv{v}{b}{r}}
  \\ \isWellFormed{\Gamma,\fbd{x}{\reftpv{v}{b}{r}}}{t}}
  {\isWellFormed{\Gamma}{\trifun{x}{\reftpv{v}{b}{r}}{t}}}

\inferrule*[right=\wtSigma]
  {\isWellFormed{\Gamma}{\reftpv{v}{b}{r}}
  \\ \isWellFormed{\Gamma,\fbd{x}{\reftpv{v}{b}{r}}}{t}}
  {\isWellFormed{\Gamma}{\tcoifun{x}{\reftpv{v}{b}{r}}{t}}}
\end{mathpar}

\medskip \judgementHead{Subtyping}{\isSubType{\Gamma}{t_1}{t_2}}

\begin{mathpar}
\inferrule*[right=\tSubBase]
  {\embed{\Gamma}(\forall v_1:b.r_1 \h{\Rightarrow} \subst{r_2}{v_2}{v_1}) \text{ is valid}}
  {\isSubType{\Gamma}{\reftpv{v_1}{b}{r_1}}{\reftpv{v_2}{b}{r_2}}}

\inferrule*[right=\tSubFun]
  {\isSubType{\Gamma}{t_2}{t_1}
  \\ \isSubType{\Gamma, x_2:{t_2}}{\subst{t_1'}{x_1}{x_2}}{t_2'}}
  {\isSubType{\Gamma}{\trfun{x_1}{t_1}{t_1'}}{\trfun{x_2}{t_2}{t_2'}}}

\inferrule*[right=\tSubForall]
  {\isSubType{\Gamma,\alpha}{t}{\subst{t'}{\beta}{\alpha}}}
  {\isSubType{\Gamma}{\ttabs{\alpha}{t}}{\ttabs{\beta}{t'}}}

\inferrule*[right=\tSubFunIL]
  {\isSubType{\Gamma}{\subst{t_1'}{x}{e}}{t_2}
  \\ \hastype{\forgetimplicits{\Gamma}}{e}{t_1}}
  {\isSubType{\Gamma}{\trifun{x}{t_1}{t_1'}}{t_2}}

\inferrule*[right=\tSubFunIR]
  {\isSubType{\Gamma, \fbd{x}{t_2}}{t_1}{t_2'}}
  {\isSubType{\Gamma}{t_1}{\trifun{x}{t_2}{t_2'}}}

\inferrule*[right=\tSubSigmaL]
  {\isSubType{\Gamma, \fbd{x}{t_1}}{t_1'}{t_2}}
  {\isSubType{\Gamma}{\tcoifun{x}{t_1}{t_1'}}{t_2}}

\inferrule*[right=\tSubSigmaR]
  {\isSubType{\Gamma}{t_1}{\subst{t_2'}{x}{e}}
  \\ \hastype{\forgetimplicits{\Gamma}}{e}{t_2}}
  {\isSubType{\Gamma}{t_1}{\tcoifun{x}{t_2}{t_2'}}}

\inferrule*[right=\tSubAlpha]
  {~}
  {\isSubType{\Gamma}{\alpha}{\alpha}}
\end{mathpar}
\captionof{figure}{{Static Semantics of \s\lang}}
\label{fig:statics:complete:1}

\judgementHead{Type Checking}{\hastype{\Gamma}{e}{t}}
\begin{mathpar}
\mprset{sep=1em}

\inferrule*[right=\tBase]
  {\fbd{x}{\reftpv{v}{b}{r}} \in \Gamma
  \\ v' \ \text{fresh}}
  {\hastype{\Gamma}{x}{\reftpv{v'}{b}{\subst{r}{v}{v'} \land v' = x}}}

\inferrule*[right=\tVar]
  {\fbd{x}{t} \in \Gamma}
  {\hastype{\Gamma}{x}{t}} 

\inferrule*[right=\tCon]
  {~}
  {\hastype{\Gamma}{c}{\tc{c}}}

\inferrule*[right=\tTAbs]
  {\hastype{\Gamma,\alpha}{v}{t}}
  {\hastype{\Gamma}{\tabs{\alpha}{v}}{\ttabs{\alpha}{t}}}

\inferrule*[right=\tTApp]
  {\hastype{\Gamma}{e}{\ttabs{\alpha}{t}}
  \\ \isWellFormed{\Gamma}{t'}
  \\ \forgetreft{t'} = \tau}
  {\hastype{\Gamma}{\tyapp{e}{\tau}}{\subst{t}{\alpha}{t'}}}

\inferrule*[right=\tAbsReft]
  {\hastype{\Gamma, \fbd{x}{t_x}}{e}{t} 
  \\ \isWellFormed{\Gamma}{t_x}}
  {\hastype{\Gamma}{\elambda{x}{t_x}{e}}{\trfun{x}{t_x}{t}}}

\inferrule*[right=\tAbsBase]
  {\hastype{\Gamma, \fbd{x}{t_x}}{e}{t} 
  \\ \isWellFormed{\Gamma}{t_x}
  \\ \forgetreft{t_x} = \tau}
  {\hastype{\Gamma}{\elambda{x}{\tau}{e}}{\trfun{x}{t_x}{t}}}

\inferrule*[right=\tAbsi]
  {\hastype{\Gamma, \ebd{x}{t_x}}{e}{t}}
  {\hastype{\Gamma}{\ilambda{x}{t_x}{e}}{(\trifun{x}{t_x}{t})}}

\inferrule*[right=\tApp]
  {\hastype{\Gamma}{e_1}{t}
  \\ \appjudge{\Gamma}{t}{e_2}{t'}}
  {\hastype{\Gamma}{\app{e_1}{e_2}}{t'}}

\inferrule*[right=\tLetReft]
  {\hastype{\Gamma,\fbd{x}{t_x}}{e_x}{t_x}
  \\ \hastype{\Gamma,\fbd{x}{t_x}}{e}{t}
  \\ \isWellFormed{\Gamma}{t_x}
  \\ \isWellFormed{\Gamma}{t}}
  {\hastype{\Gamma}{\eletin{\tb{x}{t_x}}{e_x}{e}}{t}}

\inferrule*[right=\tLetBase]
  {\hastype{\Gamma}{e_x}{t_x}
  \\ \hastype{\Gamma,\fbd{x}{t_x}}{e}{t}
  \\ \isWellFormed{\Gamma}{t}
  \\ \isWellFormed{\Gamma}{t_x}
  \\ \forgetreft{t_x} = \tau}
  {\hastype{\Gamma}{\eletin{\tb{x}{\tau}}{e_x}{e}}{t}}

\inferrule*[right=\tUnpack]
  {\hastype{\Gamma}{e_1}{\tcoifun{x'}{t_1}{t_2}}
  \\ \hastype{\Gamma,\ebd{x}{t_1},\fbd{y}{\subst{t_2}{x'}{x}}}{e_2}{t}
  \\ \isWellFormed{\Gamma}{t}}
  {\hastype{\Gamma}{\eunpack{x}{y}{e_1}{e_2}}{t}}

\inferrule*[right=\tSub]
  {\hastype{\Gamma}{e}{t_2}
  \\ \isSubType{\Gamma}{t_2}{t_1}
  \\ \isWellFormed{\Gamma}{t_1}}
  {\hastype{\Gamma}{e}{t_1}}
\end{mathpar}

\medskip \judgementHead{Application Checking}{$\appjudge{\Gamma}{t_1}{e}{t_2}$}
\begin{mathpar}
\mprset{sep=1em}
\inferrule*[right=\tAppFunVar]
  {\hastype{\Gamma}{y}{t_x}}
  {\appjudge{\Gamma}{\trfun{x}{t_x}{t}}{y}{\subst{t}{x}{y}}}

\inferrule*[right=\tAppFun]
  {\hastype{\Gamma}{e}{t_e}
  \\ \isSubType{\Gamma}{t_e}{t_x}
  \\ \isSubType{\Gamma, \fbd{y}{t_e}}{\subst{t}{x}{y}}{t'}
  \\ \isWellFormed{\Gamma}{t'}
  \\ y \ \text{fresh}}
  {\appjudge{\Gamma}{\trfun{x}{t_x}{t}}{e}{t'}}

\inferrule*[right=\tAppIFun]
  {\appjudge{\Gamma}{\subst{t}{x}{e'}}{e}{t'}
  \\ \hastype{\forgetimplicits{\Gamma}}{e'}{t_x}}
  {\appjudge{\Gamma}{\trifun{x}{t_x}{t}}{e}{t'}}
\end{mathpar}
\captionof{figure}{{Static Semantics of \s\lang}}
\label{fig:statics:complete:2}

\begin{center}
\begin{minipage}[t]{0.25\linewidth}
  \begin{align*}
    \forgetimplicits{\bullet} &\doteq \bullet \\
    \forgetimplicits{\Gamma,\fbd{x}{t}} &\doteq \forgetimplicits{\Gamma}, \fbd{x}{t} \\
    \forgetimplicits{\Gamma,\ebd{x}{t}} &\doteq \forgetimplicits{\Gamma}, \fbd{x}{t} \\
  \end{align*}
\end{minipage}
\begin{minipage}[t]{0.25\linewidth}
  \begin{align*}
    \forgetreft{\reftpv{x}{b}{r}} &\doteq b \\
    \forgetreft{\alpha} &\doteq \alpha \\
    \forgetreft{\tpoly{\alpha}{t}} &\doteq \tpoly{\alpha}{\forgetreft{t}} \\
    \forgetreft{\trfun{x}{t_1}{t_2}} &\doteq \tfun{\forgetreft{t_1}}{\forgetreft{t_2}} \\
    \forgetreft{\trifun{x}{t_1}{t_2}} &\doteq \forgetreft{t_2} \\
    \forgetreft{\tcoifun{x}{t_1}{t_2}} &\doteq \forgetreft{t_2}
  \end{align*}
\end{minipage}
\end{center}
\captionof{figure}{Metafunctions for Declarative Semantics}

  \begin{tabular}{>{$}r<{$} >{$}c<{$} >{$}l<{$}}
    \multicolumn{3}{l}{\textbf{Types}} \\
    b & \bnfdef & \tint \spmid \tbool \spmid \cdots \\
    t & \bnfdef & \reftpv{x}{b}{r} \spmid \trfun{x}{t}{t} \spmid \trifun{x}{t}{t} \spmid \tcoifun{x}{t}{t} \spmid \alpha \spmid \ttabs{\alpha}{t} \\

    \multicolumn{3}{l}{\textbf{Terms}} \\

    c & \bnfdef & \tfalse, \ttrue \spmid 0, 1, \ldots \spmid \wedge, \vee, +, -, =, \leq, \ldots \\

    e & \bnfdef & c \spmid x \spmid \elambda{x}{t}{e} \spmid \app{e}{e} \spmid \eletin{\tb{x}{t}}{e}{e} \spmid \tabs{\alpha}{e} \spmid \tyapp{e}{\tau} \spmid \\
    & & \ilambda{x}{t}{e} \spmid \eunpack{x}{y}{e}{e} \\

    \multicolumn{3}{l}{\textbf{Contexts}} \\
    \Gamma & \bnfdef & \bullet \spmid \Gamma, \fbd{x}{t} \spmid \Gamma, \ebd{x}{t}
  \end{tabular}
  \captionof{figure}{{Syntax of \lang}}
  \label{fig:syntax:lang:complete}

\section{Algorithmic Typing Rules}
\label{sec:rules:cgen}

\judgementHead{Subtyping}{\subtyping{\Gamma}{t_1}{t_2}{c}}

\begin{mathpar}
\inferrule*[right=\aSubBase, vcenter]
  {~}
  {\subtyping{\Gamma}{\reftpv{x_1}{b}{r_1}}{\reftpv{x_2}{b}{r_2}}{\forall x_1:\tau. r_1 \Rightarrow \subst{r_2}{x_2}{x_1}}}

\inferrule
  {\subtyping{\Gamma, \alpha}{t_1}{\subst{t_2}{\beta}{\alpha}}{c}}
  {\subtyping{\Gamma}{\tpoly{\alpha}{t_1}}{\tpoly{\beta}{t_2}}{c}}

\inferrule
  {~}
  {\subtyping{\Gamma}{\alpha}{\alpha}{\true{}}}

\inferrule
  {\subtyping{\Gamma}{t_2}{t_1}{c}
  \\ \subtyping{\Gamma,\fbd{x_2}{t_2}}{\subst{t_1'}{x_1}{x_2}}{t_2'}{c'}}
  {\subtyping{\Gamma}{\trfun{x_1}{t_1}{t_1'}}{\trfun{x_2}{t_2}{t_2'}}{c \wedge {} (\genimp{x_2}{t_2}{c'})}}

\inferrule
  {\subtyping{\Gamma, \fbd{z}{t_1}}{\subst{t_1'}{x}{z}}{t_2}{c}
  \\ z \text{ fresh}}
  {\subtyping{\Gamma}{\trifun{x}{t_1}{t_1'}}{t_2}{\genexists{z}{t_1}{c}}}

\inferrule
  {\subtyping{\Gamma, \fbd{x}{t_2}}{t_1}{t_2'}{c}}
  {\subtyping{\Gamma}{t_1}{\trifun{x}{t_2}{t_2'}}{\genimp{x}{t_2}{c}}}

\inferrule
  {\subtyping{\Gamma, \fbd{x}{t_1}}{t_1'}{t_2}{c}}
  {\subtyping{\Gamma}{\tcoifun{x}{t_1}{t_1'}}{t_2}{\genimp{x}{t_1}{c}}}

\inferrule
  {\subtyping{\Gamma, \fbd{z}{t_2}}{t_1}{\subst{t_2'}{x}{z}}{c}
  \\ z \text{ fresh}}
  {\subtyping{\Gamma}{t_1}{\tcoifun{x}{t_2}{t_2'}}{\genexists{z}{t_2}{c}}}
\end{mathpar}
\captionof{figure}{Algorithmic Typing Rules}

\judgementHead{Checking}{\checking{\Gamma}{e}{t}{c}}

\begin{mathpar}
\inferrule*[right=\cSub, vcenter]
  {\synth{\Gamma}{e}{t'}{c}
  \\ \subtyping{\Gamma}{t'}{t}{c'}}
  {\checking{\Gamma}{e}{t}{c \wedge c'}}

\inferrule*[right=\cLetBase, vcenter]
  {\checking{\Gamma}{e_1}{\hat{t}}{c_1}
  \\ \checking{\Gamma,\fbd{x}{\hat{t}}}{e_2}{t}{c_2}
  \\ \hat{t} = \fresh{\Gamma}{\tau}}
  {\checking{\Gamma}{\eletin{\tb{x}{\tau}}{e_1}{e_2}}{t}{c_1 \wedge (\genimp{x}{\hat{t}}{c_2})}}

\inferrule*[right=\cLetReft, vcenter]
  {\checking{\Gamma,\fbd{x}{t_x}}{e_1}{t_x}{c_1}
  \\ \checking{\Gamma,\fbd{x}{t_x}}{e_2}{t}{c_2}}
  {\checking{\Gamma}{\eletin{\tb{x}{t_x}}{e_1}{e_2}}{t}{\genimp{x}{t_x}{(c_1 \wedge c_2)}}}

\inferrule
  {\synth{\Gamma}{e_1}{t'}{c_1}
  \\ \spinechecking{\Gamma}{t'}{e_2}{t}{c_2}}
  {\checking{\Gamma}{\app{e_1}{e_2}}{t}{c_1 \wedge c_2}}

\inferrule
  {\checking{\Gamma,\beta}{\subst{v}{\alpha}{\beta}}{t}{c}}
  {\checking{\Gamma}{\tabs{\alpha}{v}}{\ttabs{\beta}{t}}{c}}

\inferrule
  {\synth{\Gamma}{e_1}{\tcoifun{x'}{t_1}{t_2}}{c_1}
  \\ \checking{\Gamma,\ebd{x}{t_1},\fbd{y}{\subst{t_2}{x'}{x}}}{e_2}{t}{c_2}}
{\checking{\Gamma}{\eunpack{x}{y}{e_1}{e_2}}{t}{c_1 \wedge (\genimp{x}{t_1}{(\genimp{y}{\subst{t_2}{x'}{x}}{c_2})})}}
\end{mathpar}

\judgementHead{Synthesis}{\synth{\Gamma}{e}{t}{c}}

\begin{mathpar}
\inferrule
  {\synth{\Gamma,\alpha}{e}{t}{c}}
  {\synth{\Gamma}{\tabs{\alpha}{e}}{\tpoly{\alpha}{t}}{c}}

\inferrule
  {\synth{\Gamma}{e}{\tpoly{\alpha}{t}}{c}
  \\ \hat{t} = \fresh{\Gamma}{\tau}}
  {\synth{\Gamma}{\tyapp{e}{\tau}}{\subst{t}{\alpha}{\hat{t}}}{c}}

\inferrule
  {~}
  {\synth{\Gamma}{c}{\tc{c}}{\true{}}}

\inferrule
  {\fbd{x}{\reftpv{y}{b}{r}} \in \Gamma
  \\ v \ \text{fresh}}
  {\synth{\Gamma}{x}{\reftpv{v}{b}{\subst{r}{y}{v} \wedge v = x}}{\true{}}}

\inferrule
  {\fbd{x}{t} \in \Gamma}
  {\synth{\Gamma}{x}{t}{\true{}}}

\inferrule
  {\synth{\Gamma,\fbd{x}{\hat{t}}}{e}{t}{c}
  \\ \hat{t} = \fresh{\Gamma}{\tau}}
  {\synth{\Gamma}{\elambda{x}{\tau}{e}}{\trfun{x}{\hat{t}}{t}}{\genimp{x}{\hat{t}}{c}}}

\inferrule
  {\synth{\Gamma,\fbd{x}{t_x}}{e}{t}{c}}
  {\synth{\Gamma}{\elambda{x}{t_x}{e}}{\trfun{x}{t_x}{t}}{\genimp{x}{t_x}{c}}}

\inferrule
  {\synth{\Gamma}{e_1}{t_1}{c_1}
  \\ \spinesynth{\Gamma}{t_1}{e_2}{t_2}{c_2}}
  {\synth{\Gamma}{\app{e_1}{e_2}}{t_2}{c_1 \wedge c_2}}

\inferrule
  {\synth{\Gamma,\ebd{x}{t_x}}{e}{t}{c}}
  {\synth{\Gamma}{\ilambda{x}{t_x}{e}}{\trifun{x}{t_x}{t}}{\genimp{x}{t_x}{c}}}

\inferrule
  {\synth{\Gamma}{e_1}{\tcoifun{x'}{t_1}{t_2}}{c_1}
  \\ \synth{\Gamma,\ebd{x}{t_1},\fbd{y}{\subst{t_2}{x'}{x}}}{e_2}{t}{c_2}
  \\ \subtyping{\Gamma,\ebd{x}{t_1},\fbd{y}{\subst{t_2}{x'}{x}}}{t}{\hat{t}}{c_3}
  \\ \hat{t} = \fresh{\Gamma}{t}}
  {\synth{\Gamma}{\eunpack{x}{y}{e_1}{e_2}}{\hat{t}}{c_1 \wedge (\genimp{x}{t_1}{(\genimp{y}{\subst{t_2}{x'}{x}}{(c_2 \wedge c_3)})})}}

\inferrule*[right=\sLetBase, vcenter]
  {\checking{\Gamma}{e_1}{\hat{t}}{c_1}
  \\ \synth{\Gamma,\fbd{x}{\hat{t}}}{e_2}{t'}{c_2}
  \\ \subtyping{\Gamma,\fbd{x}{\hat{t}}}{t'}{\hat{t}'}{c_3}
  \\ \hat{t} = \fresh{\Gamma}{\tau}
  \\ \hat{t}' = \fresh{\Gamma}{t'}}
  {\synth{\Gamma}{\eletin{\tb{x}{\tau}}{e_1}{e_2}}{\hat{t}'}{c_1 \wedge (\genimp{x}{\hat{t}}{c_2 \wedge c_3})}}

\inferrule*[right=\sLetReft, vcenter]
  {\checking{\Gamma,\fbd{x}{t_x}}{e_1}{t_x}{c_1}
  \\ \synth{\Gamma,\fbd{x}{t_x}}{e_2}{t}{c_2}
  \\ \subtyping{\Gamma,\fbd{x}{t_x}}{t}{\hat{t}}{c_3}
  \\ \hat{t} = \fresh{\Gamma}{t}}
{\synth{\Gamma}{\eletin{\tb{x}{t_x}}{e_1}{e_2}}{\hat{t}}{\genimp{x}{t_x}{(c_1 \wedge c_2 \wedge c_3)}}}
\end{mathpar}
\captionof{figure}{Algorithmic Typing Rules}

\judgementHead{Application Checking}{\spinechecking{\Gamma}{t}{e}{t'}{c}}
\begin{mathpar}
\inferrule*[right=\cAppFun, vcenter]
  {\checking{\Gamma}{y}{t_x}{c_1}
  \\ \subtyping{\Gamma}{\subst{t}{x}{y}}{t'}{c_2}}
  {\spinechecking{\Gamma}{\trfun{x}{t_x}{t}}{y}{t'}{c_1 \wedge c_2}}

\inferrule*[right=\cAppIFun, vcenter]
  {\spinechecking{\Gamma,\ebd{z}{t_x}}{\subst{t}{x}{z}}{e}{t'}{c}
  \\ z \text{ fresh}}
  {\spinechecking{\Gamma}{\trifun{x}{t_x}{t}}{e}{t'}{\genexists{z}{t_x}{c}}}

\inferrule
  {\synth{\Gamma}{e}{t_e}{c_1}
  \\ \subtyping{\Gamma}{t_e}{t_x}{c_2}
  \\ \subtyping{\Gamma,\ebd{y}{t_e}}{\subst{t}{x}{y}}{t'}{c_3}
  \\ y \ \text{fresh}}
  {\spinechecking{\Gamma}{\trfun{x}{t_x}{t}}{e}{t'}{c_1 \wedge c_2 \wedge (\genimp{y}{t_e}{c_3})}}
\end{mathpar}

\judgementHead{Application Synthesis}{\spinesynth{\Gamma}{t}{e}{t'}{c}}

\begin{mathpar}
\inferrule*[right=\sAppFun, vcenter]
  {\checking{\Gamma}{y}{t_x}{c}}
  {\spinesynth{\Gamma}{\trfun{x}{t_x}{t}}{y}{\subst{t}{x}{y}}{c}}

\inferrule
  {\synth{\Gamma}{e}{t_e}{c_1}
  \\ \subtyping{\Gamma}{t_e}{t_x}{c_2}
  \\ \subtyping{\Gamma,\ebd{y}{t_e}}{\subst{t}{x}{y}}{\hat{t}}{c_3}
  \\ \hat{t} = \fresh{\Gamma}{t}
  \\ y \ \text{fresh}}
  {\spinesynth{\Gamma}{\trfun{x}{t_x}{t}}{e}{\hat{t}}{c_1 \wedge c_2 \wedge (\genimp{y}{t_e}{c_3})}}

\inferrule*[right=\sAppIFun, vcenter]
  {\spinesynth{\Gamma,\ebd{z}{t_x}}{\subst{t}{x}{z}}{e}{t'}{c_1}
  \\ \subtyping{\Gamma,\ebd{z}{t_x}}{t'}{\hat{t}}{c_2}
  \\ \hat{t} = \fresh{\Gamma}{t'}
  \\ z \text{ fresh}}
  {\spinesynth{\Gamma}{\trifun{x}{t_x}{t}}{e}{\hat{t}}{\genexists{z}{t_x}{(c_1 \wedge c_2)}}}
\end{mathpar}
\captionof{figure}{Algorithmic Typing Rules: Spines}

  \[\begin{array}{lcl}
    \genimp{x}{\reftpv{x}{b}{r}}{c} & \doteq & \forall \tb{x}{b}.r \Rightarrow c \\
    \genimp{x}{t}{c} & \doteq & c \\

    \\

    \genexists{x}{\reftpv{x}{b}{r}}{c} & \doteq & \exists \tb{x}{b}.(r \wedge c) \\
    \genexists{x}{t}{c} & \doteq & c \\

    \\
    
    \fresh{[\fbd{x_1}{t_1},\ldots]}{\reftpv{x}{b}{r}} & \doteq & \reftpv{v}{b}{\kappa(x_1,\ldots,v)} \\
    \quad \text{where } \kappa, v \text{ fresh names} \\
    \fresh{\Gamma}{\tpoly{\alpha}{t}} & \doteq & \tpoly{\alpha}{\fresh{\Gamma}{t}} \\
    \fresh{\Gamma}{\alpha} & \doteq & \alpha \\
    \fresh{\Gamma}{\trfun{x}{t_1}{t_2}} & \doteq & \trfun{x'}{t_1'}{t_2'} \\
    \quad \text{where } \begin{aligned}[t]
      & x' \text{ fresh name} \\
      & t_1' = {\fresh{\Gamma}{t_1}} \\
      & t_2' = {\fresh{(\Gamma,\fbd{x'}{t_1'})}{t_2}}
    \end{aligned} \\
    \fresh{\Gamma}{\trifun{x}{t_1}{t_2}} & \doteq & \trifun{x'}{t_1'}{t_2'} \\
    \quad \text{where } \begin{aligned}[t]
      & x' \text{ fresh name} \\
      & t_1' = {\fresh{\Gamma}{t_1}} \\
      & t_2' = {\fresh{(\Gamma,\fbd{x'}{t_1'})}{t_2}}
    \end{aligned} \\
    \fresh{\Gamma}{\tcoifun{x}{t_1}{t_2}} & \doteq & \tcoifun{x'}{t_1'}{t_2'} \\
    \quad \text{where } \begin{aligned}[t]
      & x' \text{ fresh name} \\
      & t_1' = {\fresh{\Gamma}{t_1}} \\
      & t_2' = {\fresh{(\Gamma,\fbd{x'}{t_1'})}{t_2}}
    \end{aligned} \\
  \end{array}\]
  \captionof{figure}{Metafunctions for Algorithmic Typing Rules}
  \label{fig:cgen:complete}

\section{Refinement Predicate Logic Definitions}
\label{sec:rules:logic}

\judgementHeadNameOnly{Syntax of \smtlang}

\begin{tabular}{r >{$}r<{$} >{$}r<{$} >{$}l<{$} >{$}l<{$}}
\textit{Predicates} & r & \bnfdef & \text{\ldots varies \ldots} \\
\textit{Types} & \tau & \bnfdef & \text{\ldots varies \ldots} \\
\textit{Propositions} & p & \bnfdef & \applykvar{\kappa}{\overline{x}} \spmid r \\
\textit{Existential Horn Clauses} & c & \bnfdef & \exists x : \tau. c \spmid \forall x: \tau . p \Rightarrow c \spmid c \wedge c \spmid p\\
\textit{First Order Assignments} & \Psi & \bnfdef & \bullet \spmid \Psi, \fobind{r}{x}\\
\textit{Second Order Assignments} & \Delta & \bnfdef & \bullet \spmid \Delta, \sobind{\lambda \overline{x}.r}{\kappa}\\
\end{tabular}

\judgementHead{Semantics of \smtlang}{$\valid{\Delta}{\Psi}{c}$}

\begin{mathpar}
\inferrule
  {\valid{\Delta}{\Psi}{\fosubst{c}{x}{r}}}
  {\valid{\Delta}{\Psi, \fobind{r}{x}}{c}}

\inferrule
  {\valid{\Delta}{\bullet}{\sosubst{c}{\kappa}{\lambda \overline{x}.r}}}
  {\valid{\Delta,\sobind{\lambda \overline{x}.r}{\kappa}}{\bullet}{c}}

\inferrule
  {\evars{c} = \emptyset
  \\ \kvars{c} = \emptyset
  \\ \smtvalid{c}}
  {\valid{\bullet}{\bullet}{c}}
\end{mathpar}

\judgementHead{Denotation from \lang to \smtlang}{\embed{\cdot}}
\begin{align*}
  \embed{\Gamma,\fbd{x}{\reftpv{v}{b}{r}}}(c) &\doteq \embed{\Gamma}(\forall x:b. \embed{\subst{r}{v}{x}} \Rightarrow c)\\
  \embed{\Gamma,\ebd{x}{\reftpv{v}{b}{r}}}(c) &\doteq \embed{\Gamma}(\forall x:b. \embed{\subst{r}{v}{x}} \Rightarrow c)\\
  \embed{\Gamma,\_ \,}(c) &\doteq \embed{\Gamma}(c)\\
  \embed{\bullet}(c) &\doteq c\\
\end{align*}

\judgementHeadNameOnly{Substitutions}
\begin{align*}
  \fosubst{(\exists x:\tau.c)}{x}{r} &\doteq \subst{c}{x}{r}\\
  \fosubst{(\exists y:\tau.c)}{x}{r} &\doteq \exists y:\tau. \fosubst{c}{x}{r}\\
  \fosubst{(\forall y:\tau.p \Rightarrow c)}{x}{r} &\doteq \forall y:\tau.p \Rightarrow \fosubst{c}{x}{r}\\
  \fosubst{(c_1 \wedge c_2)}{x}{r} &\doteq \fosubst{c_1}{x}{r} \wedge \fosubst{c_2}{x}{r}\\
  \fosubst{p}{x}{r} &\doteq p\\
  \\
  \sosubst{p}{\kappa}{\lambda \overline{x}.r} &\doteq p\\
  \sosubst{\applykvar{\kappa}{\overline{y}}}{\kappa}{\lambda \overline{x}.r} &\doteq r\overline{[y/x]}\\
  \sosubst{(\exists y:\tau.c)}{\kappa}{\lambda \overline{x}.r} &\doteq \exists y:\tau. \sosubst{c}{\kappa}{\lambda \overline{x}.r} \\
  \sosubst{(\forall y:\tau.p \Rightarrow c)}{\kappa}{\lambda \overline{x}.r} &\doteq \forall y:\tau. \sosubst{p}{\kappa}{\lambda \overline{x}.r} \Rightarrow \sosubst{c}{\kappa}{\lambda \overline{x}.r}\\
  \sosubst{(c_1 \wedge c_2)}{\kappa}{\lambda \overline{x}.r} &\doteq \sosubst{c_1}{\kappa}{\lambda \overline{x}.r} \wedge \sosubst{c_2}{\kappa}{\lambda \overline{x}.r}\\
\end{align*}

\judgementHeadNameOnly{Shorthand}

\[
  \begin{array}{rcl}
    \vDash c & \defeq & \exists \Delta, \Psi.~\valid{\Delta}{\Psi}{c}\\
    \applyfo{\fobind{r_1}{x_1}, \ldots, \fobind{r_n}{x_n}, \bullet}{c} & \defeq & \fosubst{c}{x_1}{r_1} \cdots \{r_n/x_n\}\\
    \applyso{\sobind{\lambda \overline{x}.r_1}{\kappa_1}, \ldots, \sobind{\lambda \overline{x}.r_n}{\kappa_n}, \bullet}{c} & \defeq & \sosubst{c}{\kappa_1}{\lambda \overline{x}.r_1} \cdots [\lambda \overline{x}.r_n/\kappa_n]\\
  \end{array}
\]
\captionof{figure}{{Syntax and Static Semantics of \smtlang}}
\label{fig:logic:complete}
\section{Constraint Solving}
\label{sec:solving:complete}

  \[\begin{array}{lcl}
    \toprule
    \pokesym & : & \mathbbm{2}^X \times C \rightarrow C\\
    \midrule
    \poke{\overline{y}}{\exists n:\tau.c} & \doteq &
    (\forall n:\tau. \pi_{n}(n, \overline{y}) \Rightarrow \poke{\{\overline{y},n\}}{c})  \\
                                              && \, \wedge\, \exists n.\pi_{n}(n,\overline{y}\})\\
    \poke{\overline{y}}{\forall x:\tau.p \Rightarrow c} & \doteq &
    \forall x:\tau.p \Rightarrow \poke{\{x,\overline{y}\}}{c}\\
    \poke{\overline{y}}{c_1 \wedge c_2} & \doteq & \poke{\overline{y}}{c_1} \wedge \poke{\overline{y}}{c_2}\\
    \poke{\overline{y}}{p} & \doteq & p\\

    \midrule
    \nosidesym & : & C \rightarrow C\\
    \midrule
    \noside{\exists n:\tau.c} & \doteq & \true{}\\

    \noside{\forall x:\tau.p \Rightarrow c} & \doteq &
    \forall x:\tau.p \Rightarrow \noside{c}\\

    \noside{c_1 \wedge c_2}
      & \doteq & \noside{c_1} \wedge \noside{c_2} \\

    \noside{p} & \doteq & p\\

    \midrule
    \sidesym & : & C \rightarrow C\\
    \midrule
    \sides{\exists n:\tau.c} & \doteq & \exists n:\tau.c\\

    \sides{\forall x:\tau.p \Rightarrow c} & \doteq &
    \forall x:\tau.p \Rightarrow \sides{c}\\

    \sides{c_1 \wedge c_2}
      & \doteq & \sides{c_1} \wedge \sides{c_2} \\

    \sides{p} & \doteq & \true{}\\
    \bottomrule
  \end{array}\]
  \label{fig:complete:pokestar}
  \captionof{figure}{Transformation from EHC to NNF with existential side conditions}

  \[\begin{array}{lcl}
    \toprule
    \solksym & : & K \times C \rightarrow C\\
    \midrule
    \solk{\kappa}{\forall x:\tau.p \Rightarrow c} & \doteq &
      {\exists} x:\tau.p \wedge \solk{\kappa}{c}\\
    \solk{\kappa}{c_1 \wedge c_2}
      & \doteq & \solk{\kappa}{c_1} \vee \solk{\kappa}{c_2} \\
    \solk{\kappa}{\kappa(\overline{y})} & \doteq & \bigwedge_i x_i = y_i \\ \quad \mbox{where }\overline{x} =\text{params}(\kappa) \\
    \solk{\kappa}{p} & \doteq & \false{}\\

    \midrule
    \defconstrsym & : & \Pi \times C \rightarrow C\\
    \midrule
    \defconstr{\pi_n}{\forall n : \tau.\pi_n(n,\overline{x}) \Rightarrow c} & \doteq & c \\
    \defconstr{\pi_n}{\forall x : \tau.p \Rightarrow c} & \doteq & \defconstr{\pi_n}{c}\\
    \defconstr{\pi_n}{c_1 \wedge c_2} & \doteq & \defconstr{\pi_n}{c_1}\wedge \defconstr{\pi_n}{c_2}\\
    \defconstr{\pi_n}{p} & \doteq & \true\\
    \bottomrule
  \end{array}\]
  \captionof{figure}{Strongest solution for $\kappa$ and weakest solution for $\pi$}

\begin{figure}[H]
  \[\begin{array}{lcl}
    \toprule
    \scopesym & : & K \times C \rightarrow C\\
    \midrule
    \scoped{\kappa}{c_1 \wedge c_2}\\
    \quad \mid \kappa \in c_1, \kappa \notin c_2 & \doteq & \scoped{\kappa}{c_1}\\
    \quad \mid \kappa \notin c_1, \kappa \in c_2 & \doteq & \scoped{\kappa}{c_2}\\
    \scoped{\kappa}{\forall x:\tau.p \Rightarrow c}\\
    \quad \mid \kappa \notin p & \doteq & \forall x:\tau.p \Rightarrow \scoped{\kappa}{c}\\
    \scoped{\kappa}{c} & \doteq & c\\
    \bottomrule
  \end{array}\]
  \caption{The scope of a \(\kappa\) in a Constraint}
\end{figure}

\begin{figure}[H]
  \[\begin{array}{lcl}
    \toprule
    \elimksym & : & K \times C \rightarrow C \\
    \midrule
    \elimk{\kappa}{c}{\sigma} & \doteq & \elimsol{\rho}{c}\\
    \quad \mbox{where}\\
    \qquad \forall(\overline{x_i:p_i}) \Rightarrow c' & = & \scoped{\kappa}{c}\\
    \qquad \rho & = & \kappa \mapsto \lambda\overline{x}.\solk{\kappa}{\noside{c'}}\\
    \qquad \overline{x} & = & \args{\kappa}\\

    \midrule
    \elimsolsym & : & (\overline{X} \rightarrow C)^K \times C \rightarrow C \\
    \midrule
    \elimsol{\rho}{\forall x:\tau.\kappa(\overline{y}) \Rightarrow c}\\
    \quad \mid \kappa \in \text{domain}(\rho) & \doteq &
    \demorgan{x:\tau}{\rho(\kappa)(\overline{y})}{\elimsol{\rho}{c}}\\
    \elimsol{\rho}{\forall x:\tau.p \Rightarrow c} & \doteq &
      \forall x:\tau. p \Rightarrow \elimsol{\rho}{c}\\
    \elimsol{\rho}{c_1 \wedge c_2} & \doteq &
      \elimsol{\rho}{c_1} \wedge \elimsol{\rho}{c_2}\\
    \elimsol{\rho}{\kappa(\overline{y})}\\
    \quad \mid \kappa \in \text{domain}(\rho) & \doteq & \true{}\\
    \elimsol{\rho}{p} & \doteq & p\\
    \elimsol{\rho}{\exists x:\tau.c} & \doteq & \exists x:\tau.c\\

    \midrule
    \elimssym & : & \overline{K} \times C \rightarrow C \\
    \midrule
    \elims{[]}{c} & \doteq & c\\
    \elims{\kappa:\overline{\kappa}}{c} & \doteq & \elims{\overline{\kappa}}{\elimk{\kappa}{c}{meaningless}}\\

    \midrule
    \demorgansym & : & (X : \tau) \times C \times C \rightarrow C \\
    \midrule
    \demorgan{x:\tau}{{\exists}y:\tau'.p \wedge c}{c'} & \doteq &
      \forall y:\tau'.p \Rightarrow \demorgan{x:\tau}{c}{c'}\\
    \demorgan{x:\tau}{c_1 \vee c_2}{c'} & \doteq &
    \demorgan{x:\tau}{c_1}{c'} \\ && \, \wedge\, \demorgan{x:\tau}{c_1}{c'}\\
    \demorgan{x:\tau}{p}{c'} & \doteq &
      \forall x:\tau.~p \Rightarrow c'\\
    \bottomrule
  \end{array}\]
  \caption{Eliminating \(\kappa\) Variables}
\end{figure}

\begin{figure}[H]
  \[\begin{array}{lcl}
    \toprule
    \elimEsym & : & P^C \rightarrow C \rightarrow C\\
    \midrule
    \elimE{qe}{\forall x:\tau.p \Rightarrow c} & \doteq & \forall x:\tau \Rightarrow \elimE{qe}{c}\\
    \elimE{qe}{\exists x:\tau.c} & \doteq & qe(\exists x:\tau.c)\\
    \elimE{qe}{c_1 \wedge c_2} & \doteq & \elimE{qe}{c_1} \wedge \elimE{qe}{c_2}\\

    \midrule
    \elimsym_{qe} & : & P^C \rightarrow \overline{\Pi} \times C^{\Pi} \times C \rightarrow C\\
    \midrule

    \elimqe{qe}{[]}{[]}{\sigma}{c} & \doteq & c\\

    \elimqe{qe}{[]}{\pi_n:\overline{\pi}}{\sigma}{c} & \doteq &
    \elimqe{qe}{[]}{\overline{\pi}}{\sigma}{\sosubst{c}{\pi_n}{\lambda n,\overline{x}.p}}\\
    \quad \text{where}\\
    \quad \quad p & = & qe(\solp{qe}{\{\pi_n\}}{\sigma}{\sigma(\pi_n)})\\

    \midrule
    \solpsym_{qe} & : & P^C \rightarrow \overline{\Pi} \times C^{\Pi} \times C \rightarrow C\\
    \midrule
    \solp{qe}{\overline{\pi}}{\sigma}{\forall n.\pi_n(n, \overline{x}) \Rightarrow c}\\
    \quad \mid \pi_n \in \overline{\pi} & \doteq &  \solp{qe}{\overline{\pi}}{\sigma}{c} \\
    \quad \mid \pi_n \notin \overline{\pi} & \doteq & \forall n.p \Rightarrow \solp{qe}{\overline{\pi}}{\sigma}{c}\\
    \quad \quad \text{where } p & =  & qe(\solp{qe}{\overline{\pi} \cup \{\pi_n\}}{\sigma}{\sigma(\pi_n)})\\

    \solp{qe}{\overline{\pi}}{\sigma}{\forall x.p \Rightarrow c} & \doteq &
    \forall x.p \Rightarrow \solp{qe}{\overline{\pi}}{\sigma}{c}\\

    \solp{qe}{\overline{\pi}}{\sigma}{c_1 \wedge c_2} & \doteq &
    \solp{qe}{\overline{\pi}}{\sigma}{c_1} \\ && \, \wedge\,  \solp{qe}{\overline{\pi}}{\sigma}{c_2}\\

    \solp{qe}{\overline{\pi}}{\sigma}{p} & \doteq & p\\
    \bottomrule
  \end{array}\]
  \caption{Eliminating $\pi$ Variables and their Side Conditions}
\end{figure}

\begin{figure}[H]
  \begin{align*}
    \safe_{qe}(c) &\doteq \smtvalid{\text{VC}} \\
    \text{where}\\
    \quad \overline{\kappa} &= \{\kappa \mid \kappa \text{ in } c\} \\
    \quad c_1 &= \poke{\emptyset}{c} \\
    \quad \overline{\pi} &= \{\kappa \mid \kappa \text{ in } c_1\} - \overline{\kappa} \\
    \quad c_2 &= \elims{\overline{\kappa}}{c_1} \\
    \quad \sigma(\pi_n) &= \defconstr{\pi_n}{c_2} \\
    \quad c_3 &= \elimqe{qe}{meaningless drivel}{\overline{\pi}}{\sigma}{c_2} \\
    \quad \text{VC} &= \elimE{qe}{c_3}{}{}{}
  \end{align*}
  
  \caption{Checking if an \ehc{} is safe}
  \label{fig:algo:safe}
\end{figure}
\section{Proofs}
\label{sec:proofs}

\begin{lemma} \label{lem:fresh-wf}
  If $t' = \fresh{\Gamma}{t}$, and $\kvars{t'} \subseteq \domain{\Delta}$, then $\isWellFormed{\Delta(\Gamma)}{\Delta(t)}$.
\end{lemma}
\begin{proof}
  Follows by a simple induction.
\end{proof}

\begin{lemma} \label{lem:fresh-erasure}
  If $t = \fresh{\Gamma}{\tau}$, then $\forgetreft{t} = \tau$.
\end{lemma}
\begin{proof}
  Follows by a simple induction.
\end{proof}

\begin{lemma} \label{lem:gamma-genimp}
  $\valid{\Delta}{\Psi}{\embed{\Gamma}(\genimp{x}{t_x}{c})}$ iff 
  $\valid{\Delta}{\Psi}{\embed{\Gamma,\fbd{x}{t_x}}(c)}$.
\end{lemma}
\begin{proof}
  Follows from the definition of $\embed{\Gamma}$ and $\genimp{x}{t_x}{c}$.
\end{proof}

\begin{theorem}[Soundness of Type Inference] \label{lem:inference-sound}
  {~}
  \begin{enumerate}
  \item If $\synth{\Gamma}{e}{t}{c}$, $\valid{\Delta}{\Psi}{\embed{\Gamma}(c)}$, and $\kvars{\Gamma} \cup \kvars{t} \subseteq \domain{\Delta}$, then $\hastype{\Delta(\Gamma)}{e}{\Delta(t)}$.

  \item If $\checking{\Gamma}{e}{t}{c}$, $\valid{\Delta}{\Psi}{\embed{\Gamma}(c)}$, and $\kvars{\Gamma} \cup \kvars{t} \subseteq \domain{\Delta}$, then $\hastype{\Delta(\Gamma)}{e}{\Delta(t)}$.

  \item If $\spinesynth{\Gamma}{t}{e}{t'}{c}$, $\valid{\Delta}{\Psi}{\embed{\Gamma}(c)}$, and $\kvars{\Gamma} \cup \kvars{t} \cup \kvars{t'} \subseteq \domain{\Delta}$, then $\appjudge{\Delta(\Gamma)}{\Delta({t})}{e}{\Delta(t')}$.

  \item If $\spinechecking{\Gamma}{t}{e}{t'}{c}$, $\valid{\Delta}{\Psi}{\embed{\Gamma}(c)}$, and $\kvars{\Gamma} \cup \kvars{t} \cup \kvars{t'} \subseteq \domain{\Delta}$, then $\appjudge{\Delta(\Gamma)}{\Delta({t})}{e}{\Delta(t')}$.

  \item If $\subtyping{\Gamma}{t_1}{t_2}{c}$, $\valid{\Delta}{\Psi}{\embed{\Gamma}(c)}$, and $\kvars{\Gamma} \subseteq \domain{\Delta}$ then $\isSubType{\Delta(\Gamma)}{\Delta(t_1)}{\Delta(t_2)}$.
  \end{enumerate}
\end{theorem}
\begin{proof}
  We proceed by mutual induction on the derivations.

\begin{enumerate}
  \item The synthesis judgments:
  \begin{itemize}
    \pfcase{$\synth{\Gamma}{\tabs{\alpha}{e}}{\tpoly{\alpha}{t}}{c}$}

    Immediate by inductive hypothesis.
    
    \pfcase{$\synth{\Gamma}{\tyapp{e}{\tau}}{\subst{t}{\alpha}{\tau}}{c}$}

    Immediate by inductive hypothesis, \lemref{lem:fresh-wf}, and \lemref{lem:fresh-erasure}.
    
    \pfcase{$\synth{\Gamma}{c}{\tc{c}}{\true}$}

    Immediate.
    
    \pfcase{$\synth{\Gamma}{x}{\reftpv{v}{b}{\subst{r}{y}{v} \wedge v = x}}{\true{}}$}

    Consider $\Delta$, $\Psi$ such that $\valid{\Delta}{\Psi}{\true}$, and $\kvars{\Gamma} \subseteq \domain{\Delta}$.
    We must show
    \[\hastype{\Delta(\Gamma)}{x}{\Delta({\reftpv{v}{b}{\subst{r}{y}{v} \wedge v = x}})}\]
    By inverting our hypothesis we have that $x:\reftpv{y}{b}{r} \in \Gamma$ and therefore that
    $x:\reftpv{y}{b}{\Delta(r)} \in \Delta(\Gamma)$.
    We can derive our goal with an application of \tSub as follows:
    \[
      \inferrule
        {\smtvalid{\cdots \forall x:b.\embed{\subst{\Delta(r)}{y}{x}} \Rightarrow \cdots \forall y:b.\embed{\Delta(r)} \Rightarrow (\embed{\Delta(r) \wedge y = x})}}
        {\isSubType{\Delta(\Gamma)}{\reftpv{y}{b}{\Delta(r)}}{\reftpv{v}{b}{\subst{\Delta(r)}{y}{v} \wedge v = x}}}
    \]
    
    \pfcase{$\synth{\Gamma}{x}{t}{\true{}}$}

    Immediate.
    
    \pfcase{$\synth{\Gamma}{\elambda{x}{\tau}{e}}{\trfun{x}{\hat{t}}{t}}{\genimp{x}{\hat{t}}{c}}$}

    Consider $\Delta$, $\Psi$ such that $\valid{\Delta}{\Psi}{\genimp{x}{\hat{t}}{\embed{\Gamma}(c)}}$, and $\kvars{\Gamma} \cup \kvars{\trfun{x}{\hat{t}}{t}} \subseteq \domain{\Delta}$.
    We must show
    \[
      \hastype{\Delta(\Gamma)}{\elambda{x}{\tau}{e}}{\Delta(\trfun{x}{\hat{t}}{t})}
    \]
    This follows by inductive hypothesis, \lemref{lem:fresh-wf}, \lemref{lem:fresh-erasure}, and \lemref{lem:gamma-genimp}.

    \pfcase{$\synth{\Gamma}{\elambda{x}{t_x}{e}}{\trfun{x}{t_x}{t}}{\genimp{x}{t_x}{c}}$}

    Immediate by inductive hypothesis and \lemref{lem:gamma-genimp}.
    
    \pfcase{$\synth{\Gamma}{\app{e_1}{e_2}}{t_2}{c_1 \wedge c_2}$}

    Consider $\Delta$, $\Psi$ such that $\valid{\Delta}{\Psi}{c_1 \wedge c_2}$, and $\kvars{\Gamma} \cup \kvars{t_2} \subseteq \domain{\Delta}$.
    We must show
    \[
      \hastype{\Delta(\Gamma)}{\app{e_1}{e_2}}{\Delta(t_2)}
    \]

    By inversion we have $\synth{\Gamma}{e_1}{t_1}{c_1}$
    and $\spinesynth{\Gamma}{t_1}{e_2}{t_2}{c_2}$.
    
    Let $\overline{\kappa} = \kvars{t_1}$.
    Now consider an arbitrary $\Delta'$ such that $\overline{\kappa} \subseteq (\Delta',\Delta)$ and $\domain{\Delta'} \cap \domain{\Delta} = \emptyset$.
    Then if $\valid{\Delta',\Delta}{\Psi}{c_1}$ and $\valid{\Delta',\Delta}{\Psi}{c_2}$ our inductive hypothesis gives us $\hastype{(\Delta',\Delta)(\Gamma)}{e_2}{(\Delta',\Delta)(t_1)}$ and $\appjudge{(\Delta',\Delta)(\Gamma)}{(\Delta',\Delta)(t_1)}{e_2}{(\Delta',\Delta)(t_2)}$.
    From this we can derive
    \[
      \hastype{(\Delta,\Delta')(\Gamma)}{\app{e_1}{e_2}}{(\Delta,\Delta')(t_2)}
    \]
    but by strengthening this is equivalent to
    \[
      \hastype{(\Delta)(\Gamma)}{\app{e_1}{e_2}}{(\Delta)(t_2)}
    \]

    It then suffices to show that there exists such a $\Delta'$.
    But we already have $\valid{\Delta}{\Psi}{\embed{\Gamma}(c_1)}$ and $\valid{\Delta}{\Psi}{\embed{\Gamma}(c_2)}$ by assumption.
    Therefore the only restriction on $\Delta'$ is that is covers the extra $\kappa$ variables in $t_1$ and doesn't intersect with $\Delta$.
    These conditions are easy to meet by making arbitrary (well-formed) choices for the $\kappa$ variables in $t_1$.
    
    \pfcase{$\synth{\Gamma}{\ilambda{x}{t_x}{e}}{\trifun{x}{t_x}{t}}{\genimp{x}{\hat{t}}{c}}$}

    Immediate by inductive hypothesis and \lemref{lem:gamma-genimp}.
    
    \pfcase{$\synth{\Gamma}{\eletin{\tb{x}{\tau}}{e_1}{e_2}}{\hat{t}'}{c_1\wedge{} (\genimp{x}{\hat{t}}{c \wedge c_2})}$}

    Consider $\Delta$, $\Psi$ such that $\valid{\Delta}{\Psi}{\embed{\Gamma}(c_1\wedge{} (\genimp{x}{\hat{t}}{c \wedge c_2}))}$, and $\kvars{\Gamma} \cup \kvars{\hat{t}'} \subseteq \domain{\Delta}$.
    We must show
    \[
      \hastype{\Delta(\Gamma)}{\eletin{\tb{x}{\tau}}{e_1}{e_2}}{\Delta(\hat{t}')}
    \]
    Inverting our assumption we have that $\checking{\Gamma}{e_1}{\hat{t}}{c_1}$, $\synth{\Gamma,\fbd{x}{\hat{t}}}{e_2}{t'}{c_2}$, $\subtyping{\Gamma,\fbd{x}{\hat{t}}}{t'}{\hat{t}'}{c}$, $\hat{t} = \fresh{\Gamma}{\tau}$, and $\hat{t}' = \fresh{\Gamma}{t'}$.
    From our inductive hypothesis we may derive $\hastype{\Delta(\Gamma)}{e_1}{\Delta(\hat{t})}$ by \lemref{lem:fresh-wf}, \lemref{lem:fresh-erasure}, and the fact that $\valid{\Delta}{\Psi}{\embed{\Gamma}(c_1)}$.

    Let $\overline{\kappa} = \kvars{t'}$.
    Now consider an arbitrary $\Delta'$ such that $\overline{\kappa} \subseteq (\Delta',\Delta)$ and $\domain{\Delta'} \cap \domain{\Delta} = \emptyset$.
    Our inductive hypothesis gives us that, if $\valid{\Delta',\Delta}{\Psi}{\embed{\Gamma,\fbd{x}{\hat{t}}}(c_2)}$ and $\valid{\Delta',\Delta}{\Psi}{\embed{\Gamma,\fbd{x}{\hat{t}}}(c)}$,
    then $\hastype{(\Delta',\Delta)(\Gamma,\fbd{x}{\hat{t}})}{e_2}{(\Delta',\Delta)(t)}$
    and $\isSubType{(\Delta',\Delta)(\Gamma,\fbd{x}{\hat{t}})}{(\Delta',\Delta)(t')}{(\Delta',\Delta)(\hat{t}')}$.
    Both validities hold by assumption and \lemref{lem:gamma-genimp}.
    We may then pick any choices for $\overline{\kappa}$ and then a use of the subtyping judgment and strengthening gives us
    \[
      \hastype{\Delta(\Gamma)}{\eletin{\tb{x}{\tau}}{e_1}{e_2}}{\Delta(\hat{t}')}
    \]

    \pfcase{$\synth{\Gamma}{\eletin{\tb{x}{t_x}}{e_1}{e_2}}{\hat{t}}{\genimp{x}{t_x}{(c_1 \wedge c_2 \wedge c_3)}}$}

    Consider $\Delta$, $\Psi$ such that $\valid{\Delta}{\Psi}{\embed{\Gamma}(\genimp{x}{t_x}{(c_1 \wedge c_2 \wedge c_3)})}$, and $\kvars{\Gamma} \cup \kvars{\hat{t}} \subseteq \domain{\Delta}$.
    We must show
    \[
      \hastype{\Delta(\Gamma)}{\eletin{\tb{x}{t_x}}{e_1}{e_2}}{\Delta(\hat{t})}
    \]
    Inverting our assumption we have that $\checking{\Gamma, \fbd{x}{t_x}}{e_1}{t_x}{c_1}$, $\synth{\Gamma,\fbd{x}{t_x}}{e_2}{t}{c_2}$, $\subtyping{\Gamma,\fbd{x}{t_x}}{t}{\hat{t}}{c}$, and $\hat{t} = \fresh{\Gamma}{t}$.
    From our inductive hypothesis we may derive $\hastype{\Delta(\Gamma,\fbd{x}{t_x})}{e_1}{\Delta(t_x)}$ as $\valid{\Delta}{\Psi}{\embed{\Gamma}(\genimp{x}{t_x}{c_1})}$ by \lemref{lem:gamma-genimp}.

    Let $\overline{\kappa} = \kvars{t}$.
    Now consider an arbitrary $\Delta'$ such that $\overline{\kappa} \subseteq (\Delta',\Delta)$ and $\domain{\Delta'} \cap \domain{\Delta} = \emptyset$.
    Our inductive hypothesis gives us that, if $\valid{\Delta',\Delta}{\Psi}{\embed{\Gamma,\fbd{x}{t_x}}(c_2)}$ and $\valid{\Delta',\Delta}{\Psi}{\embed{\Gamma,\fbd{x}{t_x}}(c)}$,
    then $\hastype{(\Delta',\Delta)(\Gamma,\fbd{x}{t_x})}{e_2}{(\Delta',\Delta)(t)}$
    and $\isSubType{(\Delta',\Delta)(\Gamma,\fbd{x}{t_x})}{(\Delta',\Delta)(t)}{(\Delta',\Delta)(\hat{t})}$.
    Both validities hold by assumption and \lemref{lem:gamma-genimp}.
    We may then pick any choices for $\overline{\kappa}$ and then a use of the subtyping judgment and strengthening gives us
    \[
      \hastype{\Delta(\Gamma)}{\eletin{\tb{x}{t_x}}{e_1}{e_2}}{\Delta(\hat{t})}
    \]

    \pfcase{$\synth{\Gamma}{\eunpack{x}{y}{e_1}{e_2}}{\hat{t}}{c_1 \wedge (\genimp{x}{t_1}{(\genimp{y}{\subst{t_2}{x'}{x}}{(c_2 \wedge c_3)})})}$}

    Consider $\Delta$, $\Psi$ such that $\valid{\Delta}{\Psi}{\embed{\Gamma}(c_1 \wedge (\genimp{x}{t_1}{(\genimp{y}{\subst{t_2}{x'}{x}}{(c_2 \wedge c_3)})}))}$, and $\kvars{\Gamma} \cup \kvars{\hat{t}} \subseteq \domain{\Delta}$.
    We must show
    \[
      \hastype{\Delta(\Gamma)}{\eunpack{x}{y}{e_1}{e_2}}{\Delta(\hat{t})}
    \]

    Inverting our assumption we have that $\synth{\Gamma}{e_1}{\tcoifun{x'}{t_1}{t_2}}{c_1}$, $\synth{\Gamma,\ebd{x}{t_1},\fbd{y}{\subst{t_2}{x'}{x}}}{e_2}{t}{c_2}$, $\subtyping{\Gamma,\ebd{x}{t_1},\fbd{y}{\subst{t_2}{x'}{x}}}{t}{\hat{t}}{c_3}$, and $\hat{t} = \fresh{\Gamma}{t}$.
    From our inductive hypothesis we may derive $\hastype{\Delta(\Gamma)}{e_1}{\Delta(\tcoifun{x'}{t_1}{t_2})}$ as $\valid{\Delta}{\Psi}{\embed{\Gamma}(c_1)}$.

    Let $\overline{\kappa} = \kvars{t}$.
    Now consider an arbitrary $\Delta'$ such that $\overline{\kappa} \subset (\Delta', \Delta)$ and $\domain{\Delta'} \cap \domain{\Delta} = \emptyset$.
    Our inductive hypothesis gives us that, if $\valid{\Delta',\Delta}{\Psi}{\embed{\Gamma,\ebd{x}{t_x},\fbd{y}{\subst{t_2}{x'}{x}}}}(c_2)$ and $\valid{\Delta',\Delta}{\Psi}{\embed{\Gamma,\fbd{x}{t_x}}(c_3)}$,
    then $\hastype{(\Delta',\Delta)(\Gamma,\ebd{x}{t_x},\fbd{y}{\subst{t_2}{x'}{x}})}{e_2}{(\Delta',\Delta)(t)}$
    and $\isSubType{(\Delta',\Delta)(\Gamma,\ebd{x}{t_x},\fbd{y}{\subst{t_2}{x'}{x}})}{(\Delta',\Delta)(t)}{(\Delta',\Delta)(\hat{t})}$.
    Both validities hold by assumption and \lemref{lem:gamma-genimp}.
    We may then pick any choices for $\overline{\kappa}$ and then a use of the subtyping judgment and strengthening gives us
    \[
      \hastype{\Delta(\Gamma)}{\eunpack{x}{y}{e_1}{e_2}}{\Delta(\hat{t})}
    \]
  \end{itemize}

\item The checking judgments.
  \begin{itemize}
    \pfcase[Subtyping]

    Our assumption is that
    \[\inferrule
      {\synth{\Gamma}{e}{t'}{c}
      \\ \subtyping{\Gamma}{t'}{t}{c'}}
      {\checking{\Gamma}{e}{t}{c \wedge c'}}\]

    Consider $\Delta$, $\Psi$ such that $\valid{\Delta}{\Psi}{\embed{\Gamma}(c \wedge c')}$, and $\kvars{\Gamma} \cup \kvars{t} \subseteq \domain{\Delta}$.
    We must show
    \[
      \hastype{\Delta(\Gamma)}{e}{\Delta(t)}
    \]
    Let $\overline{\kappa} = \kvars{t'}$.
    Now consider an arbitrary $\Delta'$ such that $\overline{\kappa} \subseteq (\Delta',\Delta)$ and $\domain{\Delta'} \cap \domain{\Delta} = \emptyset$.
    Our inductive hypothesis gives us that, if $\valid{\Delta',\Delta}{\Psi}{\embed{\Gamma}(c)}$ and $\valid{\Delta',\Delta}{\Psi}{\embed{\Gamma}(c')}$,
    then $\hastype{(\Delta',\Delta)(\Gamma)}{e}{(\Delta',\Delta)(t')}$
    and $\isSubType{(\Delta',\Delta)(\Gamma)}{(\Delta',\Delta)(t')}{(\Delta',\Delta)(t)}$.
    $\valid{\Delta}{\Psi}{\embed{\Gamma}(c)}$ and $\valid{\Delta}{\Psi}{\embed{\Gamma}(c')}$ so we pick arbitrary instantiations of $\overline{\kappa}$.
    With strengthening we can then derive
    \[
      \hastype{\Delta(\Gamma)}{e}{\Delta(t)}
    \]

    \pfcase[$\checking{\Gamma}{\app{e_1}{e_2}}{t}{c_1 \wedge c_2}$]

    Consider $\Delta$, $\Psi$ such that $\valid{\Delta}{\Psi}{c_1 \wedge c_2}$, and $\kvars{\Gamma} \cup \kvars{t} \subseteq \domain{\Delta}$.
    We must show
    \[
      \hastype{\Delta(\Gamma)}{\app{e_1}{e_2}}{\Delta(t)}
    \]

    By inversion we have $\synth{\Gamma}{e_1}{t'}{c_1}$
    and $\spinechecking{\Gamma}{t'}{e_2}{t}{c_2}$.
    
    Let $\overline{\kappa} = \kvars{t'}$.
    Now consider an arbitrary $\Delta'$ such that $\overline{\kappa} \subseteq (\Delta',\Delta)$ and $\domain{\Delta'} \cap \domain{\Delta} = \emptyset$.
    Then if $\valid{\Delta',\Delta}{\Psi}{c_1}$ and $\valid{\Delta',\Delta}{\Psi}{c_2}$ our inductive hypothesis gives us $\hastype{(\Delta',\Delta)(\Gamma)}{e_1}{(\Delta',\Delta)(t')}$ and $\appjudge{(\Delta',\Delta)(\Gamma)}{(\Delta',\Delta)(t')}{e_2}{(\Delta',\Delta)(t)}$.
    From this we can derive
    \[
      \hastype{(\Delta,\Delta')(\Gamma)}{\app{e_1}{e_2}}{(\Delta,\Delta')(t')}
    \]
    but by strengthening this is equivalent to
    \[
      \hastype{(\Delta)(\Gamma)}{\app{e_1}{e_2}}{(\Delta)(t')}
    \]

    It then suffices to show that there exists such a $\Delta'$.
    But we already have $\valid{\Delta}{\Psi}{\embed{\Gamma}(c_1)}$ and $\valid{\Delta}{\Psi}{\embed{\Gamma}(c_2)}$ by assumption.
    Therefore the only restriction on $\Delta'$ is that is covers the extra $\kappa$ variables in $t'$ and doesn't intersect with $\Delta$.
    These conditions are easy to meet by making arbitrary (well-formed) choices for the $\kappa$ variables in $t'$.

    \pfcase[$\checking{\Gamma}{\eletin{\tb{x}{\tau}}{e_1}{e_2}}{t}{c_1 \wedge{} (\genimp{x}{\hat{t}}{c_2})}$]

    Consider $\Delta$, $\Psi$ such that $\valid{\Delta}{\Psi}{\embed{\Gamma}(c_1\wedge{} (\genimp{x}{\hat{t}}{c_2}))}$, and $\kvars{\Gamma} \cup \kvars{t} \subseteq \domain{\Delta}$.
    We must show
    \[
      \hastype{\Delta(\Gamma)}{\eletin{\tb{x}{\tau}}{e_1}{e_2}}{\Delta(t)}
    \]
    Inverting our assumption we have that $\checking{\Gamma}{e_1}{\hat{t}}{c_1}$, $\checking{\Gamma,\fbd{x}{\hat{t}}}{e_2}{t}{c_2}$, and $\hat{t} = \fresh{\Gamma}{\tau}$.
    From our inductive hypothesis we may derive $\hastype{\Delta(\Gamma)}{e_1}{\Delta(\hat{t})}$ by \lemref{lem:fresh-wf}, \lemref{lem:fresh-erasure}, and the fact that $\valid{\Delta}{\Psi}{\embed{\Gamma}(c_1)}$.

    We derive $\hastype{\Delta(\Gamma,\fbd{x}{\hat{t}})}{e_2}{\Delta(t)}$ by applying our inductive hypothesis and \lemref{lem:gamma-genimp}.

    \pfcase[$\checking{\Gamma}{\eletin{\tb{x}{t_x}}{e_1}{e_2}}{t}{(\genimp{x}{t_x}{(c_1 \wedge c_2)})}$]

    Follows immediately from the inductive hypothesis and \lemref{lem:gamma-genimp}.

    \pfcase[$\checking{\Gamma}{\eunpack{x}{y}{e_1}{e_2}}{t}{c_1 \wedge (\genimp{x}{t_1}{(\genimp{y}{\subst{t_2}{x'}{x}}{c_2})})}$]

    Consider $\Delta$, $\Psi$ such that $\valid{\Delta}{\Psi}{\embed{\Gamma}(c_1 \wedge (\genimp{x}{t_1}{(\genimp{y}{\subst{t_2}{x'}{x}}{c_2})}))}$, and $\kvars{\Gamma} \cup \kvars{t} \subseteq \domain{\Delta}$.

    Inverting our assumption we have that $\synth{\Gamma}{e_1}{\tcoifun{x'}{t_1}{t_2}}{c_1}$ and $\checking{\Gamma,\ebd{x}{t_1},\fbd{y}{\subst{t_2}{x'}{x}}}{e_2}{t}{c_2}$.
    
    Let $\overline{\kappa} = \kvars{\tcoifun{x'}{t_1}{t_2}}$.
    Now consider an arbitrary $\Delta'$ such that $\overline{\kappa} \subseteq (\Delta',\Delta)$ and $\domain{\Delta'} \cap \domain{\Delta} = \emptyset$.
    Then if $\valid{\Delta',\Delta}{\Psi}{c_1}$ and $\valid{\Delta',\Delta}{\Psi}{\genimp{x}{\hat{t}}{c_2}}$ our inductive hypothesis and \lemref{lem:gamma-genimp} gives us $\hastype{(\Delta',\Delta)(\Gamma)}{e_1}{(\Delta',\Delta)(\tcoifun{x'}{t_1}{t_2})}$ and $\hastype{(\Delta',\Delta)(\Gamma,\ebd{x}{t_x},\fbd{y}{\subst{t_2}{x'}{x}})}{e_2}{(\Delta',\Delta)(t)}$.
    From this we can derive
    \[
      \hastype{(\Delta,\Delta')(\Gamma)}{\eunpack{x}{y}{e_1}{e_2}}{(\Delta,\Delta')(t)}
    \]
    but by strengthening this is equivalent to
    \[
      \hastype{(\Delta)(\Gamma)}{\eunpack{x}{y}{e_1}{e_2}}{(\Delta,\Delta')(t)}
    \]

    It then suffices to show that there exists such a $\Delta'$.
    But we already have $\valid{\Delta}{\Psi}{\embed{\Gamma}(c_1)}$ and $\valid{\Delta}{\Psi}{\embed{\Gamma,\ebd{x}{t_x},\fbd{y}{\subst{t_2}{x'}{x}}}{\genimp{x}{\hat{t}}{c_2}}}$ by assumption and \lemref{lem:gamma-genimp}.
    Therefore the only restriction on $\Delta'$ is that is covers the extra $\kappa$ variables in $\tcoifun{x'}{t_1}{t_2}$ and doesn't intersect with $\Delta$.
    These conditions are easy to meet by making arbitrary (well-formed) choices for the $\kappa$ variables in $\tcoifun{x'}{t_1}{t_2}$.

  \end{itemize}

\item The application synthesis judgments.
  \begin{itemize}
    \pfcase[$\spinesynth{\Gamma}{\trfun{x}{t_x}{t}}{y}{\subst{t}{x}{y}}{c}$]

    Follows immediately by the inductive hypothesis and a trivial subtyping judgment ($\isSubType{\Delta(\Gamma)}{\Delta({\subst{t}{x}{y}})}{\Delta({\subst{t}{x}{y}})}$).

    \pfcase[$\spinesynth{\Gamma}{\trfun{x}{t_x}{t}}{e}{\hat{t}}{c_1 \wedge c_2 \wedge (\genimp{y}{t_e}{c_3})}$]

    Follows immediately by repeated uses of the inductive hypothesis and \lemref{lem:fresh-wf}.

    \pfcase[$\spinesynth{\Gamma}{\trifun{x}{t_x}{t}}{e}{\hat{t}}{\genexists{z}{t_x}{(c \wedge c')}}$]

    By well-formedness of $\trifun{x}{t_x}{t}$, there are some $v, r, b$ such that $t_x = \reftpv{v}{b}{r}$ and therefore $\genexists{z}{t_x}{(c \wedge c')} = \exists z:b. (\subst{r}{v}{z} \wedge c \wedge c')$.

    Consider $\Delta$, $\Psi$ such that $\valid{\Delta}{\Psi}{\embed{\Gamma}(\subst{r}{v}{z} \wedge c \wedge c')}$, and $\kvars{\Gamma} \cup \kvars{t} \cup \kvars{\hat{t}} \subseteq \domain{\Delta}$.
    We must show
    \[
      \appjudge{\Delta(\Gamma)}{\Delta(\trifun{x}{t_x}{t})}{e}{\Delta(\hat{t})}
    \]
    $\Psi$ must contain some $\fobind{e'}{z}$ such that $\valid{\Delta}{\Psi}{\subst{(\subst{r}{v}{z} \wedge c \wedge c')}{z}{e'}}$.
    We take this $e'$ as our implicit instantiation.
    Then by our inductive hypothesis we get
    \[
      \inferrule
      {\appjudge{\Delta(\Gamma)}{\Delta(\subst{t}{x}{e'})}{e}{\Delta(t')}
      \\ \isSubType{\Delta(\Gamma)}{\Delta(t')}{\Delta(\hat{t})}
      \\ \hastype{\forgetimplicits{\Delta(\Gamma)}}{e}{t_x}}
      {\appjudge{\Delta(\Gamma)}{\Delta(\trifun{x}{b}{t})}{e}{\Delta(\hat{t})}}
    \]
  \end{itemize}

\item The application checking judgments.
  \begin{itemize}
    \pfcase[$\spinechecking{\Gamma}{\trfun{x}{t_x}{t}}{e}{t'}{c \wedge c'}$]

    Follows immediately from the inductive hypotheses.

    \pfcase[$\spinechecking{\Gamma}{\trfun{x}{t_x}{t}}{e}{\hat{t}}{c_1 \wedge c_2 \wedge (\genimp{y}{t_e}{c_3})}$]

    Follows immediately by repeated uses of the inductive hypothesis.

    \pfcase[$\spinechecking{\Gamma}{\trifun{x}{t_x}{t}}{e}{t'}{\genexists{z}{t_x}{c}}$]

    By well-formedness of $\trifun{x}{t_x}{t}$, there are some $v, r, b$ such that $t_x = \reftpv{v}{b}{r}$ and therefore $\genexists{z}{t_x}{c} = \exists z:b. (\subst{r}{v}{z} \wedge c)$.

    Consider $\Delta$, $\Psi$ such that $\valid{\Delta}{\Psi}{\embed{\Gamma}(\subst{r}{v}{z} \wedge c)}$, and $\kvars{\Gamma} \cup \kvars{t} \cup \kvars{\hat{t}} \subseteq \domain{\Delta}$.
    We must show
    \[
      \appjudge{\Delta(\Gamma)}{\Delta(\trifun{x}{t_x}{t})}{e}{\Delta(t')}
    \]
    $\Psi$ must contain some $\fobind{e'}{z}$ such that $\valid{\Delta}{\Psi}{\subst{(\subst{r}{v}{z} \wedge c)}{z}{e'}}$.
    We take this $e'$ as our implicit instantiation.
    Then by our inductive hypothesis we get
    \[
      \inferrule
      {\appjudge{\Delta(\Gamma)}{\Delta(\subst{t}{x}{e'})}{e}{\Delta(t')}
      \\ \isSubType{\Delta(\Gamma)}{\Delta(t')}{\Delta(t')}
      \\ \hastype{\forgetimplicits{\Delta(\Gamma)}}{e}{t_x}}
      {\appjudge{\Delta(\Gamma)}{\Delta(\trifun{x}{t_x}{t})}{e}{\Delta(t')}}
    \]
  \end{itemize}

\item The subtyping judgments.
  \begin{itemize}
    \pfcase[$\subtyping{\Gamma}{\reftpv{x}{\tau}{r}}{\reftpv{y}{\tau}{r'}}{\forall x:\tau. \embed{r} \Rightarrow \embed{\subst{r'}{y}{x}}}$]

    Immediate.

    \pfcase[$\subtyping{\Gamma}{\tpoly{\alpha}{t}}{\tpoly{\beta}{t'}}{c}$]

    Immediate by inductive hypotheses.

    \pfcase[$\subtyping{\Gamma}{\trfun{x_1}{t_1}{t_1'}}{\trfun{x_2}{t_2}{t_2'}}{c \wedge {} (\genimp{x_2}{t_2}{c'})}$]

    Immediate by inductive hypotheses and \lemref{lem:gamma-genimp}.

    \pfcase[$\subtyping{\Gamma}{\trifun{x}{t_1}{t_1'}}{t_2}{\genexists{z}{t_1}{c}}$]

    By well-formedness of $\trifun{x}{t_1}{t_1'}$, there are some $v, r, b$ such that $t_1 = \reftpv{v}{b}{r}$ and therefore $\genexists{z}{t_1}{c} = \exists z:b. (\subst{r}{v}{z} \wedge c)$.

    Consider $\Delta$, $\Psi$ such that $\valid{\Delta}{\Psi}{\embed{\Gamma}(\exists z:b. (\subst{r}{v}{z} \wedge c))}$, and $\kvars{\Gamma} \cup \kvars{t_1} \cup \kvars{t_1'} \cup \kvars{t_2} \subseteq \domain{\Delta}$.
    We must show
    \[
      \isSubType{\Delta(\Gamma)}{\Delta(\trifun{x}{t_1}{t_1'})}{\Delta(t_2)}
    \]
    $\Psi$ must contain some $\fobind{e}{z}$ such that $\valid{\Delta}{\Psi}{\subst{(\subst{r}{v}{z} \wedge c)}{z}{e}}$.
    We take this $e$ as our implicit instantiation.
    Then by our inductive hypothesis we can derive
    \[
      \inferrule
      {\isSubType{\Delta(\Gamma)}{\Delta(\subst{t_1'}{x}{e})}{\Delta(t_2)}
      \\ \hastype{\forgetimplicits{\Delta(\Gamma)}}{e}{t_1}}
      {\isSubType{\Delta(\Gamma)}{\Delta(\trifun{x}{t_1}{t_1'})}{\Delta(t_2)}}
    \]

    \pfcase[$\subtyping{\Gamma}{t_1}{\trifun{x}{t_2}{t_2'}}{\genimp{x}{t_2}{c}}$]

    Immediate by our inductive hypothesis and \lemref{lem:gamma-genimp}.

    \pfcase[$\subtyping{\Gamma}{t_1}{\tcoifun{x}{t_2}{t_2'}}{\genexists{z}{t_2}{c}}$]

    By well-formedness of $\tcoifun{x}{t_2}{t_2'}$, there are some $v, r, b$ such that $t_2 = \reftpv{v}{b}{r}$ and therefore $\genexists{z}{t_2}{c} = \exists z:b. (\subst{r}{v}{z} \wedge c)$.

    Consider $\Delta$, $\Psi$ such that $\valid{\Delta}{\Psi}{\embed{\Gamma}(\exists z:b. (\subst{r}{v}{z} \wedge c))}$, and $\kvars{\Gamma} \cup \kvars{t_1} \cup \kvars{t_2} \cup \kvars{t_2'} \subseteq \domain{\Delta}$.
    We must show
    \[
      \isSubType{\Delta(\Gamma)}{\Delta(t_1)}{\Delta(\trifun{x}{t_2}{t_2'})}
    \]
    $\Psi$ must contain some $\fobind{e}{z}$ such that $\valid{\Delta}{\Psi}{\subst{(\subst{r}{v}{z} \wedge c)}{z}{e}}$.
    We take this $e$ as our implicit instantiation.
    Then by our inductive hypothesis we can derive
    \[
      \inferrule
      {\isSubType{\Delta(\Gamma)}{\Delta(t_1)}{\Delta(\subst{t_2'}{x_1}{e})}
      \\ \hastype{\forgetimplicits{\Delta(\Gamma)}}{e}{t_2}}
    {\isSubType{\Delta(\Gamma)}{\Delta(t_1)}{\Delta(\trifun{x}{t_2}{t_2'})}}
    \]

    \pfcase[$\subtyping{\Gamma}{\trifun{x}{t_1}{t_1'}}{t_2}{\genimp{x}{t_1}{c}}$]

    Immediate by our inductive hypothesis and \lemref{lem:gamma-genimp}.
  \end{itemize}
\end{enumerate}
\end{proof}

\thmref{thm:separable}
  c is separable if there are no 
  cyclic $\kappa$s under existential binders.
\begin{proof}
  If there does not exist such a $\kappa$, then we can simply take $c_1 = \noside{c}$ and $c_2=\sides{c}$. $c_1$ is \nnf{} by the structure of $\nosidesym$ and $c_2$ is acyclic because there doesn't exist a cyclic $\kappa$ under an  $\exists$
\end{proof}

\begin{theorem}\label{thm:pokethm}
  $\Delta \vDash c$ iff $\Delta' \vDash \poke{\vars}{c}$
\end{theorem}
\begin{proof}
For brevity and clarity, here we write our two (parallel) contexts $\Delta;\Gamma$ as a combined context $\Delta$.
  We proceed by induction on $\Delta$ in the forwards direction and $\Delta'$ in the reverse direction.
  \begin{itemize}
    \item Base case: $\bullet \vDash H$, iff $\poke{\vars}{H} = H$, iff $\bullet \vDash \poke{\vars}{H}$
    \item Iff $\Delta, R/\kappa \vDash H$,
      \begin{align*}
        \Delta & \vDash c[R/\kappa]\\
        \Delta' & \vDash \poke{\vars}{H[R/\kappa]} & \text{by IH}\\
        \Delta' & \vDash \poke{\vars}{H} [R/\kappa] & \text{by Lemma \ref{alpha}}\\
        \Delta' , R/\kappa&\vDash \poke{\vars}{H}
      \end{align*}
    \item Iff $\Delta, v/n \vDash H$,
      \begin{align*}
        \Delta & \vDash c\{v/n\}\\
        \Delta' & \vDash \poke{\vars}{H\{v/n\}} & \text{by IH}\\
        \Delta' & \vDash \poke{\vars}{H} [\lambda n . v = n/\pvar n] & \text{by Lemma \ref{beta}}\\
        \Delta' , \lambda n . v = n/\pvar n&\vDash \poke{\vars}{H}
      \end{align*}
\end{itemize}
\end{proof}

\begin{lemma} \label{alpha}
  $\poke{\vars}{H[R/\kappa]} = \poke{\vars}{H} [R/\kappa]$
\end{lemma}

\begin{proof}
  We proceed by induction on $H$:
  \begin{itemize}
    \item $e$ \begin{align*}\poke{\vars}{e}[R/\kappa] \\&= e \\&= e [R/\kappa] \\&= \poke{\vars}{e[R/\kappa]}\end{align*}
    \item $H_1\wedge H_2$ \begin{align*}\poke{\vars}{ H_1\wedge H_2 }[R/\kappa] \\&= \poke{\vars}{H_1}[R/\kappa]\wedge\poke{\vars}{H_2}[R/\kappa]\\&= \poke{\vars}{H_1[R/\kappa]}\wedge\poke{\vars}{H_2[R/\kappa]}\\&=\poke{\vars}{H_1[R/\kappa]\wedge H_2[R/\kappa]} \\&= \poke{\vars}{(H_1\wedge H_2)[R/\kappa]}\end{align*}
    \item $\forall x:t.H$ \begin{align*}\poke{\vars}{\forall x:t.H}[R/\kappa]\\&=\forall x:t.\poke{\vars}{H}[R/\kappa]\\&=\forall x:t.\poke{\vars}{H[R/\kappa]}\\&=\poke{\vars}{(\forall x:t.H)[R/\kappa]}\end{align*}
    \item $\exists x:t.H$ \begin{align*}\poke{\vars}{\exists x:t.H}[R/\kappa]\\&=\exists x:t.\poke{\vars}{H}[R/\kappa]\\&=\exists x:t.\poke{\vars}{H[R/\kappa]}\\&=\poke{\vars}{(\exists x:t.H)[R/\kappa]}\end{align*}
    \item $\kappa'(v),\kappa'\not=\kappa$ \begin{align*}\poke{\vars}{\kappa'(v)}[R/\kappa] \\&= \kappa'(v) \\&= \poke{\vars}{\kappa'(v)[R/\kappa]}\end{align*}
    \item $\kappa(v)$ \begin{align*}\poke{\vars}{\kappa(v)}[R/\kappa] \\&= \kappa(v)[R/\kappa] \\&= R(v) \\&= \kappa(v)[R/\kappa] \\&= \poke{\vars}{\kappa(v)[R/\kappa]}\end{align*}
  \end{itemize}
\end{proof}

\begin{lemma} \label{beta}
  $\Delta; \Gamma \vDash \poke{\vars}{H\{v/n\}}$ iff
  $\Delta; \Gamma\vDash \poke{\vars}{H} [\lambda n . v = n/\pvar{n}]$
\end{lemma}
\begin{proof}
  We proceed by induction on $H$:
  \begin{itemize}
    \item $e$ \begin{align*}\poke{\vars}{e} [\lambda n . v = n/\pvar n] \\&\iff e \\&\iff \poke{\vars}{e\{v/n\}}\end{align*}
    \item $H_1\wedge H_2$ \begin{align*}\poke{\vars}{ H_1\wedge H_2 }[\lambda n . v = n/\pvar n]
\\&\iff \poke{\vars}{H_1}[\lambda n . v = n/\pvar n]\wedge\poke{\vars}{H_2}[\lambda n . v = n/\pvar n]
\\&\iff \poke{\vars}{H_1\{v/n\}}\wedge\poke{\vars}{H_2\{v/n\}}
\\&\iff\poke{\vars}{H_1\{v/n\}\wedge H_2\{v/n\}} 
\\&\iff \poke{\vars}{(H_1\wedge H_2)\{v/n\}}\end{align*}
    \item $\forall x:t.H$ \begin{align*}\poke{\vars}{\forall x:t.H}[\lambda n . v = n/\pvar n]
\\&\iff\forall x:t.\poke{\vars}{H}[\lambda n . v = n/\pvar n]
\\&\iff\forall x:t.\poke{\vars}{H\{v/n\}}
\\&\iff\poke{\vars}{(\forall x:t.H)\{v/n\}}\end{align*}
    \item $\exists n:t.H$ \begin{align*}\poke{\vars}{\exists n:t.H}[\lambda n . v = n/\pvar n]
\\&\iff(\forall n:t.\pvaro{x}{x}\Rightarrow\poke{\vars}{H})[\lambda n . v = n/\pvar n]
\\&\iff\forall n:t. \left( v = n \right) \Rightarrow(\poke{\vars}{H}[\lambda n . v = n/\pvar n])
\\&\iff\forall n:t. \left( v = n \right) \Rightarrow\poke{\vars}{H\{v/n\}}
\\&\iff \poke{\vars}{H\{v/n\}}
\\&\iff\poke{\vars}{(\exists n:t.H)\{v/n\}}\end{align*}
      \item $\exists x:t.H$ \begin{align*}\poke{\vars}{\exists x:t.H}[\lambda n . v = n/\pvar n]
\\&\iff\forall x:t.\pvaro{x}{x}\Rightarrow\poke{\vars}{H}[\lambda n . v = n/\pvar n]
\\&\iff\forall x:t.\pvaro{x}{x}\Rightarrow\poke{\vars}{H\{v/n\}}
\\&\iff\poke{\vars}{(\exists x:t.H)\{v/n\}}\end{align*}
    \item $\kappa(v)$ \begin{align*}\poke{\vars}{\kappa(v)}\{v/n\} \\&\iff \kappa(v) \\&\iff \poke{\vars}{\kappa(v)\{v/n\}}\end{align*}
  \end{itemize}
\end{proof}

\begin{theorem}
  for $\overline{\kappa}$ acyclic in c \textsc{nnf}
  \( \valid{\Delta}{\bullet}{c} \text{ iff } \valid{\Delta, {\solk{c}{\kappa}}/{\kappa}}{\bullet}{c} \)
  \label{strongsol}
\end{theorem}
\begin{proof}
  Since $\nosidesym$ case is an \nnf{} Horn clause, this follows directly from Theorem 5.13 of COsman and Jhala~\cite{cosman_local_2017} with $\hat{K}$ set to all the predicate variables not in $\overline{\kappa}$.
\end{proof}

\begin{lemma}
  \[
    \forall \Delta, \Gamma .\quad
    \Delta;\Gamma \vDash \elimsol{\overline{\kappa}}{\h c}
  \text{ iff } \Delta,\solk {c}{\overline{\kappa}}/\overline{\kappa};\Gamma\vDash \h c \]
\end{lemma}

\begin{proof}
  After exploiting \thmref{strongsol} to replace each head instance of $\kappa$ with $\true$---the only
  potential remaining occurrences of $\kappa$ are in the guard position of a universal
  quantification: $\forall x:\tau.\applykvar{\kappa}{\overline{y}} \Rightarrow c$, which under the substitution
$\Delta,\solk {c}{\overline{\kappa}}/\overline{\kappa}$ becomes 
$\forall x:\tau.\applykvar{\solk{c}{\kappa}}{\overline{y}} \Rightarrow c$. But the existential
  quantification inside $\solksym$ with the disjunction in guard position are equivalent to
  universal quantification and conjunction in the head of the clause by deMorgan's Law:
  \begin{align*}
    \forall x.(\exists y. P\wedge Q ) \Rightarrow R&\\
    \iff &\forall x. \top \Rightarrow (\forall y. (P\wedge Q \Rightarrow R))\\
    \iff &\forall y. (P\wedge Q \Rightarrow (\forall x. \top \Rightarrow R))\\
    \iff &\forall y. P \Rightarrow \forall x. Q \Rightarrow R\\
  \end{align*}
\end{proof}

\lemref{lem:pi-no-mind}
  If $\pi$ is a predicate variable inserted by $\pokesym$, then
  $\forall x.\pi(x) \Rightarrow c$ is equisatisfiable with $\forall x.\pi(x)\Rightarrow c[\lambda \_ . \true/\pi]$
\begin{proof}
  Because $\pi$ is only ever applied to $x$, $\pi$ only appears as $\pi(x)$. Now, $\pi(x)$ may appear positively or negatively in $c$. If it appears positively, then it is trivial, because it appears as an obligation of itself. If it appears negatively, it is redundant, because we already have it in our context.
\end{proof}

\thmref{thm:elimqe-sound}
  If $c' = \poke{\emptyset}{c}$, 
     $\overline{\pi}$ is the set of all Skolem variables in $c'$, 
     $c'$ has no other predicate variables, 
     $\sigma(\pi_n) = \defconstr{\pi_n}{c'}$ for all $\pi_n \in \overline{\pi}$, and
     $\vDash \elimqe{qe}{}{\overline{\pi}}{\sigma}{c'}$ 
  then $\vDash c$.
\begin{proof}
  ${\poke{\vars}{c}}{}$ is equisatisfiable with $c$ by the theorem above, so we just need to show that
  \[\vDash \sosubst{\poke{\vars}{c}}{\pi_n}{\lambda n.qe(\solp{qe}{\{\pi_n\}}{\sigma}{\sigma(\pi_n)})} \implies \vDash \poke{\vars}{c}\]
  Now,
  \[ \exists \Delta,\Gamma.\, \Delta;\Gamma\vDash \sosubst{\poke{\vars}{c}}{\pi_n}{\lambda n.qe(\solp{qe}{\{\pi_n\}}{\sigma}{\sigma(\pi_n)})}\]
  so
  \[ \Delta,{\lambda n.qe(\solp{qe}{\{\pi_n\}}{\sigma}{\sigma(\pi_n)})}/\pi_n;\Gamma\vDash {\poke{\vars}{c}}\]
  so
  \[ \vDash {\poke{\vars}{c}}\]
\end{proof}

\thmref{thm:safe-complete}
  Let $c$ be an acyclic constraint. 
  \begin{CompactItemize}
  \item If $qe$ is a strengthening then $\safe(c)$ implies $\vDash c$.
  \item If $qe$ is a strengthening and weakening then $\safe(c)$ iff $\vDash c$.
  \end{CompactItemize}
\begin{proof}
  We proceed to prove the second claim of this theorem by induction on $\h c$.
\begin{itemize}
  \item $c = p$
    \[
      {\h c} = p = \sosubst{\elimk{\overline{\kappa}}{\poke{\vars}{c}}{}}{\pi_n}{\lambda n.qe(\solp{qe}{\{\pi_n\}}{\sigma}{\sigma(\pi_n)})}\]
   \item $c = c_1 \wedge c_2$
    \begin{align*}
      &\sosubst{\elimk{\overline{\kappa}}{\poke{\vars}{c}}{}}{\pi_n}{\lambda n.qe(\solp{qe}{\{\pi_n\}}{\sigma}{\sigma(\pi_n)})}\\
      &=
      \sosubst{\elimk{\overline{\kappa}}{\poke{\vars}{c_1}\wedge\poke{\vars}{c_2}}{}}{\pi_n}{\lambda n.qe(\solp{qe}{\{\pi_n\}}{\sigma}{\sigma(\pi_n)})}\\
      &=
      \sosubst{\elimk{\overline{\kappa}}{\poke{\vars}{c_1}}{}}{\pi_n}{\lambda n.qe(\solp{qe}{\{\pi_n\}}{\sigma}{\sigma(\pi_n)})}\\
      &\quad \wedge
      \sosubst{\elimk{\overline{\kappa}}{\poke{\vars}{c_2}}{}}{\pi_n}{\lambda n.qe(\solp{qe}{\{\pi_n\}}{\sigma}{\sigma(\pi_n)})}\\
    \end{align*}
     so the goal follows from the induction hypothesis
   \item  $c = \forall x :\h\tau . p \implies c'$
     so $\forall v . p ,$ we have \( \Delta; \Gamma, v/x \vDash c\).
     applying the induction hypothesis, we have that $\forall v. p$, $\Delta;\Gamma, v/x \vDash 
  \sosubst{\elimk{\overline{\kappa}}{\poke{\vars}{c'}}{}}{\pi_n}{\lambda n.qe(\solp{qe}{\{\pi_n\}}{\sigma}{\sigma(\pi_n)})}$
      and therefore
      $$\Delta, \Gamma \vDash \forall x : \h\tau . p \implies 
  \sosubst{\elimk{\overline{\kappa}}{\poke{\vars}{c'}}{}}{\pi_n}{\lambda n.qe(\solp{qe}{\{\pi_n\}}{\sigma}{\sigma(\pi_n)})}$$
    \item $\exists n . c'$
      We know that $\poke{\vars}{c'} = \noside{\poke{\vars}{c'}}\wedge\sides{\poke{\vars}{c'}}$, so we can prove the \nosidesym{} branch (which is a tautology, since \solpsym{} here returns $c'$)
      \[
  \sosubst{\elimk{\overline{\kappa}}{\poke{\vars}{c'}}{}}{\pi_n}{\lambda n.qe(\solp{qe}{\{\pi_n\}}{\sigma}{\sigma(\pi_n)})}
         \forall n . (c' \implies c')
\]
      and then the \sidesym{} branch:
      \[
  \sosubst{\elimk{\overline{\kappa}}{\poke{\vars}{c'}}{}}{\pi_n}{\lambda n.qe(\solp{qe}{\{\pi_n\}}{\sigma}{\sigma(\pi_n)})}
       = \exists n . c'  = c\]
\end{itemize}

If $qe$ is simply a strengthening, then the forward direction of this proof still holds, so the first claim of the theorem holds as well.
\end{proof}

\section{Higher-Order Binders}
Although we enable automatic higher-order reasoning without quantifying over higher-order terms, we would still like to be able to operate in conjunction with methods that do encourage and use undecidable explicit reasoning about higher-order functions. In this appendix, we breifly describe how our system operates in this regime. We allow arbitary functions to be represented using an G\"odel Encoding as described below, and we make use of lambdas for metalinguistic convenience to describe the semantics of predicate variables.

We do this by adding to our logic an explicit application form \happ{\h e}{\h e}{\h\tau},

\[
\emphbf{Expressions}       \qquad \h \expr  ::=    \h\vconst
                      \spmid \h\evar
                      \spmid \app{\h\vconst}{\overline{\h{\expr}}}
                      \spmid \h{\app{\text{app}_\utype}{\app \expr\, \expr}}
                      \]
And the following rules to the semantics of our refinement logic: We give a denotation for appplying well-typed terms:
  $$\inference{\s{e_1}:\s\tau\rightarrow \tau'}{\embed{\app{\s{e_1}}{\s{e_2}}} = \app{\h{\text{app}_{\h{}\embed{\s{\tau\rightarrow\tau'}}}}}{\app{\lceil \s{e_1}\rceil}{\embed{\s{e_2}}}}}$$
And the following rule governing well-formedness of these application forms:
\[
\inference
{\lhastype{\Gamma}{\h {e_1}}{\h\tint}
&&  \lhastype{\Gamma}{\h {e_2}}{\h {\tau_x}}
&& \h{\tau_x} \not= \h\tprop
   }
   {\lhastype{\Gamma}{\h{\text{app}}_{\tfun{\h{\utype_x}}{\h\utype}}{\app{\h {e_1}}{\h {e_2}}}}{\h \tau}}
   [T-APP]
\]

\paragraph{G\"odel Encodings} We represent values of function types using their G\"odel encodings. This allows us to represent higher-order functions within our first-order logic. This effectively amounts to numbering each function that is created, but formulating this in terms of G\"odel encodings allows us to keep track of functions locally while separating their use and definition in the constraint language, which has no knowledge of the encoding.

\begin{definition}
  The quotation operator $\lceil - \rceil_{\h\tau}$ assigns an expression an
  G\"odel numbering at that expression's \textit{unrefined} type $\h\tau$. It's
  a bijection between that type and the integers (or the integers modulo the
  cardinality of the type, $|\tau|$, for finite types), with inverse $\lfloor - \rfloor_{\h\tau}$
\end{definition}

This way, we can quantify over higher-order functions, but for our purposes, the meaning of
those higher-order functions can just be conservatively approximated using congruence, while still allowing
techniques such as Refinement Reflection\cite{vazou_refinement_2018} to work over our constraints. See that paper
for further details on this encoding.

For example: @(f x) y@ becomes $\text{let } i = \lfloor f x \rfloor \text{ in } \text{app}_\tau i\, y$.

\end{document}